\documentclass[a4paper,USenglish,cleveref,numberwithinsect]{lipics-v2021}
\pdfoutput=1
\nolinenumbers
\hideLIPIcs
\usepackage[ruled]{algorithm2e}
\usepackage[utf8]{inputenc}
\usepackage[T5,T1]{fontenc}
\usepackage{colonequals}
\usepackage{stmaryrd}
\usepackage{tabularx}
\usepackage{booktabs}
\usepackage{microtype}
\usepackage{bm}
\usepackage[bibliography=common]{apxproof}
\usepackage{tikz}
\usepackage{mathabx} %
\usetikzlibrary{arrows,decorations.pathreplacing,calligraphy,automata} %

\usepackage{relsize}
\usepackage{stackengine}
\stackMath
\usepackage{nameref}
\definecolor{forestgreen}{rgb}{0.13, 0.55, 0.13}
\hypersetup{
    colorlinks,
    linkcolor={red!50!black},
    citecolor={blue!50!black},
    urlcolor={blue!30!black}
}

\author{Antoine Amarilli}{LTCI, Télécom Paris, Institut Polytechnique de Paris,
France \and \url{https://a3nm.net/}}{a3nm@a3nm.net}{https://orcid.org/0000-0002-7977-4441}{Partially supported by the ANR project EQUUS
ANR-19-CE48-0019 and by the Deutsche Forschungsgemeinschaft (DFG, German Research Foundation) – 431183758.}
\author{Mikaël Monet}{Univ. Lille, Inria, CNRS, Centrale Lille, UMR 9189
CRIStAL, F-59000 Lille, France \and
\url{https://mikael-monet.net/}}{mikael.monet@inria.fr}{https://orcid.org/0000-0002-6158-4607}{}
\authorrunning{A.\ Amarilli and M.\ Monet}
\Copyright{Antoine Amarilli and Mikaël Monet}


\relatedversiondetails[cite={amarilli2023enumerating}]{Full version}{https://arxiv.org/abs/2209.14878}

\EventEditors{Petra Berenbrink, Mamadou Moustapha Kant\'{e}, Patricia Bouyer, and Anuj Dawar}
\EventNoEds{4}
\EventLongTitle{40th International Symposium on Theoretical Aspects of Computer Science (STACS 2023)}
\EventShortTitle{STACS 2023}
\EventAcronym{STACS}
\EventYear{2023}
\EventDate{March 7--9, 2023}
\EventLocation{Hamburg, Germany}
\EventLogo{}
\SeriesVolume{254}
\ArticleNo{31}

\newcommand{\NN}{\mathbb{N}}

\newcommand{\lev}{\mathrm{Lev}}

\newcommand{\stratum}{\mathrm{strat}}
\newcommand{\calC}{\mathcal{C}}
\newcommand\outputt{\textsf{Output}}
\newcommand\explore{\textsf{explore}}
\newcommand\enum{\textsf{enumerate}}
\newcommand\rout{\textsf{root}}

\renewcommand{\L}{\mathrm{L}}
\newcommand{\NL}{\mathrm{NL}}
\newcommand{\pp}{\mathrm{pp}}

\newcommand{\pushL}{\mathrm{pushL}}
\newcommand{\pushR}{\mathrm{pushR}}
\newcommand{\popL}{\mathrm{popL}}
\newcommand{\popR}{\mathrm{popR}}

\newtheorem{result}{Result}
\newtheoremrep{fact}[theorem]{Fact}
\newtheoremrep{theorem}{Theorem}
\newtheoremrep{claim}[theorem]{Claim}
\newtheoremrep{proposition}[theorem]{Proposition}
\newtheoremrep{lemma}[theorem]{Lemma}
\counterwithin{theorem}{section}

\acknowledgements{We thank Florent Capelli and Charles Paperman for their
insights during initial discussions about this problem.
We thank user pcpthm on the Theoretical Computer Science Stack Exchange forum for giving the argument for
Proposition~\ref{prp:nphard} in~\cite{cstheory}.
We thank Jeffrey Shallit for pointing us to related work.
We thank Torsten Mütze and Arturo Merino for other helpful pointers.
We thank the anonymous reviewers for their valuable feedback.
Finally, we are grateful to Louis Jachiet 
and
{\fontencoding{T5}\selectfont{}Lê Thành Dũng (Tito) Nguyễn}
for
feedback on the draft.}

\title{Enumerating Regular Languages with Bounded Delay} %

\begin{document}

\maketitle

\begin{abstract}

  We study the task, for a given language~$L$, of enumerating the (generally
infinite) sequence of its words, without repetitions, while bounding the
\emph{delay} between two consecutive words. To allow for delay bounds that do not depend on the current word length,
we assume a model where we produce each word by editing the preceding word with
a small edit script, rather than writing out the word from scratch. In
particular, this witnesses that the language is \emph{orderable}, i.e., we can
write its words as an infinite sequence such that the Levenshtein edit distance
between any two consecutive words is bounded by a value that depends only on the language. For instance,
$(a+b)^*$ is orderable (with a variant of the Gray code), but $a^* + b^*$ is
not.

We characterize which regular languages are enumerable in this sense, and
show that this can be decided in PTIME in an input deterministic finite
automaton (DFA) for the language. In fact, we show that, given a DFA~$A$,
we can compute in PTIME automata $A_1, \ldots, A_t$ such that $\L(A)$ is partitioned
as $\L(A_1) \sqcup \ldots \sqcup \L(A_t)$ and every $\L(A_i)$ is orderable in
this sense. Further, we show that the value of~$t$ obtained is optimal, i.e., we cannot partition~$\L(A)$
into less than~$t$ orderable languages.

In the case where $\L(A)$ is orderable (i.e.,~$t=1$), we show that the ordering can be
produced by a bounded-delay algorithm: specifically, the algorithm runs in a
suitable pointer machine model, and produces a sequence of bounded-length edit
scripts to visit the words of~$\L(A)$ without repetitions, with bounded delay -- exponential in~$|A|$ --
between each script. In fact, we show that we can achieve this while only
allowing the edit operations \emph{push} and \emph{pop} at the beginning and
end of the word, which implies that the word can in fact be maintained in
a double-ended queue.

By contrast, when fixing the distance bound~$d$ between consecutive words
and the number of classes of the partition, it is NP-hard in the input DFA~$A$
to decide if $\L(A)$ is orderable in this sense, already for finite languages. 

Last, we study the model where push-pop edits are only allowed at the end of
the word, corresponding to a case where the word is maintained on a stack.
We show that these operations are strictly weaker and that the \emph{slender
languages} are precisely those that can be partitioned into finitely many
languages that are orderable in this sense. For the slender languages, we can again
characterize the minimal number of languages in the partition, and achieve
bounded-delay enumeration.

\end{abstract}

\vfill
\pagebreak

\section{Introduction}
\label{sec:intro}
\emph{Enumeration algorithms}~\cite{strozecki2019enumeration,Wasa16}
are a way to study the complexity of problems beyond decision or function problems,
where we must produce a large number of outputs without repetitions.
In such algorithms, 
the goal is usually to minimize the worst-case \emph{delay} between any two consecutive outputs. 
The best possible bound is to make the delay
\emph{constant}, i.e., independent from the size of the input. This is
the case, for example, when enumerating the results 
of acyclic free-connex
conjunctive queries~\cite{bagan2007acyclic} or 
of MSO queries over
trees~\cite{bagan2006mso,kazana2013enumeration}.

Unfortunately, constant-delay is an unrealistic requirement when
the objects to enumerate can have unbounded size,
simply because of the time needed to write them out.
Faced by this problem, one option is to neglect this part of the running time,
e.g., following Ruskey's ``Do not count the output
principle''~\cite[p.~8]{ruskey2003combinatorial}.
In this work, we address this challenge in a different way:
we study enumeration 
where each new object is not written from scratch but produced by 
\emph{editing the previous object},
by a small sequence of edit operations called an \emph{edit script}.
This further allows us to study the enumeration of 
\emph{infinite} collections of objects,
with an algorithm that runs indefinitely
and ensures that each object is produced
after some finite number of steps, and exactly once.
The size of the edit scripts must be \emph{bounded}, i.e., it only depends on
the collection of objects to enumerate, but not on the size of the current
object.
The algorithm thus outputs
an infinite series of edit scripts
such that applying them successively
yields the infinite collection of all objects. 
In particular, the algorithm witnesses that the collection admits a so-called
\emph{ordering}: it can be 
ordered
as an infinite sequence with a bound on the \emph{edit distance} between any two consecutive objects, namely, the number of edit operations.

In this paper, we study enumeration for regular languages in this sense,
with the Levenshtein edit distance and variants thereof.
One first question is to determine
if a given regular language~$L$ admits an
ordering, i.e., 
can we order its words such that the Levenshtein
distance of any two consecutive words only depends on~$L$ and not on the word lengths?
For instance, the language $a^*$ is easily orderable in this sense. The language $a^* b^*$ is orderable, e.g., following any Hamiltonian path on the infinite $\NN \times \NN$ grid. More interestingly, the language $(a+b)^*$ is orderable, for instance by considering words by increasing length and using a \emph{Gray code}~\cite{mutze2022combinatorial}, which enumerates all $n$-bit words by changing only one bit at each step.
More complex languages such as $a (a + bc)^* + b(cb)^* dd d^*$ can also be shown to be orderable (as our results will imply). 
However, one can see that some languages are not orderable, e.g., $a^* + b^*$. We can nevertheless generalize orderability by allowing multiple ``threads'': then we can partition $a^* + b^*$ as $a^*$ and $b^*$, both of which are orderable. This leads to several questions: Can we characterize the orderable regular languages? Can every regular language be partitioned as a finite union of orderable languages? And does this lead to a (possibly multi-threaded) enumeration algorithm with bounded delay (i.e., depending only on the language but not on the current word length)?

\subparagraph*{Contributions.}
The present paper gives an affirmative answer to these questions.
Specifically, we show that, given a DFA~$A$, we can decide in PTIME if $\L(A)$ is orderable. If it is not, we can compute in PTIME DFAs $A_1, \ldots, A_t$ partitioning the language as $\L(A) = \L(A_1) \sqcup \ldots \sqcup \L(A_t)$ such that each~$\L(A_i)$ is orderable; and we show that the $t$ given in this construction is optimal, i.e., no smaller such partition exists. If the language \emph{is} orderable (i.e., if~$t=1$), we show in fact that the same holds for a much more restricted notion of distance, the \emph{push-pop distance}, which only allows edit operations at the beginning and end of the word.
The reason we are interested in this restricted edit distance
is that edit
scripts featuring push and pop can be easily applied in constant-time
to a word represented in a double-ended queue; by contrast, Levenshtein edit
operations are more difficult to implement, because they refer to integer
word positions that change whenever
characters are inserted or deleted.\footnote{There is, in fact, an
$\Omega(\log |w| / \log \log |w|)$ lower bound on the complexity of applying Levenshtein edit
operations and querying which letter occurs at a given position: 
crucially, this bound depends on the
size of the word.
See \url{https://cstheory.stackexchange.com/q/46746} for details.
This is in contrast to the application of push-pop-right edit
operations, which can be performed in constant time (independent from the word
length) when the word is stored in a double-ended queue.}

And indeed, this result on the push-pop distance then allows us to design a bounded-delay algorithm for~$\L(A)$, which produces a sequence of bounded edit scripts of push or pop operations that enumerates~$\L(A)$.
The length of the edit scripts is polynomial in~$|A|$ and the delay
of our algorithm is exponential in~$|A|$, but crucially it remains
bounded throughout the (generally infinite) execution of the
algorithm, and does not depend on the size of the words that are
achieved. Formally, we show:

\begin{result}
  \label{res:main}
  Given a DFA~$A$,
  one can compute in PTIME automata $A_1, \ldots, A_t$ for
  some~$t \leq |A|$ such that $\L(A)$ is the disjoint union of the $\L(A_i)$,
  and we can enumerate each
  $\L(A_i)$ with bounded delay for the push-pop distance with distance bound
  $48|A|^2$ and exponential delay in $|A|$.
  Further, $\L(A)$ has no partition of cardinality $t-1$ into orderable languages, even for the Levenshtein distance.
\end{result}

Thus, we show that orderability and enumerability, for the push-pop or Levenshtein edit distance, are in fact all 
logically equivalent on regular languages, and we characterize them (and find the optimal partition cardinality) in PTIME.
By contrast, 
as was pointed out in~\cite{cstheory}, testing orderability for a fixed distance $d$ is NP-hard in the input DFA, even for finite languages. 

Last, we study the \emph{push-pop-right distance}, which only
allows edits at the end of the word. The motivation for studying
this distance is that it corresponds to  enumeration algorithms in
which the word is maintained on a stack. We show that, among the
regular languages, the \emph{slender languages}~\cite{mpri} are
then precisely those that can be partitioned into finitely many
orderable languages, and that these languages are themselves enumerable. Further, 
the optimal cardinality of the partition can again be computed in PTIME:

\begin{result}
\label{res:slender}
  Given a DFA~$A$, then $\L(A)$
  is partitionable into finitely many 
  orderable languages for the push-pop-right distance if and only if
  $\L(A)$ is slender (which we can
  test in PTIME in~$A$).
  Further, 
  in this case, we can compute in PTIME the smallest partition cardinality,
  and each language in the partition
  is enumerable with bounded delay with distance bound $2|A|$ and linear delay
  in~$|A|$.
\end{result}

In terms of proof techniques, our PTIME characterization of Result~\ref{res:main} relies on a notion of \emph{interchangeability} of automaton states, defined via paths between states and via 
states having common loops.
We then show orderability by establishing \emph{stratum-connectivity}, i.e., for
any \emph{stratum} of words of the language within some length interval, there
are finite sequences obeying the distance bound that connect any two words in
that stratum.
We show stratum-connectivity by pumping and de-pumping loops close to the word
endpoints. We then deduce an ordering from this by adapting a standard
technique~\cite{uno2003two} of visiting a spanning tree and enumerating even and
odd levels in alternation (see also~\cite{sekanina1960ordering,karaganis1968cube}). The bounded-delay enumeration algorithm then proceeds by iteratively enlarging a custom data structure called a \emph{word DAG}, where the construction of the structure for a stratum is amortized by enumerating the edit scripts to achieve the words of the previous stratum.

\subparagraph*{Related work.}
As we explained, enumeration has been extensively studied for many structures~\cite{Wasa16}.
For regular languages specifically, some authors have studied the problem of
enumerating their words in \emph{radix order}
\cite{makinen1997lexicographic,ackerman2009three,ackerman2009efficient,fleischer2021recognizing}.
For instance, the authors of \cite{ackerman2009three,ackerman2009efficient} provide an algorithm
that enumerates all words of a regular language in that order, with a delay
of~$O(|w|)$ for~$w$ the next word to enumerate. Thus, this delay is not bounded, and
the requirement to enumerate in radix order makes it challenging to guarantee a
bounded distance between consecutive words (either in the Levenshtein or
push-pop distance), which is necessary for bounded-delay enumeration in our
model. Indeed, our results show that not all regular languages are orderable in
our sense, whereas their linear-delay techniques apply to all regular
languages.

We have explained that enumeration for $(a+b)^*$ relates to Gray codes, of which there exist several variants~\cite{mutze2022combinatorial}. Some variants, e.g., the so-called \emph{middle levels problem}~\cite{mutze2016proof}, aim at enumerating binary words of a restricted form; but these languages are typically finite (i.e., words of length~$n$), and their generalization is typically not regular. While Gray codes typically allow arbitrary substitutions, one work has studied a variant that only allows restricted operations on the endpoints~\cite{feldmann1996shuffle}, implying the push-pop orderability of the specific language $(a+b)^*$.

Independently, some enumeration problems on automata have been studied recently
in the database theory literature, in particular for \emph{document
spanners}~\cite{fagin2015document}, which can be defined by finite automata with
capture variables. It was recently
shown~\cite{florenzano2018constant,amarilli2019constant} that we can enumerate
in constant delay all possible assignments of the capture variables of a fixed
spanner on an input word.
In these works, the delay is constant in \emph{data complexity}, which means that it only
depends on the (fixed) automaton, and does not depend on the word; this matches what
we call \emph{bounded delay} in our work (where there is no input word and the automaton is given as input). However, our results do not follow from these works, which focus on the
enumeration of results of constant size. Some works allow second-order variables and results of non-constant size~\cite{amarilli2019enumeration} but the delay would then be linear in each output, hence unbounded.

\subparagraph*{Paper structure.}
We give preliminaries in Section~\ref{sec:prelims}. 
In Section~\ref{sec:automaton} we present our PTIME construction of a partition of a regular language into finitely many orderable languages, and prove that the cardinality of
the obtained partition is minimal for orderability.
We then show in Section~\ref{sec:upper} that each term of the union is
orderable, and then that it is enumerable in Section~\ref{sec:worddag}. We
present the NP-hardness result on testing orderability for a fixed distance and our
results on push-pop-right operations in Section~\ref{sec:extensions}. We conclude and
mention some open problems in Section~\ref{sec:conc}.
Due to space constraints, we mostly present the general structure of
the proofs and give the main ideas; detailed proofs of all statements can be found in the 
appendix.

\section{Preliminaries}
\label{sec:prelims}
We fix a finite non-empty \emph{alphabet}~$\Sigma$ of \emph{letters}.
A \emph{word} is a finite sequence $w = a_1 \cdots a_n$ of letters.
We write $|w| = n$, and write $\epsilon$ for the empty word.
We write $\Sigma^*$ the infinite set of words over~$\Sigma$.
A \emph{language} $L$ is a subset of~$\Sigma^*$.
For $k \in \NN$, we denote~$L^{< k}$ the language~$\{w\in L \mid |w| < k\}$.
In particular we have~$L^{< 0} = \emptyset$.

In this paper we study \emph{regular languages}.
Recall that such a language can be described by a
\emph{deterministic finite automaton} (DFA) $A = (Q,
\Sigma, q_0, F, \delta)$, which consists of a finite set~$Q$ of
\emph{states}, an \emph{initial state} $q_0 \in Q$, a set $F \subseteq Q$ of
\emph{final states}, and a partial \emph{transition function}
$\delta\colon Q \times \Sigma \to Q$.
We write~$|A|$ the size of representing~$A$, which is $O(|Q| \times
|\Sigma|)$.
A \emph{(directed) path} in~$A$ from a state $q \in Q$ to a state $q' \in Q$ is
a sequence of states $q = q_0, \ldots, q_n = q'$ where for each $0 \leq i < n$ we have
$q_{i+1} = \delta(q_i, a_i)$ for some~$a_i$. For a suitable choice $a_0, \ldots,
a_{n-1}$, we call the word $a_0 \cdots a_{n-1} \in \Sigma^*$ a \emph{label} of
the path.
In particular, there is an empty path with label~$\epsilon$ from every state to itself.
The \emph{language} $\L(A)$ accepted by~$A$ consists of the words~$w$
that label a path from~$q_0$ to some final state.
We assume without loss of
generality that all automata are \emph{trimmed}, i.e., every state
of~$Q$ has a path from~$q_0$ and has a path to some final state; this can be
enforced in linear time.

\subparagraph*{Edit distances.}
For an alphabet~$\Sigma$,
we denote by~$\delta_\lev\colon\Sigma^*\times\Sigma^*\to\NN$
the \emph{Levenshtein edit distance}:
given~$u,v\in \Sigma^*$, 
the value $\delta_\lev(u,v)$ is the minimum number of \emph{edits} needed to transform~$u$ into~$v$, where the edit operations are
single-letter \emph{insertions}, \emph{deletions} or \emph{substitutions} (we omit their formal definitions).

While our lower bounds hold for the Levenshtein distance,
our positive results already hold with a restricted set of $2|\Sigma|+2$ edit operations called the \emph{push-pop edit operations}: $\pushL(a)$ and $\pushR(a)$ for $a \in \Sigma$, which respectively 
insert~$a$ at the beginning and at the end of the word,
and $\popL()$ and $\popR()$, which respectively
remove the first and last character of the word (and cannot be applied if the word is empty). Thus, we define the 
\emph{push-pop edit distance}, denoted~$\delta_\pp$,
like~$\delta_\lev$ but allowing only these edit operations.
\subparagraph*{Orderability.}
Fixing a distance function~$\delta\colon\Sigma^*\times\Sigma^*\to\NN$ over~$\Sigma^*$,
for a language~$L\subseteq \Sigma^*$ and~$d\in \NN$,
a \emph{$d$-sequence in~$L$} is a (generally infinite) sequence~$\bm{s}$
of words $w_1,\ldots,w_n,\ldots$ of~$L$ \emph{without repetition},
such that for every two consecutive words~$w_i,w_{i+1}$ in~$\bm{s}$
we have $\delta(w_i,w_{i+1}) \leq d$. We say that~$\bm{s}$ \emph{starts at~$w_1$}
and, in case~$\bm{s}$ is finite and has~$n$ elements, that~$\bm{s}$ \emph{ends
at~$w_n$} (or that~$\bm{s}$ is \emph{between~$w_1$ and~$w_n$}).
A~\emph{$d$-ordering of~$L$} is a~$d$-sequence~$\bm{s}$ in~$L$ such
that every word of~$L$ occurs 
in~$\bm{s}$; equivalently, it is a permutation of~$L$ such that any two consecutive words are at distance at most~$d$. An \emph{ordering} is a $d$-ordering for some $d \in \NN$.
If these exist, we call the language~$L$, respectively, \emph{$d$-orderable} and \emph{orderable}.
We call $L$ \emph{$(t,d)$-partition-orderable}
if it can be partitioned into~$t$ languages that each are $d$-orderable:
\begin{definition}
Let~$L$ be a language and~$t,d\in \NN$. We call $L$ 
  \emph{$(t,d)$-partition-orderable} if $L$ has a partition $L = \bigsqcup_{1
  \leq i \leq t} L_i$ 
  such that each~$L_i$ is $d$-orderable.\footnote{We use~$\bigsqcup$ for disjoint unions.}
\end{definition}

Note that, if we allowed repetitions in $d$-orderings, then the language of any
DFA~$A$ would be $O(|A|)$-orderable: indeed, any word $w$ can be transformed into a
word $w'$ of length~$O(|A|)$ by iteratively removing simple loops in the run
of $w$. By contrast, we will see in Section~\ref{sec:automaton} that allowing a \emph{constant} number of repetitions of each word 
makes no difference.

\begin{example}
We consider the Levenshtein distance in this example.
The language $(aa)^*$ is~$(1,2)$-partition-orderable (i.e., $2$-orderable)
and not~$(k,1)$-partition-orderable for any~$k\in
\NN$.  The language $a^* + b^*$ is~$(2,1)$-partition-orderable and not orderable, i.e., not~$d$-orderable for 
any~$d\in \NN$.
Any finite language is $d$-orderable with $d$ the maximal length of a word in~$L$.
The non-regular language $\{a^{n^2} \mid n \in \NN\}$ is not $(t,d)$-partition-orderable for any $t,d \in \NN$.
\end{example}

\subparagraph*{Enumeration algorithms.}
We study \emph{enumeration algorithms},
which output
a (generally infinite) sequence of \emph{edit scripts} $\sigma_1,
\sigma_2, \ldots$. We only study enumeration algorithms where each
\emph{edit script} $\sigma_i$ is a finite sequence of push-pop
edit operations. 
The algorithm enumerates a language~$L$ if the sequence satisfies the following condition:
letting $w_1$ be the result of applying~$\sigma_1$ on
the empty word, $w_2$ be the result of applying $\sigma_2$ to~$w_1$, and so on,
then all $w_i$ are distinct and $L = \{w_1, w_2, \ldots\}$.
If $L$ is infinite then
the algorithm does not terminate, but the infinite sequence ensures that every $w \in L$ is produced as the result
of applying (to~$\epsilon$) some finite prefix $\sigma_1, \ldots, \sigma_n$ of the output.

We aim for \emph{bounded-delay} algorithms, i.e., each edit script must be output in
time that only depends on the language~$L$ that is enumerated, but not on the current length of the words.
Formally, the algorithm can emit 
any push-pop edit operation and a delimiter $\outputt$, it must successively emit the edit operations of~$\sigma_i$ followed by~$\outputt$,
and there is a bound $T > 0$ (the delay) depending only on~$L$ such that the first $\outputt$ is emitted at most~$T$ operations after the beginning of the algorithm, and for each~$i>1$ the $i$-th $\outputt$ is emitted at most $T$ operations after the $(i-1)$-th $\outputt$.
Note that our notion of delay also accounts for what is usually called the preprocessing phase in the literature, i.e., the phase before the first result is produced.
Crucially the words $w_i$ obtained by applying the edit scripts~$\sigma_i$ are not written,
and $T$ 
does not depend on their length.

We say that a bounded-delay algorithm \emph{$d$-enumerates} a language~$L$ if it
produces a~\mbox{$d$-ordering}
of~$L$ (for the push-pop distance).
Thus, if~$L$ is~$d$-enumerable (by an algorithm), then~$L$ is in particular~$d$-orderable, and 
we will show that for regular languages, the converse also holds.

\begin{example}
  \SetKwFor{For}{for (}{) $\lbrace$}{$\rbrace$}
  \label{expl:enum_algo}
  Consider the regular language~$L \colonequals a^*b^* + b^*a^*$.
  This language is~$2$-orderable for the push-pop distance. Indeed,
  we can order it by increasing word length, finishing for word length~$i$ by the
  word~$a^i$ as follows. We start by length
  zero with the empty word~$\epsilon$ (so the first edit script is
  empty), then, assuming we have ordered all words of~$L$ of size~$\leq i$
  while finishing with~$a^i$, we continue with words of~$L$ of size~$i+1$
  in the following manner: we push-right the letter~$b$ to obtain~$a^i b$,
  and then we “shift” with edit scripts of the form $(\pushR(b); \popL())$ until
  we obtain~$b^{i+1}$, and then we shift again with edit
  scripts of the form~$(\pushR(a); \popL())$ until we obtain~$a^{i+1}$ as
  promised. This gives us an enumeration algorithm for~$L$,
  shown in Algorithm~\ref{alg:ppr}. As such, the delay of
  Algorithm~\ref{alg:ppr} is not bounded, because of the time
  needed to increment the integer variable~$\mathit{size}$: this variable
  becomes arbitrarily large throughout the enumeration, so it is not realistic
  to assume that we can increment it in constant time. This can however be fixed
  by working in a suitable 
  \emph{pointer machine model}, as explained next.
\end{example}

\begin{algorithm}
\SetKwData{Left}{left}\SetKwData{This}{this}\SetKwData{Up}{up}
\SetKwFunction{popL}{popL}\SetKwFunction{popL}{popL}\SetKwFunction{pushR}{pushR}\SetKwFunction{Output}{Output}\SetKwFunction{Union}{Union}\SetKwFunction{FindCompress}{FindCompress}
\SetKw{int}{int}
\BlankLine
\tcp{The first edit script is empty, corresponding to the empty word.}
\Output\;
  \int $\mathit{size} = 0$\;
\While{true}{
  $\mathit{size}{++}$\;
$\pushR(b)$ ; \Output\;
  \For{\int $j = 0$;\ $j < \mathit{size} -1$;\ $j{++}$}{
    $\pushR(b)$ ; $\popL()$ ; \Output\;
  }
  \For{\int $j = 0$;\ $j < \mathit{size}$;\ $j{++}$}{
    $\pushR(a)$ ; $\popL()$ ; \Output\;
  }
}
\caption{Push-pop enumeration algorithm for the language~$a^*b^* + b^*a^*$ from Example~\ref{expl:enum_algo}.}\label{alg:ppr}
\end{algorithm}\DecMargin{1em} 

Note that our enumeration algorithms run indefinitely, and thus use unbounded memory: this is unavoidable because their output would necessarily be ultimately periodic otherwise, which is not suitable in general
(see Appendix~\ref{apx:prelim}). To avoid specifying the size of memory cells or
the complexity of arithmetic computations (e.g., incrementing the integer
$\mathit{size}$ in Algorithm~\ref{alg:ppr}), we
consider a different model called \emph{pointer machines}~\cite{tarjan1979class}
which only allows arithmetic on a bounded domain.
We use this model for our enumeration algorithms (but not, e.g., our other complexity results such as PTIME bounds).

Intuitively, a pointer machine works with \emph{records} consisting of a
constant number of
labeled \emph{fields} holding either \emph{data values}
(in our case of constant size, i.e., constantly many possible values)
or \emph{pointers} (whose representation is not specified).
The machine has \emph{memory} consisting of a finite but unbounded collection of records,
a constant number of which are designated as \emph{registers} and are always accessible.
The machine can allocate records in constant time, retrieving a pointer to the memory location of the new record.
We can access the fields of records, read or write pointers, dereference them, and test them for equality, all in constant time, but we cannot perform any other manipulation on pointers or other arithmetic operations. (We can, however, count in unary with a linked list, or perform arbitrary operations on the constant-sized data values.)
See Appendix~\ref{apx:machine} for details.

\begin{example}
  Continuing Example~\ref{expl:enum_algo}, Algorithm~\ref{alg:ppr} can easily
  be adapted to a pointer-machine algorithm that 2-enumerates~$L$, maintaining the word in a 
  double-ended queue (deque) and keeping pointers to the first and last positions in
  order to know when to stop the \textbf{for} loops.  Deques can indeed be simulated in
  this machine model, e.g., with linked lists.
\end{example}

\begin{toappendix}
  \subsection{Necessity of unbounded memory}
  \label{apx:prelim}
We substantiate the claim that in general it is necessary for the memory usage to grow
indefinitely. First note that, if the
memory usage is bounded by some constant, then it is clear that the enumeration
is ultimately periodic: once the memory has reached its maximal number of
records, as there are only finitely many possible graph structures of pointers
and finitely many possible choices of data values (from a constant alphabet), then there are only finitely many
possible states of the memory, so considering the algorithm as a function
mapping one memory state to the next and optionally producing some values, this
function must be eventually periodic, and so is the output.

We show that, if the memory usage is bounded so that the sequence of edit
scripts is ultimately periodic, then we can only achieve slender languages (c.f. Section~\ref{sec:extensions} for the formal definition of slender languages).
  Hence, this is in general not sufficient (consider, e.g., $(a+b)^*$, which is
  enumerable but not slender). 

\begin{proposition}
  \label{prp:slenderlb}
  Let $L$ be a language achieved by a sequence of edit scripts which is
  ultimately periodic. Then $L$ is slender.
\end{proposition}

With regular languages, this proposition admits a converse: the regular slender
  languages can be achieved by unions of ultimately periodic sequences of edit
  scripts (see Proposition~\ref{prp:slenderub} which shows it for one term of
  the union).

  We now prove the proposition:

  \begin{proof}[Proof of Proposition~\ref{prp:slenderlb}]
  Consider the ultimately periodic suffix $u$ of edit scripts, and let $N$ be
  its length. Let $\Delta$ be the difference in length between the size of the
  current word under construction between the
  beginning and end of an occurrence of~$\Delta$, and let $M$ be a value such
  the length varies by at most $M$ while reading~$u$: specifically, if the
  length before reading $u$ is~$N$, then the length after reading~$u$ is~$N +
  \Delta$ and the intermediate lengths are between $N-M$ and $N+M$.

  If $\Delta \leq 0$, then the language is finite as the lengths of words
  visited during the rest of the enumeration when starting the periodic part at
  length~$N_0$ is upper bounded by~$N_0 + M$. Hence, it is in particular
    slender.

  If $\Delta > 0$, let us assume that we start the periodic part of the
  enumeration at a point where the current size of the word minus~$M$ is greater
  than any word size seen in the non-periodic part of the enumeration. Now, for
  any word length, we can only edit scripts producing words of this length
  for $(2M+1) \times |u|$ steps, which is constant, so we can only produce a
  constant number of words of every length. This is the definition of a slender
    language, so we conclude.
\end{proof}

  \subsection{Details about the machine model}
  \label{apx:machine}

  We first elaborate on the standard notion of a pointer machine
  from~\cite{tarjan1979class}. In this model, the memory consists of a finite
  but unbounded collection of records. Each record consists of a constant number
  of fields. The fields are either data values or pointers, and the number of
  possible data values is also constant. The machine can allocate records in
  constant time (retrieving a pointer to the new record), dereference a pointer
  in constant time to access a field of the corresponding record, test pointers
  for equality, and perform arbitrary operations on the data values (as there is
  only a constant set of possible values). However, no arithmetic operations
  involving pointers are permitted. In contrast with the RAM model, the pointer
  machine model makes it possible to work with memory that grows in an unbounded
  fashion, without worrying about how pointers are represented (they are assumed
  to take constant space) and how the memory is addressed.

  We next clarify how we represent the automaton given as input to the
  algorithm. The automaton is given as a linked list of states, with a pointer
  to the initial state, and an indication on each state of whether it is final
  or not and a pointer to an adjacency list for this state. The adjacency list
  of a state $q$ is again a linked list containing one element for each letter
  $a$ of
  the alphabet $\Sigma$, in order. For each state $q$ and letter $a$, the
  adjacency list item contains a pointer to an object representing the letter
  $a$, which the machine can use to output in constant time $\pushL(a)$ or
  $\pushR(a)$
  and a pointer to the target state of the transition (or an
  indication that the transition is undefined).

\end{toappendix}

\section{Interchangeability partition and orderability lower bound}
\label{sec:automaton}
In this section, we start the proof of our main result, Result~\ref{res:main}.
Let~$A$ be the DFA and let~$Q$ be its set of states.
The result is trivial if the language $\L(A)$ is finite, as we can always
enumerate it naively with distance $O(|A|)$ and some arbitrary delay bound, so in the rest of the proof we assume that $\L(A)$ is infinite.

We will first define a notion of \emph{interchangeability} on DFAs by introducing the notions of \emph{connectivity} and \emph{compatibility} on DFA states 
(this notion will be used in the next section to characterize orderability).
We then partition $\L(A)$ into languages $\L(A_1) \sqcup \cdots \sqcup \L(A_t)$ 
following a so-called \emph{interchangeability partition}, with each $A_i$ having this interchangeability property. Last, we show in the section our lower bound establishing that~$t$ is optimal.

\subparagraph*{Interchangeability.}
To define our notion of interchangeability, we first define the \emph{loopable} states of the DFA as those that are part of a non-empty cycle (possibly a self-loop):

\begin{definition} 
For a state $q\in Q$, we let $A_q$ be the DFA obtained from~$A$ by setting $q$
as the only initial and final state. We call~$q$ \emph{loopable}
if~$\L(A_q)\neq \{\epsilon\}$, and \emph{non-loopable} otherwise.
\end{definition}

We then define the \emph{interchangeability} relation on loopable states as the transitive closure of the union of two relations, called \emph{connectivity} and \emph{compatibility}:

\begin{definition}
We say that two loopable states $q$ and~$q'$ are \emph{connected} if there is a directed
path from~$q$ to~$q'$, or from~$q'$ to~$q$. 
We say that two loopable states~$q,q'$ are \emph{compatible} if~$\L(A_q) \cap
\L(A_{q'}) \neq \{\epsilon\}$. These two relations are symmetric and reflexive on loopable states.
We then say that two loopable states $q$ and $q'$ are
\emph{interchangeable} if they are in the transitive closure of the union of
the connectivity and compatibility relations. 
In other words, $q$ and $q'$ are interchangeable if there is a sequence $q = q_0, \ldots, q_n = q'$ of
loopable states such that for any $0 \leq i < n$, the states $q_i$ and $q_{i+1}$ are
either connected or compatible. 
Interchangeability is then an equivalence relation
over loopable states.
\end{definition}

Note that if two loopable states~$q, q'$ are in the same strongly connected
component (SCC) of~$A$ then they are connected, hence interchangeable.
Thus, we
can equivalently see the interchangeability relation at the level of SCCs
(excluding those that do not contain a loopable state, i.e., excluding the
trivial SCCs containing only one state having no self-loop).

\begin{definition}
  We call \emph{classes of interchangeable states}, or simply \emph{classes},
the equivalence classes of the interchangeability relation. 
Recall that, as $\L(A)$ is infinite, there is at least one class.
We say that the DFA $A$ is \emph{interchangeable} if the
partition  has only one class, in other words, if all loopable states of~$A$ are interchangeable.
\end{definition}

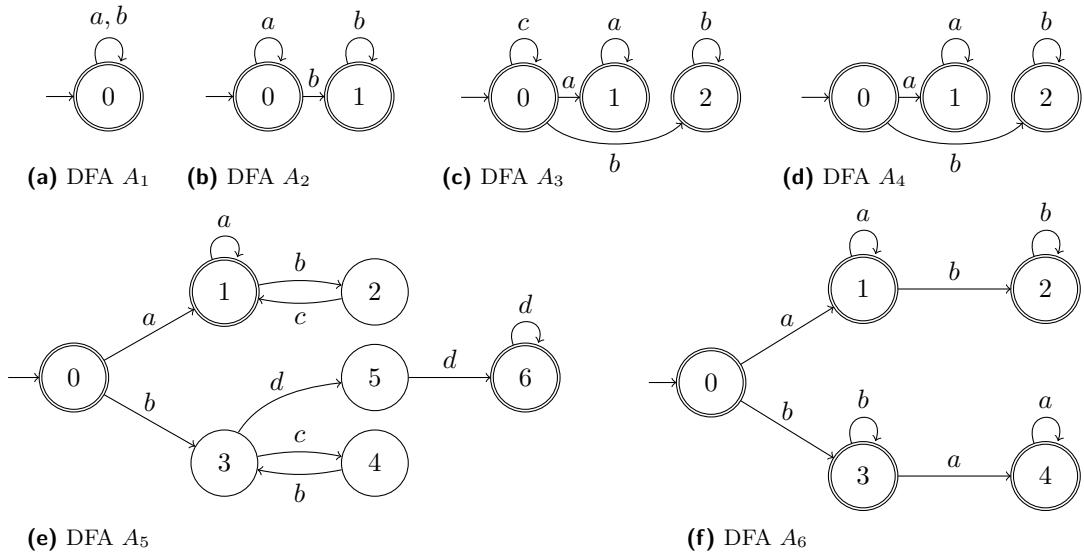
\begin{figure}%
\noindent%
\begin{minipage}{\linewidth}%
\noindent%
\begin{minipage}[t]{.12\linewidth}%
\noindent%
\begin{subfigure}[b]{\linewidth}
    \begin{tikzpicture}[initial text=,yscale=.5,xscale=.8]
      \node [state,initial,accepting] (q0) at (0, 0) {$0$};
      \path (q0) edge [->, loop above] node {$a,b$} (q0);
    \end{tikzpicture}
    \vspace{.25cm}
    \caption{DFA $A_1$}
    \label{fig:a1}
  \end{subfigure}
  \end{minipage}
    \hfill
  \begin{minipage}[t]{.21\linewidth}
  \begin{subfigure}[b]{\linewidth}
    \begin{tikzpicture}[initial text=,yscale=.5,xscale=.8]
      \node [state,accepting,initial] (q0) at (0, 0) {$0$};
      \node [state,accepting] (q1) at (1.5, 0) {$1$};
      \path (q0) edge [->, loop above] node {$a$} (q0);
      \path (q1) edge [->, loop above] node {$b$} (q1);
      \path (q0) edge [->, above] node {$b$} (q1);
    \end{tikzpicture}
    \vspace{.25cm}
    \caption{DFA $A_2$}
    \label{fig:a2}
  \end{subfigure}
  \end{minipage}
    \hfill
    \begin{minipage}[t]{.29\linewidth}
  \begin{subfigure}[b]{\linewidth}
    \begin{tikzpicture}[initial text=,xscale=.8,yscale=.5]
      \node [state,initial,accepting] (q0) at (0, 0) {$0$};
      \node [state,accepting] (q1) at (1.5, 0) {$1$};
      \node [state,accepting] (q2) at (3, 0) {$2$};
      \path (q0) edge [->, loop above] node {$c$} (q0);
      \path (q0) edge [->, above] node {$a$} (q1);
      \path (q0) edge [->, below, bend right=60] node {$b$} (q2);
      \path (q1) edge [->, loop above] node {$a$} (q1);
      \path (q2) edge [->, loop above] node {$b$} (q2);
    \end{tikzpicture}
    \null\vspace{-.4cm}
    \caption{DFA $A_3$}
    \label{fig:a3}
  \end{subfigure}
    \end{minipage}
    \hfill
    \begin{minipage}[t]{.29\linewidth}
  \begin{subfigure}[b]{\linewidth}
    \begin{tikzpicture}[initial text=,xscale=.8,yscale=.5]
      \node [state,initial,accepting] (q0) at (0, 0) {$0$};
      \node [state,accepting] (q1) at (1.5, 0) {$1$};
      \node [state,accepting] (q3) at (3, 0) {$2$};
      \path (q0) edge [->, above] node {$a$} (q1);
      \path (q0) edge [->, below, bend right=60] node {$b$} (q3);
      \path (q1) edge [->, loop above] node {$a$} (q1);
      \path (q3) edge [->, loop above] node {$b$} (q3);
    \end{tikzpicture}
    \null\vspace{-.4cm}
    \caption{DFA $A_4$}
    \label{fig:a4}
  \end{subfigure}
  \end{minipage}
  \end{minipage}
  \begin{minipage}{\linewidth}
    \begin{minipage}[t]{.55\linewidth}
  \begin{subfigure}[t]{\linewidth}
    \hspace{-.5cm}\begin{tikzpicture}[initial text=,xscale=.9,yscale=.5]
      \node [state,initial,accepting] (q0) at (0, 0) {$0$};
      \node [state,accepting] (q1) at (2.2, 2.27) {$1$};
      \node [state] (q2) at (4.4, 2.27) {$2$};
      \node [state] (q3) at (2.2, -2.27) {$3$};
      \node [state] (q4) at (4.4, -2.27) {$4$};
      \node [state] (q5) at (4.4, 0) {$5$};
      \node [state,accepting] (q6) at (6.6, 0) {$6$};
      \path (q0) edge [->, above] node {$a$} (q1);
      \path (q0) edge [->, above] node {$b$} (q3);
      \path (q3) edge [->, bend left=30, above] node {$d$} (q5);
      \path (q5) edge [->, above] node {$d$} (q6);
      \path (q1) edge [->, loop above] node {$a$} (q1);
      \path (q6) edge [->, loop above] node {$d$} (q6);
      \path (q1) edge [->, above,bend left=20] node {$b$} (q2);
      \path (q2) edge [->, below,bend left=20] node {$c$} (q1);
      \path (q3) edge [->, above,bend left=20] node {$c$} (q4);
      \path (q4) edge [->, below,bend left=20] node {$b$} (q3);
    \end{tikzpicture}
    \caption{DFA $A_5$}
    \label{fig:a5}
  \end{subfigure}
  \end{minipage}\hfill
    \begin{minipage}[t]{.38\linewidth}
  \begin{subfigure}[t]{\linewidth}
    \hspace{-.75cm}\begin{tikzpicture}[initial text=,xscale=.8,yscale=.5]
      \node [state,initial,accepting] (q0) at (0, 0) {$0$};
      \node [state,accepting] (q1) at (2.5, 2.48) {$1$};
      \node [state,accepting] (q2) at (2.5, -2.48) {$3$};
      \node [state,accepting] (q3) at (5.5, 2.48) {$2$};
      \node [state,accepting] (q4) at (5.5, -2.48) {$4$};
      \path (q0) edge [->, above] node {$a$} (q1);
      \path (q0) edge [->, above] node {$b$} (q2);
      \path (q1) edge [->, above] node {$b$} (q3);
      \path (q2) edge [->, above] node {$a$} (q4);
      \path (q3) edge [->, loop above] node {$b$} (q3);
      \path (q4) edge [->, loop above] node {$a$} (q4);
      \path (q1) edge [->, loop above] node {$a$} (q1);
      \path (q2) edge [->, loop above] node {$b$} (q2);
    \end{tikzpicture}
    \caption{DFA $A_6$}
    \label{fig:a6}
  \end{subfigure}
  \end{minipage}
  \end{minipage}
  \caption{Example DFAs from Example~\ref{exa:interchangeable}}
  \label{fig:interchangeable}
\end{figure}

\begin{example}
  \label{exa:interchangeable}
  The DFA $A_1$ shown in Figure~\ref{fig:a1} for the language $(a + b)^*$ has only one loopable state, so~$A_1$ is
  interchangeable.

  The DFA $A_2$ shown in Figure~\ref{fig:a2} for the language $a^* b^*$ has two loopable states~$0$ and~$1$ which
  are connected, hence interchangeable. Thus, $A_2$ is interchangeable.

  The DFA~$A_3$ shown in Figure~\ref{fig:a3} for the language $c^* (a^* + b^*)$
  has three loopable states: $0$,~$1$ and $2$.
  The states~$0$ and~$1$ are connected, and~$0$ and~$2$ are
  also connected, so all loopable states are interchangeable
  and~$A_3$ is interchangeable.

  The DFA~$A_4$ shown in Figure~\ref{fig:a4} for the language $a^* + b^*$ 
  has two loopable states~$1$ and~$2$ which are neither connected nor compatible.
  So $A_4$ is not interchangeable.

  The DFA~$A_5$ shown in Figure~\ref{fig:a5} for the language $a (a + bc)^* + b (cb)^* ddd^*$ mentioned in the introduction
  has five loopable states:~$1$, $2$, $3$, $4$, and~$6$.
  Then~$1$ and~$2$ are connected, $3$ and~$4$
  are connected, $3$ and~$6$ are connected, and~$1$ and~$4$ are compatible (with the word
  $bc$). Hence,
  all loopable states are interchangeable and~$A_5$ is interchangeable.

  The DFA~$A_6$ shown in Figure~\ref{fig:a6} for the language~$a^*b^* + b^*a^*$
  from Example~\ref{expl:enum_algo} has four loopable states: $1$, $2$, $3$, and
  $4$. Then~$1$ and~$2$ are connected,~$3$ and~$4$ are connected, and (for
  instance) $1$
  and~$4$ are compatible (with the word~$a$). Hence all loopable states are
  interchangeable and~$A_6$ is interchangeable.
\end{example}

\subparagraph*{Interchangeability partition.} We now 
partition~$\L(A)$ using interchangeable DFAs:
\begin{definition}
  An \emph{interchangeability partition} of~$A$ is a sequence $A_1,
  \ldots, A_t$ of DFAs such that $\L(A)$ is the disjoint union of the $\L(A_i)$
  and every $A_i$ is interchangeable. Its \emph{cardinality} is the number~$t$
  of~DFAs.
\end{definition}

Let us show how to compute an
interchangeability partition whose cardinality is the number of classes. We will later show that this cardinality is optimal.
Here is the statement:

\begin{propositionrep}
  \label{prp:computeip}
  We can compute in polynomial time in~$A$ an
  interchangeability partition $A_1, \ldots, A_t$ of~$A$,
  with $t \leq |A|$ the number of classes of interchangeable states.
\end{propositionrep}

\begin{toappendix}
  To prove this result, we first show the auxiliary claim stated in the main
  text:
\end{toappendix}

Intuitively, the partition is defined following the
classes of~$A$. Indeed, considering any word $w \in \L(A)$ and its accepting run $\rho$
in~$A$, for any loopable state $q$ and $q'$ traversed in~$\rho$, the word~$w$
witnesses that $q$ and~$q'$ are connected, hence interchangeable.
Thus, we would like to
partition the words of~$\L(A)$ based on the common class of the loopable states traversed in their accepting run.
The only subtlety is that $\L(A)$ may also contain words whose accepting run does not traverse any loopable state, called \emph{non-loopable words}. 
For instance, 
$\epsilon$ is a non-loopable word of $\L(A_5)$ for $A_5$ given in
Figure~\ref{fig:a5}.
Let us formally define the non-loopable words, and our partition of the loopable words based on the interchangeability classes:

\begin{definition}
  A word $w = a_1 \cdots a_n$ of~$\L(A)$ is \emph{loopable} if,
  considering its accepting run $q_0, \ldots, q_n$
  with $q_0$ the initial state and $q_i = \delta(q_{i-1}, a_i)$ for $1 \leq i \leq n$,
  one of the $q_i$ is loopable. Otherwise, $w$ is 
  \emph{non-loopable}. We write $\NL(A)$ the set of the non-loopable
  words of~$\L(A)$.
  
  Letting $\calC$ be a class of interchangeable states, we write $\L(A,
  \calC)$ the
  set of (loopable) words of~$\L(A)$ whose accepting run traverses a state of~$\calC$.
\end{definition}

We then have the following, with finiteness of $\NL(A)$ shown by the pigeonhole principle:

\begin{claimrep}
  \label{clm:langpart}
  The language $\L(A)$ can be partitioned as $\NL(A)$ and $\L(A, \calC_1),
  \ldots, \L(A, \calC_t)$ over the
  classes~$\calC_1, \ldots, \calC_t$ of interchangeable states, and further $\NL(A)$ is finite.
\end{claimrep}

\begin{proof}
We first show that these sets are pairwise disjoint. This is clear by
definition for~$\NL(A)$ and~$\L(A,\calC)$ for any class~$\calC$ of
interchangeable states.
Further, for $\calC\neq \calC'$, we also have that
  $\L(A,\calC) \cap \L(A,\calC') =\emptyset$, because for any word $w$ in their intersection,
  considering the accepting run of~$w$, it must go via a loopable state 
  $q$ of~$\calC$ and
  a loopable state $q'$ of~$\calC'$, so this run witnesses that $q$ and $q'$ are
  connected, hence interchangeable, which would contradict the fact that $\calC$ and
  $\calC'$ are in different classes.

  We next show both inclusions.
  By definition, a word of~$\L(A)$ is either loopable, hence in some $\L(A,\calC)$, or
  non-loopable and in $\NL(A)$. 
  For the converse inclusion, by definition again all words of~$\NL(A)$ and
  $\L(A, \calC)$ for any class~$\calC$ are in $\L(A)$.

  Last, the fact that $\NL(A)$ is finite is because any non-loopable word~$w$
has length at most $|A|-1$, because otherwise by the pigeonhole principle the
same state $q$ would appear twice in its accepting run, witnessing that $q$ is
loopable and contradicting the fact that~$w$ is non-loopable. Thus, there are
finitely many words in~$\NL(A)$.
\end{proof}

We now construct an interchangeability partition of~$A$ of the right cardinality
by defining one DFA~$A_i$ for each class of interchangeable states,
where we simply remove the loopable states of the other classes.
These DFAs are interchangeable by construction.
We modify the DFAs to ensure that the non-loopable words are only captured by~$A_1$.
This construction (explained in the appendix) is doable in
PTIME, in particular the connectivity and compatibility relations can be
computed in PTIME, testing compatibility by checking the nonemptiness of product automata.
This establishes Proposition~\ref{prp:computeip}.

\begin{toappendix}
  We are now ready to prove Proposition~\ref{prp:computeip}:

  \begin{proof}[Proof of Proposition~\ref{prp:computeip}]
    Given the automaton~$A$ with state space~$Q$, we first materialize in PTIME the connectivity and compatibility relations by naively testing
    every pair of states. Note that to test the compatibility of two states $q$
and~$q'$ we build the automaton $A_q$ and $A_{q'}$, build their intersection
automaton by the product construction, and check for emptiness, which is doable
in PTIME. We then materialize the interchangeability relation, and compute the
    classes $\calC_1, \ldots, \calC_t$. Note that $t \leq |Q|$.

    The high-level argument is that we create one copy of~$A$ for each class
    $\calC_i$, but ensure that only the first copy accepts the non-loopable
    words. For this, we need to duplicate the automaton in two copies, so as to
    keep track of whether the run has passed through a loopable state or not,
    and ensure that in all automata but the first, the states of the first copy
    (where we accept the non-loopable words) are non-final.

    Formally, we modify~$A$ to a DFA~$A'$ by doing the following: the new state space
    is $Q \times \{0, 1\}$, the final states are the $(q, b)$ for $q$ final
    in~$A$ and $b \in \{0, 1\}$, the initial state is
    $(q_0, b)$ for $q_0$ the
    initial state of~$A$ and $b$ being $1$ if~$q_0$ is loopable (which happens only when~$\NL(A)=\emptyset$) and $0$
    otherwise, and the transitions are:
    \begin{itemize}
      \item For every $q \in Q$ and $a \in \Sigma$ such that $\delta(q, a)$ is
        non-loopable, the transitions $\delta((q,b), a) = (\delta(q, a), b)$ for
        all $b \in \{0, 1\}$
      \item For every $q \in Q$ and $a \in \Sigma$ such that $\delta(q, a)$ is
        loopable, the transitions $\delta((q, b), a) = (\delta(q, a), 1)$ for
        all~$b \in \{0, 1\}$.
    \end{itemize}
    We then trim~$A'$. Note that if~$q_0$ is loopable in~$A$ then this in
    fact removes all states of the form~$(q,0)$:
indeed no such state is reachable from~$(q_0,1)$, simply because there are no
transitions from states with second component~$1$ to states with second component~$0$.

    Clearly the DFA $A'$ is such that $\L(A') = \L(A)$. However, $A'$
    intuitively keeps track of whether a loopable state has been visited, as is
    reflected in the second component of the states. Further, the loopable states
    of $A$ are in bijection with the loopable states of~$A'$ by the operation
    mapping a loopable state $q$ of~$A$ to the state $(q, 1)$ of~$A'$. Indeed, $(q, 1)$ is clearly a loopable state of~$A'$, and this
    clearly describes the loopable states of~$A'$ with second
    component~$1$. Further, 
    the states of~$A'$ with second component~$0$ are non-loopable
    because any loop involving such a state $(q, 0)$ would witness a loop on
    the corresponding states of~$A$, so that they would be loopable in~$A$, and 
    each transition of the loop in~$A'$ would then lead by definition to a
    state of the form~$(q',1)$, which is impossible because no transition can
    then lead back to a state of the form~$(q'', 0)$. Last, the
    interchangeability relationship on~$A'$ is the same as that relationship
    on~$A$ up to adding the second component with value~$1$, because all loopable states
    of~$A'$ all have that value in their second component
    as we explained, and the states and transitions with value~$1$ in the second
    component are isomorphic to~$A$ so the compatibility and connectivity
    relationships are the same.

    We now compute copies $A_1, \ldots, A_t$ of the DFA~$A'$, and modify them as
    follows: for each $1 \leq i \leq t$, we remove from~$A_i$ the states $(q,1)$
    for each $q$ in a class $\calC_j$ with $j \neq i$.
    Further, for $i > 1$, we make non-final all states with second component~$0$
    (but leave these states final in~$A_1$).

    This process is in polynomial time. Further, we claim that each $A_i$ is interchangeable.
    Indeed, its loopable states are precisely the 
    loopable states of~$A$ that are in $\calC_i$ (up to adding the second
    component with value~$1$),
    because removing states from different classes does not change the loopable or non-loopable status of states in~$\calC_i$.
    Further, for any two loopable states $q$ and $q'$ of $\calC_i$ (modified to add the
    second component with value~$1$),
    the fact that they are interchangeable in~$A$
    was witnessed by a sequence $q = q_0, \ldots, q_n =
    q'$ with $q_j$ and $q_{j+1}$ being either connected or compatible in~$A$ for
    all $0 \leq j < n$. Note that when two states $q_i$ and $q_{i+1}$ are
    connected, then there is a sequence of loopable states $q_i = q'_{i,0}, \ldots,
    q'_{i,n_i}$ where any two consecutive loopable states $q'_{i,j}$ and $q'_{i,j+1}$ for
    $0 \leq j < n_i$ are connected via a (directed) path consisting only of non-loopable
    states, which we call \emph{immediately connected}. Hence, up to modifying the previous path to a longer path 
    $q = q_0, \ldots, q_m = q'$,
    we can assume that two successive loopable states $q_j$ and $q_{j+1}$ 
    with all $0 \leq j < n$
    are either compatible or immediately connected.
    Now, the path witnesses that all $q_j$ for $0 \leq j \leq
    m$ are also in~$\calC_i$ (up to adding the second component).
    Further, for any two loopable states in~$A$ that are
    compatible or immediately connected in~$A$, this compatibility or immediate
    connectivity relation still holds when removing loopable states that are not in $\calC_i$ (up to adding the second component).
    Thus, the same
    path witnesses that the loopable states $q$ and~$q'$ are still interchangeable
    in~$A_i$. This implies that indeed $A_i$ is interchangeable, as all its
    loopable states are in the same class, which is exactly $\calC_i$ (up to changing the second component).

    It remains to show that $\L(A_1), \ldots, \L(A_t)$ is a partition
    of~$\L(A)$. For this, we show that $\L(A_1) = \L(A, \calC_1) \cup \NL(A)$
    and $\L(A_i) = \L(A, \calC_i)$ for $i > 1$; this suffices to conclude by
    Claim~\ref{clm:langpart}. 

    We first show that $\L(A, \calC_i) \subseteq \L(A_i)$ for $1 \leq i \leq t$.
    This is because the accepting run of a word of $\L(A, \calC_i)$ must go
    through a state of~$\calC_i$, and we have argued in the proof of
    Claim~\ref{clm:langpart} that it does not use any state of~$\calC_j$ for $j
    \neq i$. Thus, we can build a corresponding path in~$A_i$, where the second
    component of each state is~$0$ until we reach the first state of~$\calC_i$
    in the accepting run, and~$1$ afterwards. Note that if the initial state is
    in~$\calC_i$ then we defined the initial state of~$A'$ to have
    second component~$1$, so the path can start with second component one.
    In particular, letting $q$ the
    final state reached in~$A$ by the accepting run, the corresponding run
    in~$A_i$ reaches $(q,1)$, which is final, and the word is accepted by~$A_i$.

    We then show that $\NL(A) \subseteq \L(A_1)$. This is because the accepting
    run of a non-loopable word of~$A$ goes only through non-loopable states,
    hence does not go through states of $\calC_j$ with $j \neq 1$, hence we can
    reflect the accepting path in~$A_1$ and reach a final state, witnessing that
    the word is accepted.

    Now, we show that $\L(A_i) \subseteq \L(A, \calC_i)$ for $i > 1$. This is
    because an accepting run for a word $w$ of~$\L(A_i)$ must finish by a state of the form $(q,1)$
    where $q$ is final, as these are the only final states of~$A_i$. Thus, there
    is an accepting path for the word in~$A$, and $w \in \L(A)$. Further, the
    construction of~$A'$, hence of~$A_i$, guarantees that, if we reach a state with second
    component~$1$, then we have gone via a state $(q,1)$ with~$q$ loopable
    in~$A$. The construction of~$A_i$ further guarantees that this~$q$ must 
    be in~$\calC_i$. Considering the accepting path of~$w$ in~$A$, we see that
    this path goes through the state $q$, so that $w \in \L(A, \calC_i)$.

    Last, we show that $\L(A_1) \subseteq \L(A, \calC_1) \cup \NL(A)$.
    Considering a word $w \in \L(A_1)$ and its accepting run, there are two
    possibilities. Either the run ends at a state of the form~$(q,1)$, in which
    case the same reasoning as in the previous paragraph shows that $w \in \L(A,
    \calC_1)$. Or the run ends at a state of the form~$(q, 0)$. In this case,
    the construction of~$A'$, hence of~$A_1$, guarantees that the entire
    accepting path only goes via states of the form~$(q', 0)$, so the
    corresponding run in~$A$ only goes via non-loopable states and $w \in
    \NL(A)$.

    The claims that we established together with Claim~\ref{clm:langpart} show
    that $\L(A_1), \ldots, \L(A_t)$ is indeed a partition of $\L(A)$, concluding
    the proof.
  \end{proof}

\end{toappendix}

\subparagraph*{Lower bound.}
We have shown how to compute an interchangeability partition of a DFA $A$ with cardinality 
the number $t$ of classes. 
Let us now show that this value of~$t$ is optimal, in the sense that $\L(A)$
cannot be partitioned into less than~$t$
orderable (even non-regular) languages.
This lower bound holds even when allowing
Levenshtein edits.
Formally:

\begin{theorem}
  \label{thm:lower}
  For any partition of the language $\L(A)$ as
  $\L(A) = L_1 \sqcup \cdots \sqcup L_{t'}$  if for each $1 \leq i \leq t'$ the language
  $L_i$ is orderable for the Levenshtein distance, then we have $t' \geq t$ for $t$ the number of classes of~$A$.
\end{theorem}

This establishes the negative part of Result~\ref{res:main}.
Incidentally, this lower bound can also be shown even if the unions are not disjoint,
indeed even if we allow repetitions, provided that there is some constant bound on the number of repetitions of each word.

Theorem~\ref{thm:lower} can be shown from the following claim which establishes
that sufficiently long words from different classes are arbitrarily far away for the Levenshtein distance:

\begin{propositionrep}
  \label{prp:lowerdist}
  Letting $\calC_1, \ldots, \calC_t$ be the classes of~$A$,
  for any distance $d \in \NN$,
  there is a threshold $l \in \NN$ such that
  for any two words $u \in \L(A, \calC_i)$ and $v \in \L(A, \calC_j)$
with~$i\neq j$ and~$|u| \geq l$ and~$|v| \geq l$, we have $\delta_\lev(u, v) >
d$.
\end{propositionrep}

This proposition implies Theorem~\ref{thm:lower} because, 
if we could partition~$\L(A)$ into less than~$t$ orderable languages,
then some ordering must include
infinitely many words from two different classes~$\L(A, \calC_i)$ and~$\L(A, \calC_j)$, hence alternate infinitely often between the two. Fix the
distance~$d$, and consider a point when all 
words of~$L$ of length~$\leq \max(l,\max_{w\in \NL(A)} |w|)$
have been enumerated,
for~$l$ the threshold of the proposition: then it is no longer possible for any ordering to move from one class to another, yielding a contradiction.
As for the proof of Proposition~\ref{prp:lowerdist}, we give a sketch below (the complete proofs are in appendix):

\begin{proofsketch}
  Given a sufficiently long word $u \in \L(A, \calC_i)$,
  by the pigeonhole principle its run must contain a large number of loops over some state $q \in \calC_i$. Assume that we can edit~$u$ into~$v \in \L(A, \calC_j)$ with $d$ edit operations: this changes at most~$d$ of these loops. Now, considering the accepting run of~$v$ and using the pigeonhole principle again on the sequence of endpoints of contiguous unmodified loops, we deduce that some state $q'$ occurs twice; then $q' \in \calC_j$ by definition of $\L(A, \calC_j)$.
  The label of the resulting loop on~$q'$ is then also the label of a loop on~$q$, so $q$ and $q'$ are compatible, hence $\calC_i = \calC_j$.
\end{proofsketch}

\begin{proof}
  Fix the distance $d \in \NN$, and
  let $k$ be the number of states of the automaton~$A$.
  We prove that the claim holds for the threshold $l \colonequals k \times (k(d+1) +1)$.

  Assume by contradiction that there exist words $u \in \L(A, \calC_i)$ 
  and $v \in \L(A, \calC_j)$ with $i \neq j$
  such that $|u| \geq l$,
  $|v|\geq l$ and~$\delta(u,v)\leq d$.
  By the pigeonhole principle, there must be a state~$q$ occurring in at least $k(d+1)+1$ positions
  in the accepting run for~$u$, i.e., we can write 
$u = h u_1 \cdots u_{k(d+1)}  t$ 
for some $h$ and~$t$ where each $u_i$ starts and ends at~$q$. This implies in particular that $q$ is
  loopable so it is in the
  class $\calC_i$ as explained in the proof of Claim~\ref{clm:langpart}.

  As~$\delta(u,v)\leq d$,
  consider an arbitrary edit script that transforms $u$ into~$v$ using at
  most $d$ edit operations of the Levenshtein distance.
  Observe then that some contiguous sequence of $k$ of the~$u_j$'s must be untouched by these
  edit operations. That is, we can write $v = h' u_i \cdots u_{i+k-1} t'$ for
  some~$h'$ and~$t'$.

  Now, by the pigeonhole principle again, in the accepting run of~$v$, of the $k+1$ states at the endpoints of the $u_i, \ldots,
  u_{i+k-1}$, two must be the same, say $q'$. So we can write $v
  = h'' u_l \cdots u_r t''$ for some $h''$ and $t''$ and $i \leq l < r \leq i+k-1$, where in the
  accepting run of~$v$ we have the same state $q'$ at the beginning and end
  of~$u_l \cdots u_r$. We know in particular that $q'$ is loopable, so it is in~$\calC_j$, again using the argument in the proof of Claim~\ref{clm:langpart}.

  But now the word $u_l \cdots u_r$ is the label of a path in~$A$ from~$q'$
  to~$q'$ as seen on the accepting run of~$v$ in the previous paragraph, and it
  is also the label of a path in~$A$ from~$q$ to~$q$, as seen on the accepting
  run of~$u$ earlier. Thus, the states $q$ and~$q'$ are compatible, and they are 
  in the same class, which contradicts the fact that~$q$ and~$q'$ were in different classes. This concludes the proof.
\end{proof}

\begin{toappendix}
From Proposition~\ref{prp:lowerdist}, we can show Theorem~\ref{thm:lower}:

\begin{proof}[Proof of Theorem~\ref{thm:lower}]
  We proceed by contradiction and assume that $\L(A)$ is $(t-1,d)$-partition-orderable for some bound
  $d> 0$ on the Levenshtein distance, i.e., we have partitioned $\L(A) = \bigsqcup_{1 \leq i \leq t-1} L_i$ where each $L_i$ is $d$-orderable, i.e., has a $d$-sequence $\mathbf{s}_i$.
  Let $l$ be the threshold given by 
  Proposition~\ref{prp:lowerdist}, let $l'$ be the maximal length of a non-loopable word in~$\NL(A)$, which is finite by Claim~\ref{clm:langpart}, and let $l'' = \max(l, l')$.
  There is some index $p > 0$ such that, for all $1 \leq i \leq t-1$,
  all words of length $\leq l''$ of~$\L(A)$ are only present in the initial prefix of length $p$ of the $\mathbf{s}_i$, in particular this is the case of all non-loopable words.
  Now, as there are $t-1$ $d$-sequences and $t$ classes, and as the languages
  of each of the $t$ classes
  are infinite, there must be a $d$-sequence which after point~$p$  contains 
  infinitely many words from two different subsets $\L(A, \calC_i)$ and $L(A, \calC_j)$ 
  with $i \neq j$. In particular, there must be an index $p' > p$ such that the $p'$-th word of~$\mathbf{s}_i$ is a word of  $\L(A, \calC_i)$  of length $> l''$, hence $\geq l$, and the $(p'+1)$-th word of~$\mathbf{s}$ is a loopable word of~$\L(A)$ also of length $\geq l''$ which is not in $\L(A, \calC_i)$, hence it is in some $\L(A, \calC_k)$ with $k \neq i$.
  However, we know by 
  Proposition~\ref{prp:lowerdist} that the words in $\L(A, \calC_i)$ of length $\geq l$ are at Levenshtein distance $>d$ from the words of $\L(A, \calC_k)$ of length $\geq l$.
  This contradicts the fact that $\mathbf{s}_i$ is a $d$-sequence, and concludes the proof.
  
  We last substantiate the remark made in the main text that the proof still applies if we allow each word to be repeated some constant number $C$ of times. The proof above works as-is, except that we now define $p > 0$ to be the smallest index after which there are no longer any occurrences of any word of length $\leq l''$ in the remainder of all sequences $\mathbf{s}_i$: as there are only finitely many words occurring finitely many times, this is well-defined. The rest of the proof is unchanged.
\end{proof}

\end{toappendix}

\section{Orderability upper bound}
\label{sec:upper}
\begin{toappendix}
  \label{apx:upper}
\end{toappendix}

We have shown in the previous section that
we could find an interchangeability partition of any regular language $\L(A)$
into languages $\L(A_1), \ldots, \L(A_t)$ of interchangeable DFAs, for $t$ the number of classes. We know by our lower bound (Theorem~\ref{thm:lower}) that we cannot hope to order~$\L(A)$ with less than~$t$ sequences.
Thus, in this section, we focus on each interchangeable~$A_i$ separately, and
show how to order $\L(A_i)$ as one sequence.
Hence, we fix for this section a DFA $A$ that is interchangeable, write $k$ its
number of states, and show that~$\L(A)$ is orderable. We will in fact show that
this is the case for the push-pop distance:

\begin{theorem}
  \label{thm:orderable}
  For any interchangeable DFA~$A$, 
  the language $\L(A)$ is $48k^2$-orderable for the push-pop distance.
\end{theorem}

We show this result in the rest of this section, and strengthen it in the
next section to a bounded-delay algorithm. 
Before starting, we give an overview of the
structure of the proof.
The proof works by first introducing \emph{$d$-connectivity} of a language (not to be confused with the connectivity relation on loopable automaton states). This weaker notion is necessary for $d$-orderability, but for finite languages we will show a kind of converse: $d$-connectivity implies~$3d$-orderability.
We will then show that~$\L(A)$ is \emph{stratum-connected}, i.e., the finite
\emph{strata} of words of~$\L(A)$ in some length interval are each $d$-connected
for some common~$d$. Last, we will show show that this implies orderability, using the result on finite languages.

\subparagraph*{Connectivity implies orderability on finite languages.}
We now define \emph{$d$-connectivity}:

\begin{definition}
A language~$L$ is \emph{$d$-connected} if for every pair of words~$u,v \in L$,
there exists a $d$-sequence in~$L$ between~$u$ and~$v$.
\end{definition}

Clearly $d$-connectivity is a necessary condition for
$d$-orderability: indeed if~$w_1,w_2,\ldots$ is
a~$d$-ordering of~$L$, and~$u=w_i$, $v=w_j$ are two words of~$L$ with $i\leq j$
(without loss of generality), then~$w_i,w_{i+1},\ldots,w_j$ is indeed
a~$d$-sequence in~$L$ between~$u$ and~$v$.
What is more, for finite languages, the converse holds, up to multiplying the distance by a constant factor:

\begin{toappendix}
  \subsection{Proof of Lemma~\ref{lem:connected-enum}}
\end{toappendix}

\begin{lemmarep}
\label{lem:connected-enum}
Let~$L$ be a finite language that is $d$-connected
  and~$s \neq e$ be words of~$L$. Then there exists a $3d$-ordering
of~$L$ starting at~$s$ and ending at~$e$.
\end{lemmarep}

\begin{proofsketch}
  We use the fact, independently proved by Sekanina and by
  Karaganis~\cite{sekanina1960ordering,karaganis1968cube}, that the
  cube of every connected graph~$G$ has a Hamiltonian path between any
  pair of vertices (see also~\cite{mutze2022combinatorial}). One algorithmic way
  to see this is by 
  traversing a spanning tree of~$G$ and handling odd-depth and
  even-depth nodes in prefix and postfix fashion (see,
  e.g.,~\cite{uno2003two}). Applying this to the graph~$G$ whose vertices are the words of~$L$ and where two words~$w,w'$ are connected by an edge when~$\delta(w,w')\leq d$ yields the result.
\end{proofsketch}

The constant $3$ in this lemma is optimal, as follows from
\cite{radoszewski2011hamiltonian}; see Appendix~\ref{apx:upper} for more details.
Note that the result does not hold for infinite languages: $a^* + b^*$ is
1-connected (via $\epsilon$) but not $d$-orderable for any~$d$.

\begin{toappendix}
  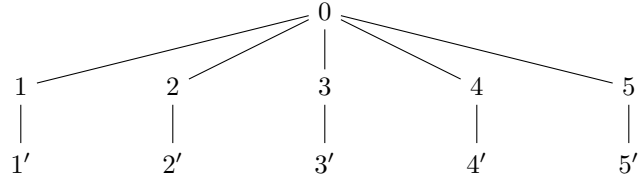
\begin{figure}
    \centering
    \begin{tikzpicture}[xscale=2]
      \node (o) at (0, 0) {$0$};
      \node (n1) at (-2, -1) {$1$};
      \node (n2) at (-1, -1) {$2$};
      \node (n3) at (0, -1) {$3$};
      \node (n4) at (1, -1) {$4$};
      \node (n5) at (2, -1) {$5$};
      \node (m1) at (-2, -2) {$1'$};
      \node (m2) at (-1, -2) {$2'$};
      \node (m3) at (0, -2) {$3'$};
      \node (m4) at (1, -2) {$4'$};
      \node (m5) at (2, -2) {$5'$};
      \draw (o) -- (n1);
      \draw (o) -- (n2);
      \draw (o) -- (n3);
      \draw (o) -- (n4);
      \draw (o) -- (n5);
      \draw (m1) -- (n1);
      \draw (m2) -- (n2);
      \draw (m3) -- (n3);
      \draw (m4) -- (n4);
      \draw (m5) -- (n5);
    \end{tikzpicture}
    \caption{Tree used in the proof of Claim~\ref{clm:treeopt}}
    \label{fig:treeopt}
  \end{figure}

  The result follows from the following claim on trees:
  \begin{lemma}[\cite{sekanina1960ordering,karaganis1968cube}]
    \label{lem:tree}
    Let $T$ be a tree, i.e., an acyclic connected undirected graph.
    Let $s \neq e$ be arbitrary nodes of~$T$.
    We can compute in linear time in~$T$ a sequence $s = n_1, \ldots,
    n_m = e$ enumerating all nodes of~$T$ exactly once, starting and ending
    at~$s$ and~$e$ respectively, such that the distance between any two
    consecutive nodes is at most~$3$.
  \end{lemma}

  As we explained in the main text, Lemma~\ref{lem:tree} was already shown
  independently by Sekanina and by
  Karaganis~\cite{sekanina1960ordering,karaganis1968cube}.
  However, we give our own algorithmic proof of this result, to
  emphasize the fact that the Hamiltonian path in question can be computed in
  linear time: we will need this fact in Section~\ref{sec:worddag}. Our algorithmic
  proof uses the method of exploring trees and enumerating nodes differently
  depending on the parity of their depth (see, e.g.,~\cite{uno2003two}). The
  same algorithm can be easily derived from the proof
  of~\cite{karaganis1968cube}, by turning the inductive argument into a recursive
  algorithm.

  \subparagraph*{Optimality of the constant~$\bm{3}$.}
  Before giving our proof of Lemma~\ref{lem:tree}, we
  show that the value~$3$ in this lemma is optimal, as was claimed in
  the main text:
  \begin{claim}
    \label{clm:treeopt}
    There is an acyclic connected undirected graph~$T$ such that any sequence
enumerating all nodes of~$T$ exactly once must contain a pair of vertices at
distance
    $\geq 3$.
  \end{claim}
  Claim~\ref{clm:treeopt} is implied by the results of~\cite{radoszewski2011hamiltonian}
  giving a (non-trivial) characterization of the trees for which it is possible
  to achieve a distance of~$2$ instead of~$3$.
  A slightly weaker result (applying to Hamiltonian cycles, corresponding in our
  setting to sequences that additionally have to start and end at vertices which are at
  distance $\leq 3$) is also given as Exercise~4.1.19~b in
  \cite{bondy2008graph}. We again give a self-contained proof for convenience:
  
  \begin{proof}[Proof of Claim~\ref{clm:treeopt}]
    Consider the tree pictured on Figure~\ref{fig:treeopt}, and some sequence
    $n_1, \ldots, n_{11}$ enumerating its nodes. Assume by contradiction that
    the sequence achieves a distance $\leq 2$, i.e., all pairs of consecutive
    vertices in the sequence is at distance $\leq 2$.
    Consider the vertices $1'$, $2'$, $3'$, $4'$, and $5'$ of the tree. The
    first and last position of the sequence may be covering at most two of these
    vertices, so there must be at least three of these values that occur at
    positions of the sequence which are not the first or last position. Without
    loss of generality, we assume up to symmetry that these are $1'$, $2'$, and
    $3'$, occurring in this order, and we write $l, m, r$ the positions where
    they occur in the sequence. In other words, we have
    $n_l = 1'$, $n_m = 2'$, and $n_r = 3'$, for some positions
    $2 \leq l < m < r \leq 10$.

    Now, given that $n_l = 1'$ and $n_m = 2'$ and the distance between $1'$ and
    $2'$ in the graph is four, we know that these values cannot be consecutive,
    i.e., we must have $m - l \geq 2$. For the same reason, we must have $r - m
    \geq 2$. From this, we deduce that the indices $l-1, l, l+1, r-1, r,
    r+1$ are pairwise distinct and are all indices in $\{1, \ldots, 11\}$.

    Now, as $n_l = 1'$, we know that the only possibility to respect the
    distance bound is that one of
    $n_{l-1}, n_{l+1}$ is~$0$ and the other is~$1$. Likewise, one of $n_{r-1}$
    and $n_{r+1}$ is~$0$ and the other is~$3$.
    As $0$ occurs only once and the indexes are pairwise
    distinct, this is a contradiction.
  \end{proof}

  Now, to justify that the constant~$3$ is optimal, not only in
  Claim~\ref{clm:treeopt}, but also in Lemma~\ref{lem:connected-enum}, it suffices to
  notice that the tree used in the proof (Figure~\ref{fig:treeopt}) can be
  realized in the proof, with the push-pop distance or push-pop-right distance. Indeed, consider for instance the language $L =
  \{\epsilon, a, aa, b, bb, c, cc, d, dd, e, ee\}$. This language is
  1-connected, and the graph connecting the nodes at edit distance~$1$ in the
  push-pop or (equivalently) the push-pop-right distance is exactly the tree of
  Figure~\ref{fig:treeopt}, so
  while it has a 3-ordering we know that it cannot have a 2-ordering. For
  the Levenshtein distance, we 
  replace $\epsilon$ by $xxxx$, replace $a$ by $axxxx$ and $a^2$ by $a^2 xxxx$,
  replace $b$ by $xbxxx$ and $b^2$ by $xb^2xxx$, and so on. The
  language is 1-connected, and the graph connecting the nodes at distance 1 in
  the Levenshtein distance is again exactly the tree of
  Figure~\ref{fig:treeopt}. This is because no edit can insert, remove, or
  substitute an~$x$ (it would not give a word of the language because the number
  of $x$-es would not be correct); clearly the tree root is connected to its
  children (by an insertion) and each
  child to its respective child (by an insertion again); and no other
  connections are possible (in particular substitutions applied to words of the
  language never give a word of the language).

  \subparagraph*{Proving the result.}
  We now prove Lemma~\ref{lem:tree}:
  \begin{proof}[Proof of Lemma~\ref{lem:tree}]
    We root~$T$ at the starting node~$s$, yielding a rooted tree where $s$ is
    the root and~$e$ is an arbitrary node different from~$s$. We call
    \emph{special} the nodes in the unique path from~$s$ to~$e$, including~$s$
    and~$e$. We order the children of each special node except~$e$ 
    to ensure that their unique special child is the last child.

    We first explain how to enumerate all nodes of~$T$ except~$e$ and its
    descendants, while obeying the distance requirement. 
    To do this, we will handle differently the nodes depending on (the parity
    of) their depth in~$T$. We will ensure that the enumeration starts at~$s$,
    and ends:
    \begin{itemize}
      \item If $e$ is at even depth, then we finish at the parent
        of~$e$.
      \item If $e$ is at odd depth and was the only child of its parent,
        then we finish at the parent of~$e$.
      \item If $e$ is at odd depth and its parent has other children,
        then we finish at a sibling of~$e$.
    \end{itemize}
    To do this, we start processing at the root, where processing a node~$n$
    means the following:
    \begin{itemize}
      \item For nodes $n$ at even depth (in particular~$n = s$), no matter
        whether they are special or not, we process
        them in prefix order: enumerate~$n$, then recursively processing their children in
        order.
      \item For non-special nodes~$n$ at odd depth, we process them in postfix
        order: recursively processing their children in order, then
        enumerate~$n$.
      \item For special nodes~$n$ at odd depth, we process them in a kind of infix order:
        recursively process their children in order except the last (if any), then
        enumerate~$n$, then recursively process their last child (which is again
        special, or is~$e$).
    \end{itemize}
    We finish this process when we recursively process~$e$. Then it is clear
    that the enumeration starts at~$s$, and visits all nodes of~$T$ except~$e$
    and its children. We check that it finishes at the right node:
    \begin{itemize}
      \item If $e$ is at even depth, then its parent $n$ was a node at odd depth which was
        special. We finished the enumeration by recursively processing~$n$,
        processing all other children of~$n$, then enumerating~$n$, so indeed we
        finish at the parent of~$n$.
      \item If~$e$ is at odd depth, then its parent $n$ was a node at even
        depth. We finished the enumeration by recursively processing~$n$. Then
        there are two cases:
        \begin{itemize}
          \item If $e$ was the only child of~$n$, then the processing of~$n$
            enumerated~$n$ and then recursed on~$e$, so the claim is correct.
          \item If $e$ was one of the children of~$n$, then it was the last
            child, so the processing of~$n$ enumerated~$n$, then recursively
            processed all its other children, before recursively processing~$e$.
            Now, the other children were non-special nodes at odd depth, so when
            we processed the last children, the enumeration finished by
            enumerating that other child, which is a sibling of~$e$ as claimed.
        \end{itemize}
    \end{itemize}

    We then check that this enumeration respects the distance condition, by
    checking that any two consecutive nodes in this enumeration are at distance
    less than~$3$.
    \begin{itemize}
      \item When we enumerate a node $n$ at even depth, then we have just
        started to process it; let us study what comes after~$n$.
        Either $n$ has children or it does not:
        \begin{itemize}
          \item If $n$ has no children, then~$n$ is not the root~$s$. Then the parent $n'$ of~$n$ is a node at odd depth.
            Either it is special or non-special. If $n'$ is non-special, then
            either $n$ was the last child or not. If $n$ was the last child, we
            next enumerate $n'$, so the distance is~$1$. If $n$ was not the last
            child, we next enumerate the next child of~$n'$, which is at even
            depth and is not~$e$, and the distance is~$2$. If $n'$ is special,
            then $n$ was not the last child because the last child of a special
            node different from~$e$ has children (as the special nodes are
the nodes in the path from the root to~$e$). Either $n$ was the penultimate
            child or not. If it is not the penultimate child, then we next
            enumerate the next child of~$n'$ and the distance is~$2$. If it is
            the penultimate child, then we next enumerate~$n'$ and the distance
            is~$1$.
          \item If $n$ has children, then if its only child is~$e$ we have
            finished. Otherwise, we will produce another node. Specifically, we next recurse on
            the first child $n'$, which is not~$e$, and either we first
            enumerate $n'$ (if it has no children, or is a special node whose
            only child is~$e$) and the distance is~$1$, or we enumerate the first child
            of $n'$ (which is at even depth) and the distance is~$2$.
        \end{itemize}
      \item When we enumerate a node~$n$ at odd depth which is non-special, then
        we have just finished processing it. Then its parent $n'$ was at even
        depth. Either $n$ was the last child of~$n'$ or not:
        \begin{itemize}
          \item If $n$ was the last child of~$n'$, we go back to the parent $n''$
        of~$n'$, which is at odd depth. If $n'$ was the last child of~$n''$,
        we enumerate~$n''$, at distance~$2$. If $n'$ was not the last child
            of~$n''$ but was its last non-special child, then~$n''$ itself is
            special and we next enumerate~$n''$, at
            distance~$2$ (before recursing in the special child). Last, if the
            next sibling of~$n'$ is non-special, then from~$n''$
        we recurse into it (it is at even depth) and enumerate it, for a
        distance of~$3$.
      \item If $n$ was not the last child of~$n'$, then we next recurse in the
        next child~$n''$, which is at odd depth. If $n''$ has no other children,
            or if its only child is special, we next produce~$n''$ for a
            distance of~$2$. Otherwise we recurse in the first child of~$n''$,
            which is at even depth, and produce it for a distance of~$3$.
        \end{itemize}
      \item When we enumerate a node~$n$ at odd depth which is special, then we
        are about to recurse in its special child. Either it is~$e$ and we have
        finished, or it is a node at even depth and we enumerate it for a
        distance of~$1$.
    \end{itemize}
    Hence, in all cases the distance within that first part of the enumeration
    is at most~$3$.

    \bigskip
    We now explain how to enumerate the remaining nodes, i.e., the subtree of
    $T$ rooted at~$e$, including~$e$. For this, we consider the depth of the
    nodes from the parent of~$e$, i.e., $e$ is now considered at odd depth, its
    children (if they exist) are at even depth, and so on. We re-do the
    previously described enumeration scheme on that subtree, except that there
    are no special nodes. It is clear that this defines an enumeration sequence
    which visits the entire subtree, ends at $e$ (as it is at odd depth), and
    respects the distance bound from the previous proof. (More specifically,
    considering a subtree of~$T$ rooted at a non-special node~$n$ at odd depth
    in the previous proof, when we started processing that subtree, the
    enumeration clearly covered all its nodes, finished at~$n$, respected the
    distance bound, and there were no special nodes in that subtree, so this
    shows that correctness also holds when applying that scheme on the subtree
    rooted at~$e$.)

    Hence, the last point to show is that the distance between the last node of
    the first part of the enumeration and the first node of the second part of
    the enumeration is at most~$3$. Now, the first part ended either at the
    parent of~$e$ or a sibling of~$e$, i.e., at distance~$2$ to~$e$. Now, the
    second phase of the enumeration starts on~$e$ at odd depth, so either $e$
    was in fact a leaf and we enumerate it for a distance~$2$, or we recurse on
    the first child of~$e$, which is now at even depth, and enumerate it, for a
    total distance of~$3$. Hence, the entire enumeration respects the distance
    bound of~$3$. This concludes the proof.
  \end{proof}

  Last, we can easily prove Lemma~\ref{lem:connected-enum} from
  Lemma~\ref{lem:tree}:
\begin{proof}[Proof of Lemma~\ref{lem:connected-enum}]
  Consider the (finite) undirected graph~$G = (L,\{\{u,v\}\mid
\delta(u,v)\leq d\}$. Since~$L$ is $d$-connected,~$G$ is connected (in the usual sense) and so it has a
spanning tree~$T$.
  Now, applying Lemma~\ref{lem:tree} on~$T$ gives us a sequence enumerating
  exactly once every vertex of~$G$ starting and ending at the requisite nodes,
  such that the distance between any two nodes in the sequence is at most~$3$
  in~$T$, hence at most~$3$ in~$G$, so the distance between the words is at
  most~$3d$.
\end{proof}
\end{toappendix}

\subparagraph*{Stratum-connectivity.}
To show orderability for infinite languages, we will decompose them into 
\emph{strata}, which simply contain the words in a
certain length range. Formally:

\begin{definition}
  Let $L$ be a language, let $\ell > 0$ be an integer,
  and let $i > 0$. The \emph{$i$-th stratum of width~$\ell$} (or
  \emph{$\ell$-stratum}) of~$L$,
  written $\stratum_\ell(L,i)$, is $L^{< i\ell} \setminus L^{<(i-1)\ell}$.
\end{definition}

We will show that, for the language $\L(A)$ of our interchangeable DFA~$A$,
we can pick~$\ell$ and~$d$ such that every $\ell$-stratum of~$\L(A)$
is $d$-connected, i.e., $\L(A)$
is  \emph{$(\ell,d)$-stratum-connected}:

\begin{definition}
  Let $L$ be a regular language and fix $\ell, d > 0$.
  We say that $L$ is \emph{$(\ell,d)$-stratum-connected} if every~$\ell$-stratum
  $\stratum_\ell(L, i)$ is $d$-connected.
\end{definition}

Note that our example language $a^* + b^*$, while $1$-connected, is not
$(\ell,d)$-stratum-connected
for any $\ell,d$, because any~$i$-th~$\ell$-stratum for~$i>d$ is not
$d$-connected. 
We easily show that stratum-connectivity implies
orderability:

\begin{toappendix}
  \subsection{Proof of Lemma~\ref{lem:scord}}
\end{toappendix}
\begin{lemmarep}
\label{lem:scord}
  Let $L$ be an infinite language recognized by a DFA with $k'$ states, and assume that $L$
  is $(\ell,d)$-stratum-connected for some $\ell \geq 2k'$ and some $d \geq 3k'$. Then $L$ is $3d$-orderable.
\end{lemmarep}

\begin{proofsketch}
  We show by pumping that we can move across contiguous strata. Thus, we combine orderings on each stratum obtained by Lemma~\ref{lem:connected-enum} with well-chosen endpoints.
\end{proofsketch}

\begin{toappendix}
\begin{proof}
  A \emph{$(\ell,d)$-ladder} of~$L$ consists of two infinite sequences $e_1,
  \ldots, e_n, \ldots$ of \emph{exit points}, and $s_2, \ldots, s_n, \ldots$ of
  \emph{starting points}, such that  $e_i \in \stratum_\ell(L,i)$ for all $i \geq 1$ and 
  $s_{i} \in \stratum_\ell(L,i)$ for all $i \geq 2$, and such that $s_i \neq
  e_i$ for all $i \geq 2$ and $\delta_\pp(e_i, s_{i+1}) \leq d$.
  We will show that such a ladder exists, which will suffice to show that~$L$ is~$3d$-orderable.
  Indeed, because each~$(\ell,d)$-stratum is~$d$-connected, we know by~Lemma~\ref{lem:connected-enum} that there exists a~$3d$-ordering~$\mathbf{s}_1$ of~$\stratum_\ell(L,1)$ that ends at~$e_1$ (and starts at an arbitrary word of~$\stratum_\ell(L,1)$), and also, for~$i\geq 2$, that there exists a~$3d$-ordering~$\mathbf{s}_i$ of~$\stratum_\ell(L,i)$ that starts at~$s_i$ and ends at~$e_i$. Now, the order~$\mathbf{s}_0 \mathbf{s}_1 \ldots$ is clearly a~$3d$-order of~$L$.
  
  Hence, let us show that a ladder exists. As $L$ is infinite, we know that the
  DFA~$A$ has a loopable state~$q$. Let $rz$ be a word of~$q$ such that $z$ is
  the label of a simple loop on a loopable state on~$q$: by the pigeonhole
  principle, we can ensure $|rz| \leq k'$. Now, using the fact that the DFA is
  trimmed, let $t$ be a word such that $rt$ is accepted: we can ensure $|t| <
  k'$.
  Now consider the
  sequence of words $w_i \colonequals r z^i t$ for all $i \geq 0$, which are all
  accepted.
  By definition
  the length difference between $w_i$ and $w_{i+1}$ is at most~$k'$ for each $i
  \geq 0$, and the push-pop distance between them is at most $3k'$, as evidenced by popping right $t$ and pushing right $zt$; so it is~$\leq d$.
  
  We now observe that, as $\ell \geq 2k'$, each stratum contains at least two
  distinct $w_i$: in particular the $rt$ and $rzt$ are words of the first
  stratum because their length is strictly less than~$2k'$.
  For each stratum $i \geq 1$, we choose $s_i$ to be the shortest $w_j$ such that $w_j$ is in $\stratum_\ell(L,i)$, and choose~$e_i$ to be the longest such~$w_j$. Thus, $e_i \neq s_i$ for all $i \geq 2$, and $\delta_\pp(e_i, s_{i+1}) \leq d$ for all~$i \geq 1$. We have thus defined a ladder, which as we argued concludes the proof.
\end{proof}
\end{toappendix}

We can then show using several pumping and de-pumping arguments
that the language of our interchangeable DFA~$A$ is $(\ell,d)$-stratum-connected for  $\ell \colonequals 8k^2$ and~$d \colonequals 16k^2$.

\begin{toappendix}
  \subsection{Proof of Proposition~\ref{prp:stratumconnected}}
We now prove the main technical result of this section, namely:
\end{toappendix}

\begin{propositionrep}
\label{prp:stratumconnected}
  The language
  $\L(A)$ is $(8k^2,16k^2)$-stratum-connected.
\end{propositionrep}

\begin{proofsketch}
  As there are only a finite number of non-loopable words, we focus on loopable
  words.
  Consider a stratum $S$ and two loopable words $u$ and~$v$ of~$S$.
  Their accepting runs involve loopable states, respectively~$q$ and~$q'$, that are 
  interchangeable because $A$ is.
  We first show that $u$ is $d$-connected (in~$S$) to a \emph{normal form}: a
  repeated loop on~$q$ plus a prefix and suffix whose length is bounded, i.e., only depends on the
  language. We impose this in two steps: first we move the last occurrence of $q$ in~$u$ near the end of the word by pumping at the left end and de-pumping at the right end, second we pump the loop on~$q$ at the right end while de-pumping the left end. This can be done while remaining in the stratum~$S$.
We obtain similarly a normal form consisting of a repeated loop on~$q'$ with
  bounded-length prefix and suffix that is~$d$-connected to~$v$ in~$S$.
  
  Then we do an induction on the number of connectivity and compatibility
  relations needed to witness that $q$ and~$q'$ are interchangeable. If $q =
  q'$, we conclude using the normal forms of~$u$ and~$v$. If $q$ is connected to~$q'$, we impose the normal form on~$u$, then we modify it to a word whose accepting run also visits~$q'$, and we apply the previous case. If $q$ is compatible with~$q'$, we conclude using the normal form with some loop label $z$ in $A_q \cap A_{q'}$ (of length~$\leq k^2$) that witnesses their compatibility.
  The induction case is then easy.
\end{proofsketch}

\begin{toappendix}

We let~$(\ell,d)\colonequals (8k^2,16k^2)$.
  We show this result in the remainder of the appendix.
  If $\L(A)$ is finite, then clearly it consists only of words of length $\leq k$, so the first stratum is trivially $d$-connected because $d \geq 2k$ and the other strata are empty, hence vacuously $d$-connected. Thus, in the sequel, we assume that $\L(A)$ is infinite, in particular it contains infinitely many loopable words.
  
  To show that $\L(A)$ is $d$-stratum connected, take $i \geq 1$ and let us show that the $i$-th $\ell$-stratum $S = \stratum_\ell(\L(A), i)$ is $d$-connected. 
This will indeed be enough by Lemma~\ref{lem:scord}.
In the rest of the proof we only work with the language~$S$, so $d$-sequences, $d$-connectivity between words, always require that the 
words of the sequences 
are in~$S$.
  
  Let us first take care of the non-loopable words,
  and show that each non-loopable word is at
  distance $\leq d$ from a loopable word of~$S$. For this, taking a non-loopable word~$w$, we know from the proof of Claim~\ref{clm:langpart} that $|w| \leq k-1$.
  From the definition of $\ell$ (specifically, since $\ell \geq k-1$), we know that we are in the first stratum, i.e., $i = 1$.
  As $d \geq 2k$, we can simply edit~$w$ into a loopable word by first removing all letters of~$w$, then adding all letters of some loopable word of length at most $k-1$, which must exist because $\L(A)$ is infinite and we can simply choose any word whose accepting path goes via some loopable state and does not go twice through the same state.
  
  We will now show the result statement: any two words $u$ and $v$ of~$S$ are $d$-connected. By what we showed, we can assume that $u$ and $v$ are loopable. 
 
  It will often be necessary to adjust our constructions depending on the length of a word~$w \in S$, because when modifying $w$ we wish to ensure that it remains in~$S$, in particular that it has the right length. So we will call a word $w \in S$ \emph{short} if it is in the lower half of the stratum, and \emph{long} otherwise. 
  Formally, considering a word $w'$ of~$S$, we know that its length is $(i-1) \times 8k^2 \leq |w'| < i \times 8k^2$,
  i.e., 
  $(2i-2) \times 4k^2 \leq |w'| < 2i \times 4k^2$: we call \emph{valid} a word satisfying this condition on length, so that a word is in~$S$ iff it is valid and is accepted by~$A$. Now, we 
  call $w'$ \emph{short} if
  $(2i-2) \times 4k^2 \leq |w'| < (2i-1) \times 4k^2$ and \emph{long}
  if 
  $(2i-1) \times 4k^2 \leq |w'| < 2i \times 4k^2$. This ensures that increasing the length of a short word by at most~$k^2$, or by at most~$4k$, gives a valid word (which may be long or short); and likewise decreasing the length of a long word by at most~$k^2$ or by at most~$4k$ gives a valid word (which may also be long or short).
  Further the definition ensures that $\ell \geq 7k$, which we will use at some point as a bound on the size of the first stratum.

  We will first show that all loopable words of~$S$ are $d$-connected to
  loopable words of~$S$ of a specific form, where, up to a prefix and suffix
  whose length only depends on the language, the words consists of a power of a word which labels a loop on some loopable state (on which we must impose a length bound). We will apply this to loops of length at most~$k^2$, because of our choice of~$\ell$.
  
  \begin{lemma}
    \label{lem:bigloop}
    Let $w$ be a loopable word of~$S$, let $q$ be a loopable state that occurs in the run of~$w$, let $e$ be the label of a loop on~$q$, i.e., a non-empty word of $\L(A_q)$, such that $|e| \leq k^2$. Then $w$ is $d$-connected to some word $s e^n t$ with $|s| \leq k$ and $|t| \leq k$.
  \end{lemma}
  
  To show the lemma, we will first show that we can ``move'' the latest occurrence of~$q$ near the end of the word by making edits of length at most~$2k$. To this end, we define the \emph{$q$-measure} of a word $w'$ and a state $q$ occurring in~$w'$ as the smallest possible length of~$t'$ when writing $w' = r' t'$ such that $q$ is reached between~$r'$ and~$t'$. We then claim:
  
  \begin{claim}
  \label{clm:moveloop}
    Any loopable word $w$ of~$S$ where some loopable state $q$ occurs is $2k$-connected to some word~$w'$ of $q$-measure $\leq k$.
  \end{claim}
  
  \begin{proof}
    We show the claim by induction on the measure of~$w$. The base case is when the measure is $\leq k$, in which case the result is immediate. So let us show the induction claim: given $w$ with measure $> k$, we will show that it is $2k$-connected to some word (in~$S$) with strictly smaller measure. Note that we can assume that $|w| \geq k$, as otherwise its measure is~$<k$.
    
    Intuitively, we want to de-pump a simple loop in the suffix of length~$k$ to decrease the measure, but if the word is short we first want to make it long by pumping some loop as close to the beginning as possible. So the first step is to make~$w$ long if it is short, and then the second step is to decrease the measure (potentially making the word short again) and conclude by induction hypothesis.

    The first step works by repeating the following process, which does not increase the measure. Formally, as $w$ contains a loopable state and the non-loopable states occur at most once, we can write $w = \rho \tau$ with the length $|\rho|$ of $\rho$ being minimal (in particular $\leq k$) such that the state between~$\rho$ and~$\tau$ is loopable. Let $\sigma$ be a simple loop on~$q$, i.e., $\sigma \in \L(A_q) \setminus \{\epsilon\}$. Observe that then~$|\rho\sigma|\leq k$. By pumping, we know that $\rho \sigma \tau$ is accepted by~$A$. We obtain $\rho \sigma \tau$ by editing~$w = \rho\tau$ in the following way: we pop left elements of $\rho$, i.e., at most~$k$ edits, then push~$\rho\sigma$, i.e., at most~$k$ edits. Note that, writing the original $w$ as $s' t'$ with $q$ reached between $s'$ and $t'$ and $|s'|$ being maximal, the modification described here modifies a prefix of the word which is no longer than $s'$, and makes $s'$ longer (thus making $w$ longer). Repeating this process thus gives us a word $w'$ which is no longer short, and where the measure is unchanged because $\sigma$ was added on the first occurrence of a loopable state, hence on or before the last occurrence of~$q$. As 
    we enlarge the word at each step by at most $|\sigma| \leq k$, and stop it when the word is no longer short, the definition of being short ensures that the eventual result $w'$ of this process is not too big, i.e., $w'$ is still valid; and $w$ is $2k$-connected to it.
    
    Now, the second step on a word~$w'$ which is long is to make the measure decrease. For this, we simply write $w' = r' t'$ with $|t'|$ the measure of~$w'$. If $|t'| \leq k$, we conclude by the base case. Otherwise, by the pigeonhole principle, we can find two occurrences of the same state in the run within the suffix~$t'$; we write $w' = \rho \sigma \tau$ with $|\sigma\tau| \leq k$ as small as possible (i.e., we take the first occurrence of a repeated state when reading the run from right to left) such that some state $q'$ occurs between $\rho$ and $\sigma$ and between $\sigma$ and $\tau$. Note that $q' \neq q$, otherwise we would have concluded by the base case. By depumping, we know that $\rho\tau$ is accepted by the automaton, and as $q\neq q'$ it has strictly smaller $q$-measure. We obtain $\rho\tau$ from~$w' = \rho\sigma\tau$ in a similar way to that of the previous paragraph: we pop right elements of~$\sigma\tau$, i.e., at most~$k$ edits, then we push right the elements of~$\tau$, i.e., at most~$k$ edits. The length of the resulting word is decreased by~$\sigma$, i.e., by at most~$k$, so the definition of being long ensures that the resulting word is in~$S$.
    Thus, we know that $w'$, hence~$w$ is $2k$-connected to a word of~$S$ with smaller $q$-measure. We conclude by the induction hypothesis that $w$ is $2k$-connected to a word with $q$-measure $\leq k$, which concludes the proof.
  \end{proof}
  
  With this claim, we can show Lemma~\ref{lem:bigloop}:
  
  \begin{proof}[Proof of Lemma~\ref{lem:bigloop}]
    Fix the word $w$ and the loopable state~$q$. By Claim~\ref{clm:moveloop}, we know that $w$ is $2k$-connected to a word $w'$ of~$S$ with $w' = r't'$, $|t'| \leq k$, and $q$ is achieved between~$r'$ and~$t'$. Fix $e$ the label of the loop to achieve. We will intuitively repeat two operations, starting with the word~$w'$, depending on whether the current word is long or short, with each word obtained being a word of~$S$ that is at distance $\leq d$ from the previous one. We stop the process as soon as we obtain a word of the desired form.
    
    If the current word is short, we add a copy of~$e$ just before~$t'$, by
popping~$t'$ and then pushing $et'$. This takes $2k+|e|$ operations, which is $\leq d$ because
$|e| \leq k^2$, and we know that the result is accepted by the automaton.
Further, as the length increased by~$|e|$ and $|e| \leq k^2$, the definition of
$\ell$ ensures that the word is still in~$S$ (but it may be long).
    
    If the current word is long, we make it shorter by editing the left endpoint. Formally, as a long word has length $\geq k$, we can write $w' = \rho \tau$ with $|\rho|\leq k$, and furthermore $\rho$ does not overlap the copies of~$e$ that have been inserted so far, otherwise the word would already be of the desired form. By the pigeonhole principle, there is a state $q'$ occurring twice in~$\rho$, i.e., we can write $\rho = \rho' \sigma \rho''$, such that $\rho' \rho'' \tau$ is accepted by the automaton,
and by taking~$q'$ to be the first such state we can have~$|\rho'\sigma|\leq k$. We can obtain $\rho' \rho'' \tau$ from~$w'$ by popping left $\rho' \sigma$ and pushing left $\rho'$, i.e., at most $2k$ operations. Further, as the length decreased by at most~$k$, the definition of~$\ell$ ensures that the word is still in~$S$ (but it may be short).
    
    It is clear that repeating this process enlarges a factor of the form $e^i$ inside the word (specifically, the word ends with $e^i t$ with $i$ increasing), while the word remains in~$S$, so we will eventually obtain a word of~$S$ of the form $s e^n t$ with $|s| \leq k$ and $|t| \leq k$ to which the word~$w'$, hence the original word~$w$, is $d$-connected. This establishes the claimed result.
  \end{proof}
  
  Now that we have established Lemma~\ref{lem:bigloop}, we have a kind of ``normal form'' for words, i.e., we will always be able to enforce that form on the words that we consider.
  Let $u$ and $v$ be two loopable words, and let $p$ and $q$ be any  loopable states reached in the accepting run of $u$ and $v$ respectively.
  We know that the automaton $A$ is interchangeable, so~$p$ and $q$ are interchangeable, and there exists a sequence $p = q_0, \ldots, q_h = q$ witnessing this, with any two successive $q_i$ and $q_{i+1}$ for $1 \leq i < h$ being either connected or compatible. We show that $u$ and $v$ are $d$-connected by induction on the value~$h$.
 
  \subparagraph*{Base case $h=0$.}  
  The base case is $h = 0$, i.e., $p = q$. In this case, let $e$ be an arbitrary non-empty word in $A_q$ of length $\leq k$. We know by applying Lemma~\ref{lem:bigloop} to~$u$ that $u$ is $d$-connected to a word of the form $u' = s e^n t$ with $|s| \leq k$ and $|t| \leq k$, and by applying the lemma to~$v$ we know that $v$ is $d$-connected to a word of the form $v' = x e^m y$ with $|x| \leq k$ and $|y|\leq k$.  We now claim the following, to be reused later:
  
  \begin{claim}
  \label{clm:base0}
    Consider words of~$S$ of the form $u' = s u'' t$ and $v' = x v'' y$ with $|s| \leq k$ and $|t| \leq k$ and $|x| \leq k$ and $|y|\leq k$ and $u''$ and $v''$ both being powers of some word $\lambda$. 
    Then $\delta_\pp(u', v') \leq \ell+8k$.
  \end{claim}
  
  \begin{proof}
    First note that, because $u''$ and $v''$ are powers of a common word, we have $\delta_\pp(u'', v'') = \left||u''| - |v''|\right|$, simply by making their length equal by adding or removing powers of $\lambda$.
    Now, the distance from $u'$ to $v'$ is at most $|s| + |t| + |x| + |y| + \left||u''| - |v''|\right|$, by popping $s$ and $t$, then making the length of the middle parts $u''$ and $v''$ equal via the previous observation, then pushing $x$ and $y$.
    As the two words $u'$ and $v'$ are in~$S$,
    their length difference is at most $\ell$, i.e., $\left||u'|-|v'|\right| \leq \ell$.
    Now $|u'| = |s| + |t| + |u''|$ and $|v'| = |x| + |y| + |v''|$, so by the triangle inequality
    the length difference $\left||u''| - |v''|\right|$ between the middle parts $u''$ and $v''$ is at most 
    $\ell + |s|+|t|+|x|+|y|$, so $\delta_\pp(u', v') \leq \ell + 2(|s|+|t|+|x|+|y|)$,
    i.e., at most $\ell+8k$, establishing the result.
  \end{proof}
    
    Now, applying Claim~\ref{clm:base0} to our $u'$ and $v'$ with $\lambda = e$, we know that $\delta_\pp(u', v') \leq \ell+8k$, so the distance is at most $d$ because $d \geq \ell+8k$. This shows that $u'$ and $v'$, hence $u$ and $v$, are $d$-connected, establishing the base case $h =0$.
  
  \subparagraph*{Base case $\bm{h=1}$.}  
  We show a second base case, with $h = 1$, which will make the induction case trivial afterwards. In this case, $u$ and $v$ respectively contain loopable states $q$ and $q'$ which are either connected or compatible. We deal with each case separately.
  
  \subparagraph*{Base case $\bm{h=1}$ with two connected states.}  
  if $q$ and $q'$ are connected, it means that there is path from~$q$ to~$q'$ in the automaton, or vice-versa. We assume that there is a path from~$q$ to~$q'$, otherwise the argument is symmetric up to exchanging~$u$ and~$v$. Up to making this path simple, let $\pi$ be the label of such a path in the automaton, with $|\pi| \leq k$, and let $\pi'$ be the label of a simple path in the automaton from~$q'$ to some final state; hence $|\pi\pi'| \leq 2k$. 
  Let $e$ be a non-empty word of length $\leq k$ accepted by~$A_q$.
  We know by Lemma~\ref{lem:bigloop} that $u$ is $d$-connected to a word of the form $u' = s e^n t$ 
  with $|s| \leq k$ and $|t| \leq k$ and $q$ occurring before and after each occurrence of~$e$. Intuitively, we want to replace $t$ by $\pi\pi'$, to obtain a word to which $u'$ is $d$-connected and which goes through~$q'$, so that we can apply the first base case $h = 0$, but the subtlety is that doing this could make us leave the stratum~$S$. We adjust for this by adding or removing occurrences of~$e$ to achieve the right length. Formally, if $u' \leq 5k$ then because $\ell \geq 5k$ we know that~$u'$ is in the first stratum and that $s e^n \pi \pi'$ has length at most $7k$, so as $\ell \geq 7k$ and $d \geq 3k$ we know that $u'$ is $d$-connected to $u'' = s e^n \pi \pi'$ which is in~$S$. Otherwise, we assume that $u' \geq 5k$, so that $e^n \geq 3k$. There are three subcases: either $|t| = |\pi\pi'|$, or $|t| < |\pi\pi'|$, or $|t| > |\pi\pi'|$.
  
  In the first subcase where $|t| = |\pi\pi'|$, then we know that $u'$ is $3k$-connected to the word $u'' = s e^n \pi\pi'$, which is accepted by~$A$ and is in $S$ because $|u'| =|u''|$. 
  In the second subcase where $|t| < |\pi\pi'|$, if $u'$ is short then we conclude like in the previous subcase, using the fact that $|u'| \leq |u''| \leq |u'|+2k$ so $u''$ is in~$S$ by the definition of~$\ell$. If $u'$ is long then
  we can choose some number $\eta$ such that $2k \leq |e^\eta| \leq 3k$, which is possible because $|e| \leq k$. Now we know that from $u'$ we can obtain the word $s e^{n-\eta} \pi\pi'$, which is well-defined because $e^n \geq 3k$ so $n \geq \eta$, is accepted by~$A$, and has length at most~$|u'|$ and at least $|u'| - |e^{\eta}|$, i.e., length at least $|u'| - 3k$, so it is still in~$S$ by the definition of~$\ell$.
  In the third subcase where $|t| > |\pi\pi'|$, if~$u'$ is long then we conclude like in the first subcase because $|u'| - k \leq |u''| \leq |u'|$, so $u''$ is in~$S$. If $u'$ is short then we choose some number $\eta$ like in the second subcase, and obtain from~$u'$ the word $s e^{n+\eta} \pi\pi'$, which is accepted by~$A$, is longer than~$u'$, and has length at most $|u'| + 3k$ so is still in~$S$.
  
  In all three subcases we have shown that $u'$ is $d$-connected to a word $u''$ in~$S$ whose accepting path goes through~$q'$ because $u''$ finishes by $\pi\pi'$ and the state immediately before was~$q$ so~$\pi$ brings us to~$q'$. By the base case $h = 0$, we know that $u''$ and $v$, whose accepting runs both include the state~$q'$, are $d$-connected.
  
  \subparagraph*{Base case $\bm{h=1}$ with two compatible states.}  
  Second, if $q$ and $q'$ are compatible, we know that there is some non-empty witnessing word $z$ which is both in~$A_q$ and in~$A_{q'}$. Up to taking a simple path in the product of these two automata, we can choose $z$ to have length~$\leq k^2$. By Lemma~\ref{lem:bigloop}, we know that $u$ is $d$-connected to a word $u' = s z^n t$ with $|s| \leq k$ and $|t|\leq k$ and $q$ at the beginning and end of every occurrence of~$z$. Likewise, applying the lemma to~$v$, we know that $v$ is $d$-connected to a word $v' = x z^m y$ with $|x| \leq k$ and $|y| \leq k$ and $q$ at the beginning and end of every occurrence of~$z$. We now conclude like in the base case $h = 0$, by applying Claim~\ref{clm:base0} with $\lambda = z$.
  
  \subparagraph*{Induction case.} Let $u$ and $v$ be two loopable words, and $p$ and $q$ be loopable states respectively occurring in the accepting run of~$u$ and of~$v$, and consider a sequence $p = q_0, \ldots, q_h = q$ witnessing this, with $h \geq 1$.
  Let $w$ be any word of~$S$ whose accepting path goes via the state $q_{h-1}$, which must exist, e.g., taking some path in the automaton that goes via $q_{h-1}$ and then repeating some loop of size $\leq k$ on the loopable state~$q_{h-1}$. 
  By applying the induction hypothesis, we know that $u$ is $d$-connected to~$w$. By applying the case $h = 1$, we know that $w$ is $d$-connected to~$v$. Thus, $u$ is $d$-connected to~$v$. 
  This shows the claim, and concludes the induction, proving Proposition~\ref{prp:stratumconnected}.

\end{toappendix}

From this, we deduce with Lemma~\ref{lem:scord} that $\L(A)$ is
$48k^2$-orderable, so Theorem~\ref{thm:orderable} holds. Note that the
construction ensures that the words are ordered stratum after stratum, so
``almost'' by increasing length: in the ordering that we obtain, after producing some
word~$w$, we will never produce words of length less than $|w| - \ell$.

\section{Bounded-delay enumeration}
\label{sec:worddag}
In this section, we show how the orderability result of the previous section
yields a bounded-delay algorithm.
We use the pointer-machine model from Section~\ref{sec:prelims}, which we modify
for convenience to allow data values and the number of fields of records to be
exponential in the automaton (but fixed throughout the enumeration, and
independent on the size of words): see Appendix~\ref{apx:machineword} for
explanations on why we can do this.
We show:

\begin{toappendix}
  \subsection{Details on the machine model}
  \label{apx:machineword}

  Throughout Section~\ref{sec:worddag}, we phrase our algorithms in a slight
  variant of the pointer machine model presented in Section~\ref{sec:prelims}
  and Appendix~\ref{apx:machine}. The modification is that the number of fields
  of records, and the number of possible data values, is no longer constant, but
  is allowed to depend on the input automaton, specifically it can be
  exponential in the input automaton. This allows us to talk, e.g., of records
  whose fields are indexed by edit scripts of constant length, or of data values
  denoting integers that are exponential in the automaton; and to allow
  arbitrary arithmetics on these. However, crucially,
  the number of fields and the data values do \emph{not} depend on the current
  length of words when we enumerate, or on the current memory size. Intuitively,
  as our enumeration algorithm generally runs indefinitely, the unbounded memory
  size will eventually be much larger than the exponential quantity on the input
  used to bound the number of fields and the number of data values; and our use
  of a pointer machine model means that the resulting pointers cannot be used
  for arithmetic operations.

  Let us now explain why this difference in our definition of the
  pointer machine model is in fact inessential. For any number $x$, it is easy to see that a pointer machine allowing $x$ fields per
  pointer and $x$ different data values can be emulated by a pointer machine in
  the standard sense presented in Section~\ref{sec:prelims}, e.g., which allows
  only 2 fields per pointer and 2 different data values. Indeed, the $x$
  different data values can be coded, e.g., in unary as linked lists of length
  up to $x$, and records with $x$ fields can be coded as a linked list of
  records, each of which contains one data field and one continuation field
  pointing to the next record. Operations on data values can be emulated by a
  hardcoded program that allocates a new linked list whose length is a function
  of the length of the two other lists, and accessing a field of a record can be
  emulated by a hardcoded program that follows the linked lists of records to
  retrieve the correct field. The overhead induced by the emulation is a
  multiplicative factor in the number $x$.

  Now, recall that our goal in Section~\ref{sec:worddag} is to show
  Theorem~\ref{thm:enum-main}, which claims an exponential delay in the input
  automaton. Thus, we can work in the modified pointer machine model where the
  quantity $x$ is exponential in the input automaton, because the resulting
  algorithm can be emulated by an algorithm on a pointer machine model in the
  standard sense, multiplying only by an exponential factor in the automaton,
  and yielding an overall complexity which is still exponential in the
  automaton.
\end{toappendix}

\begin{theorem}
\label{thm:enum-main}
There is an algorithm which, given an
interchangeable DFA $A$ with $k$ states, enumerates the language
  $\L(A)$ with push-pop distance bound~$48k^2$ and exponential delay in~$|A|$.
\end{theorem}

Let us accordingly fix the interchangeable DFA~$A$ with~$k$ states.
Following
Proposition~\ref{prp:stratumconnected}, we 
let~$d\colonequals 16k^2$ and~$\ell \colonequals 8k^2$.

\subparagraph*{Overall amortized scheme.}
The algorithm will consider the strata of the input language~$L$ and will
run two processes in parallel: the first process simply enumerates a previously prepared sequence of edit scripts that gives a~$3d$-ordering of some stratum, while the second process computes the sequences for subsequent strata (and of course imposing that the endpoints of the sequences for contiguous strata are sufficiently close).

The challenging part is to prepare efficiently the sequences for all strata, and in particular to build a data structure that represents the strata.
We will require of our algorithm that it processes each stratum
in \emph{amortized linear time} in its size.
Formally, letting $N_j \colonequals |\stratum_\ell(L, j)|$ be the number of words of the $j$-th stratum for all $j\geq 1$, there is a value $C \in \NN$ that is exponential in~$|A|$ such that, after having run for $C \sum_{j=1}^i N_j$ steps, the algorithm is done processing the $i$-th stratum.
Note that this is weaker than processing each separate stratum in linear time: the algorithm can go faster to process some strata and spend this spared time later so that some later strata are processed arbitrarily slowly relative to their size.

If we can achieve amortized linear time, then the overall algorithm runs with bounded delay. To see why, notice that the prepared sequence for the $i$-th stratum has length at least its size~$N_i$, and we can show that the size $N_{i+1}$ of the next stratum is within a factor of~$N_i$ that only depends on~$L$
(this actually holds for any infinite regular language and does not use
interchangeability):

\begin{lemma}
\label{lem:density}
  Letting $C_A \colonequals (k+1)|\Sigma|^{\ell+k+1}$, for all $i \geq 1$ we have
  $N_i /C_A \leq N_{i+1} \leq C_A N_i$.
\end{lemma}

\begin{proof}
  Each word in the $(i+1)$-th stratum of~$L$ can be transformed into a word in
  the~$i$-th stratum as follows: letting $k$ be the number of DFA states, first remove a 
prefix 
  of length at most $\ell+k$ 
  to get a word (not necessarily in~$L$) of length $i\ell -k -1$, and then add back a prefix corresponding to some path of length $\leq k$ from the initial state to get a word in the $i$-th stratum of~$L$ as desired. Now, for any word $w$ of the~$i$-th stratum, the number of words of the $(i+1)$-th stratum that lead to~$w$ in this way is bounded by~$C_A$, by considering the reverse of this rewriting, i.e., all possible ways to rewrite~$w$ by removing a prefix of length at most~$k$ and then adding a prefix of length at most $\ell+k$.
  A simple union bound gives $N_{i+1} \leq C_A N_i$. Now, a similar argument in
  the other direction gives
  $N_i /C_A \leq N_{i+1}$.
\end{proof}

Thanks to this lemma, it suffices to argue that we can process the strata in amortized linear time, preparing $3d$-orderings for each stratum: enumerating these orderings in parallel with the first process thus guarantees 
(non-amortized) bounded delay.

\subparagraph*{Preparing the enumeration sequence.}
We now explain in more detail the working of the amortized linear time algorithm. The algorithm consists of two components. The first component runs in amortized linear time over the successive strata, and prepares a sequence $\Gamma_1, \Gamma_2, \ldots$ of concise graph representations of each stratum, called \emph{stratum graphs};
for each~$i \geq 1$, after $C \sum_{j=1}^i N_j$ computation steps, it has finished preparing the $i$-th stratum graph $\Gamma_i$ in the sequence. The second component will run as soon as some stratum graph $\Gamma_i$ is finished: it reads the graph $\Gamma_i$ and computes a~$3d$-ordering for $\stratum_\ell(L, i)$ in (non-amortized) linear-time, using Lemma~\ref{lem:connected-enum}.
Let us formalize the notion of a stratum graph:

\begin{definition}
  Let $\Delta$ be the set of all push-pop edit scripts of length
  at most~$d$; note that $|\Delta| \leq (2 |\Sigma| + 2)^{d+1}$, and this bound only depends on the alphabet and on~$d$.
  For $i\geq 1$, the
  \emph{$i$-th stratum graph} is the edge-labeled directed graph $\Gamma_i = (V_i, \eta_i)$
  where the nodes
  $V_i=\{v_w \mid w \in \stratum_\ell(L,i)\}$ correspond to words of the~$i$-th stratum,
  and the directed (labeled) edges are given by the
  function $\eta_i \colon V_i \times \Delta \to V_i \cup \{\bot\}$ and describe the possible scripts:
  for each $v_w \in V_i$ and each $s \in \Delta$,
  if the script $s$ is applicable to~$w$ and the resulting word~$w'$ is in~$\stratum_\ell(L,i)$ then~$\eta(v_w,s) = v_{w'}$, otherwise $\eta(v_w,s)
  =\bot$.

  In our machine model, each node $v_w$ of $\Gamma_i$ is a record with
  $|\Delta|$ pointers, i.e., we do not store the word~$w$. 
  Hence, $\Gamma_i$ has linear size in~$N_i$.

  A  \emph{stratum graph sequence} is an infinite sequence $(\Gamma_1, v_{s_1}, v_{e_1}), (\Gamma_2, v_{s_2}, v_{e_2}), \ldots$ consisting of the successive stratum graphs together with couples of nodes of these graphs such that, for all $i \geq 1$,~$s_i$ and~$e_i$ are distinct words of the~$i$-th stratum, and we have $\delta_\pp(e_i, s_{i+1}) \leq d$.
\end{definition}

We can now present the second component of our algorithm. Note that the
algorithm runs on the in-memory representations of the stratum graphs, in which,
e.g., the subscripts are not stored.

\begin{toappendix}
  \subsection{Second component: Proof of Proposition~\ref{prp:second-comp}}
In this section we give the proofs for our bounded-delay enumeration
algorithm (Theorem~\ref{thm:enum-main}). We start by presenting in more details
what we called the second component of the amortized linear time algorithm,
namely:
\end{toappendix}

\begin{propositionrep}
\label{prp:second-comp}
  For $i \geq 1$, given the stratum graph $\Gamma_i$ and starting and ending nodes $v_{s_i} \neq v_{e_i}$ of~$\Gamma_i$, we can compute in time $O(|\Gamma_i|)$ a sequence of edit scripts $\sigma_1, \ldots, \sigma_{N_i-1}$  such that,
  letting $s_i = u_1, \ldots, u_{N_i}$ be the successive results of applying $\sigma_1, \ldots, \sigma_{N_i-1}$ starting with~$s_i$,
  then $u_1, \ldots, u_{N_i}$ is a $3d$-ordering of~$\stratum_\ell(L, i)$ starting at~$s_i$ and ending at~$e_i$.
\end{propositionrep}

\begin{proofsketch}
  We apply the spanning tree enumeration technique from
Lemma~\ref{lem:connected-enum} (in~$O(|\Gamma_i|)$) on~$\Gamma_i$, starting with~$v_{s_i}$ and
ending with~$v_{e_i}$, and read the scripts from the edge labels.
\end{proofsketch}
\begin{proof}
First, observe that by Proposition~\ref{prp:stratumconnected}, the
undirected graph corresponding to~$\Gamma_i$ is connected.
We then compute in linear time in~$\Gamma_i$ a spanning tree~$T$
of~$\Gamma_i$: this can indeed be done in linear time in our pointer machine
model using any standard linear-time algorithm for computing spanning trees, e.g., following
a DFS exploration, with the stack of the DFS being implemented as a linked list.
Notice that, importantly, the number of edges of~$\Gamma_i$, hence its size, is
linear in~$N_i$, since all nodes have bounded degree (namely, exponential in~$|A|$).

Next, we apply the enumeration technique described in the proof of
Lemma~\ref{lem:tree} on the tree~$T$ and starting node~$v_{s_i}$ and exit
node~$v_{e_i}$, which yields a sequence~$v_{s_i} = n_1 ,\ldots, n_{|N_i|} =
v_{e_i}$ enumerating all nodes of~$T$ exactly once, and where the distance
between any two consecutive nodes (in~$T$) is at most~$3$. The traversal
  can easily be implemented by a recursive algorithm: we prepare an output
  doubly linked list containing the nodes to enumerate, and when we enumerate a
  node in the algorithm, we append it to the end of this linked list, so that, at
  the end of the algorithm, the $3$-ordering of the nodes of the graph is stored
  in that linked list.

  We must simply justify that the recursive algorithm can indeed be implemented
  in our machine model. For this, we store the recursion stack in a 
  linked list. To make this explicit, we can say that the stack
  contains two kinds of pairs, where $n$ is a node of~$T$:
  \begin{itemize}
    \item $(\explore,n)$, meaning ``call the exploration function on node $n$'';
      and
    \item $(\enum,n)$, meaning ``append $n$ to the output linked list''.
  \end{itemize}
  Initially, the stack contains only $(\explore, v_{s_i})$ for $v_{s_i}$ the
  root of~$T$. At each step, we pop the first pair of the stack and do what it
  says, namely:
  \begin{itemize}
    \item If it is $(\explore,n)$, we call the exploration function on~$n$;
    \item If it is $(\enum,n)$, we append~$n$ to the output list;
    \item If the stack is empty, the exploration is finished.
  \end{itemize}
  The recursive exploration function can then be implemented as follows:
\begin{itemize}
\item When we explore a node~$n$ that is at even depth, we push at the beginning
  of the recursion stack the following list of elements in order: $(\enum,n)$,
    $(\explore,n_1')$, $\ldots$,
    $(\explore,n_m')$ for $n_1', \ldots, n_m'$ the children of~$n$;
\item When we explore a node~$n$ at odd depth that is not special, we push at the
  beginning of the recursion stack the following list of elements in order:
$(\explore,n'_1)$, $\ldots$, $(\explore,n'_m)$, $(\enum,n)$
 for $n_1', \ldots, n_m'$ the children of~$n$;
\item When we explore a node~$n$ at odd depth that is special and has~$\geq 1$ children~$n'_1\ldots n'_m$, 
  we push at the beginning of the recursion stack the following list in order:
$(\explore,n'_1)$, $\ldots$, $(\explore,n'_{m-1})$, $(\enum,n)$, $(\explore,n'_m)$.
\end{itemize}
Once we have the sequence~$v_{s_i} = n_1,\ldots, n_{|N_i|} = v_{e_i}$ produced
in the output doubly linked list, the last step is to 
compute in linear time the sequence $\sigma_1, \ldots, \sigma_{N_i-1}$ of edit
scripts. We determine each edit script by reading the edge labels of~$\Gamma_i$;
formally, to determine~$\sigma_j$, we simply re-explore~$\Gamma_i$ from~$n_j$ at
distance~$3$, which is in time exponential in~$|A|$, and when we find~$n_{j+1}$ we
concatenate the labels of the (directed) edges of~$\Gamma_i$ that lead from~$n_j$
to~$n_{j+1}$. This indeed gives us what we wanted and concludes the proof.
\end{proof}

In the rest of the section we present the first component of
our enumeration algorithm:

\begin{toappendix}
  \subsection{First component: Proof of Proposition~\ref{prp:worddag}}
In the rest of this appendix, we will present the formal details of the first component of
our enumeration algorithm, that is:
\end{toappendix}
\begin{propositionrep}
  \label{prp:worddag}
  There is an integer $C \in \NN$ exponential in~$|A|$
  such that we can produce a stratum graph sequence $(\Gamma_1, v_{s_1}, v_{e_1}), (\Gamma_2, v_{s_2}, v_{e_2}), \ldots$ for~$L$ in amortized linear time, i.e., for each $i \geq 1$, after having run $C \sum_{j=1}^i N_j$ steps, the algorithm is done preparing~$(\Gamma_i, v_{s_i},v_{e_i})$.
\end{propositionrep}

\subparagraph*{Word DAGs.}
The algorithm to prove Proposition~\ref{prp:worddag} will grow a large structure in memory, common to all strata, from which we can easily compute the~$(\Gamma_i,v_{s_i}, v_{e_i})$. We call this structure a \emph{word DAG}.
A word DAG is informally a representation of a collection of words, each of
which has outgoing edges corresponding to the possible left and right push operations.

\begin{definition}
  Let $\Lambda \colonequals \{\pushR(a) \mid a \in \Sigma\} \cup \{\pushL(a) \mid a \in \Sigma\}$ be the set of labels corresponding to the possible left and right push operations.
  A \emph{pre-word DAG} is an edge-labeled directed acyclic graph (DAG) $G = (V, \eta, \rout)$ where $V$ is a set of anonymous vertices, $\rout \in V$ is the \emph{root}, and
  $\eta\colon V\times \Lambda \to V \cup \{\bot\}$ represents the labeled edges in the following way:
  for each node $v \in V$ and label $s\in\Lambda$, if $\eta(v, s) \neq \bot$ then~$v$ has one successor $\eta(v, s)$ for label~$s$, and none otherwise. We impose:
  \begin{itemize}
    \item The root has no incoming edges. All other nodes have exactly two incoming edges: one labeled $\pushR(a)$ for
    some $a \in \Sigma$, the other labeled~$\pushL(b)$ for some~$b \in \Sigma$. Each node stores two pointers leading to these two parents, which may be identical.
    \item All nodes can be reached from the root via at least one directed path.
      \item The root has one outgoing edge for each child, i.e., for all $s \in \Lambda$, we have $\eta(\rout, s) \neq \bot$.
  \end{itemize}
  The \emph{word} represented by a directed path from the root to a node~$n$ is
  defined inductively:
  \begin{itemize}
  \item the word represented by the empty path is $\epsilon$,
  \item the word represented by a path $P, \pushR(a)$ is $wa$ where $w$ is the word represented
  by~$P$,
  \item the word represented by a path $P, \pushL(a)$ is $aw$ where $w$ is the word represented
  by~$P$.
  \end{itemize}
  The pre-word DAG $G$ is called a \emph{word DAG} if
  for each node $n$, all paths from~$\rout$ to~$n$ represent the
  same word. This word is then called \emph{the word represented by~$n$}.
\end{definition}

\begin{toappendix}
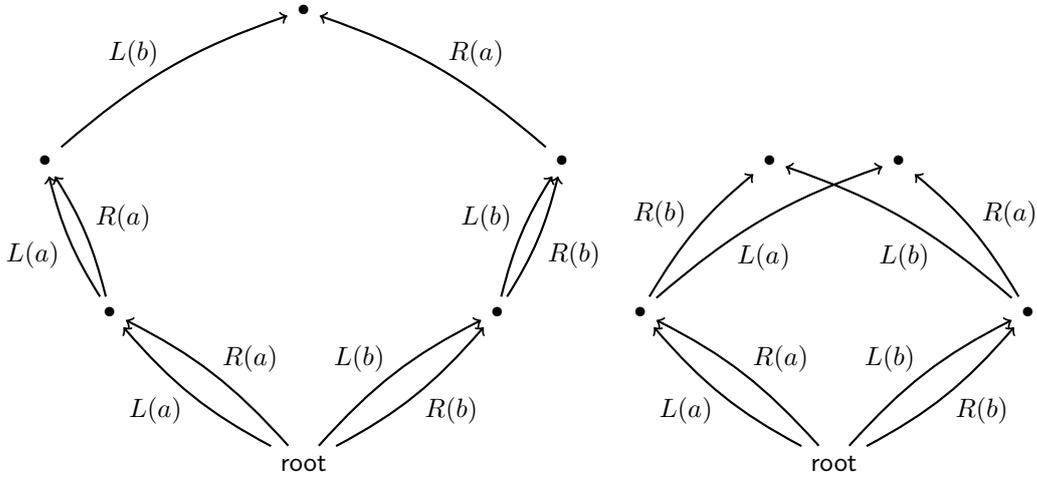
\begin{figure}
    \hfill
    \begin{tikzpicture}[xscale=.85]
      \node (root) at (0, 0) {$\rout$};
      \node (na) at (-3, 2) {$\bullet$};
      \node (nb) at (3, 2) {$\bullet$};
      \node (naa) at (-4, 4) {$\bullet$};
      \node (nbb) at (4, 4) {$\bullet$};
      \node (nx) at (0, 6) {$\bullet$};
      \draw[thick,->] (root) to [bend right=10] node[right,yshift=.2cm] {$R(a)$} (na);
      \draw[thick,->] (root) to [bend left=10] node[left,yshift=-.2cm] {$L(a)$} (na);

      \draw[thick,->] (root) to [bend right=10] node[right,yshift=-.2cm] {$R(b)$} (nb);
      \draw[thick,->] (root) to [bend left=10] node[left,yshift=.2cm] {$L(b)$} (nb);
      \draw[thick,->] (na) to [bend right=10] node[right,yshift=.2cm] {$R(a)$}
      (naa);
      \draw[thick,->] (na) to [bend left=10] node[left,yshift=-.2cm] {$L(a)$}
      (naa);

      \draw[thick,->] (nb) to [bend right=10] node[right,yshift=-.2cm] {$R(b)$}
      (nbb);
      \draw[thick,->] (nb) to [bend left=10] node[left,yshift=.2cm] {$L(b)$}
      (nbb);
      \draw[thick,->] (naa) to [bend left=10] node[left,yshift=.2cm] {$L(b)$} (nx);
      \draw[thick,->] (nbb) to [bend right=10] node[right,yshift=.2cm] {$R(a)$}
      (nx);
    \end{tikzpicture}
    \hfill
    \begin{tikzpicture}[xscale=.85]
      \node (root) at (0, 0) {$\rout$};
      \node (na) at (-3, 2) {$\bullet$};
      \node (nb) at (3, 2) {$\bullet$};
      \node (nab) at (-1, 4) {$\bullet$};
      \node (nab2) at (1, 4) {$\bullet$};
      \draw[thick,->] (root) to [bend right=10] node[right,yshift=.2cm] {$R(a)$} (na);
      \draw[thick,->] (root) to [bend left=10] node[left,yshift=-.2cm] {$L(a)$} (na);

      \draw[thick,->] (root) to [bend right=10] node[right,yshift=-.2cm] {$R(b)$} (nb);
      \draw[thick,->] (root) to [bend left=10] node[left,yshift=.2cm] {$L(b)$} (nb);
      \draw[thick,->] (na) to [bend left=10] node[left,yshift=.2cm] {$R(b)$} (nab);
      \draw[thick,->] (na) to [bend left=10] node[below,yshift=-.2cm] {$L(a)$} (nab2);
      \draw[thick,->] (nb) to [bend right=10] node[right,yshift=.2cm] {$R(a)$} (nab2);
      \draw[thick,->] (nb) to [bend right=10] node[below,yshift=-.2cm] {$L(b)$} (nab);
    \end{tikzpicture}
    \hfill

    \caption{Two example pre-word DAGs which are not word DAGs.
    The labels $\pushL$ and
    $\pushR$ are abbreviated for legibility.
    In the left pre-word DAG,
    the four paths to the top node that start to the left of the root all represent the word $baa$, whereas the four paths to that same node that start to the right of the root all represent the word $bba$. 
    In the right pre-word DAG, the left topmost node represents $ab$ and $bb$
    and the right topmost node represents $aa$ and $ba$. The criteria of word
    DAGs, and our construction to enlarge them, are designed to prevent these problems.}
    \label{fig:worddagnot}
\end{figure}
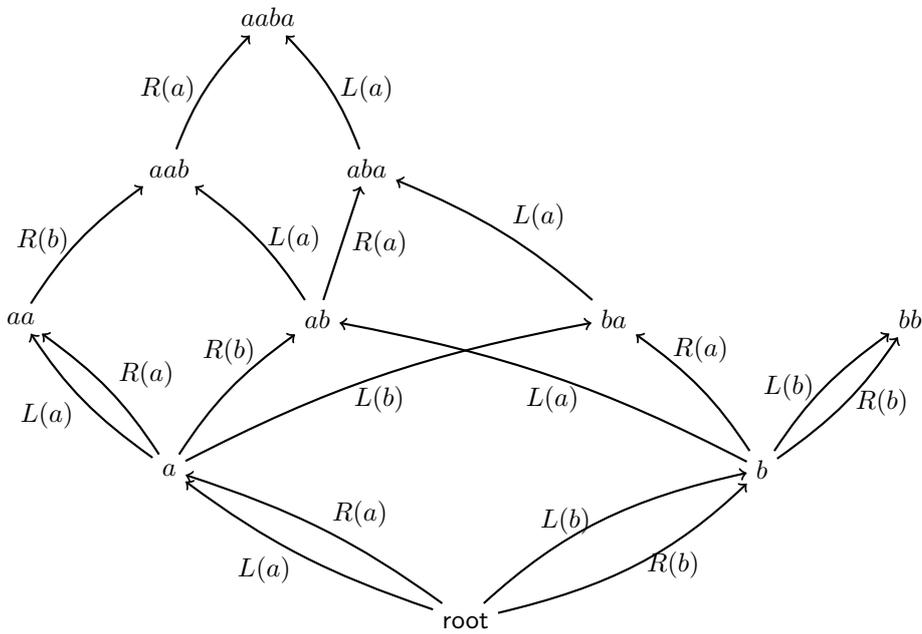
\begin{figure}
    \centering
    \begin{tikzpicture}[xscale=1.3]
      \node (root) at (0, 0) {$\rout$};
      \node (na) at (-3, 2) {$a$};
      \node (nb) at (3, 2) {$b$};
      \draw[thick,->] (root) to [bend right=10] node[right,yshift=.2cm] {$R(a)$} (na);
      \draw[thick,->] (root) to [bend left=10] node[left,yshift=-.2cm] {$L(a)$} (na);

      \draw[thick,->] (root) to [bend right=17] node[right] {$R(b)$} (nb);
      \draw[thick,->] (root) to [bend left=17] node[left] {$L(b)$} (nb);

      \node (naa) at (-4.5, 4) {$aa$};
      \node (nab) at (-1.5, 4) {$ab$};
      \node (nba) at (1.5, 4) {$ba$};
      \node (nbb) at (4.5, 4) {$bb$};

      \draw[thick,->] (na) to [bend right=10] node[right,yshift=.2cm] {$R(a)$}
      (naa);
      \draw[thick,->] (na) to [bend left=10] node[left,yshift=-.2cm] {$L(a)$}
      (naa);

      \draw[thick,->] (nb) to [bend right=10] node[right] {$R(b)$} (nbb);
      \draw[thick,->] (nb) to [bend left=10] node[left] {$L(b)$} (nbb);

      \draw[thick,->] (na) to [bend left=10] node[above,yshift=6] {$R(b)$} (nab);
      \draw[thick,->] (na) to [bend left=10] node[below] {$L(b)$} (nba);
      \draw[thick,->] (nb) to [bend right=10] node[above,yshift=6] {$R(a)$} (nba);
      \draw[thick,->] (nb) to [bend right=10] node[below] {$L(a)$} (nab);

      \node (naab) at (-3, 6) {$aab$};
      \node (naba) at (-1, 6) {$aba$};
      \node (naaba) at (-2, 8) {$aaba$};

      \draw[thick,->] (nab) to [bend right=10] node[right] {$L(a)$} (naab);
      \draw[thick,->] (nab) to  node[right] {$R(a)$} (naba);
      \draw[thick,->] (nba) to [bend right=10] node[right,yshift=6] {$L(a)$} (naba);
      \draw[thick,->] (naa) to [bend left=10] node[left] {$R(b)$} (naab);
      
      \draw[thick,->] (naab) to [bend left=10] node[left] {$R(a)$} (naaba);
      \draw[thick,->] (naba) to [bend right=10] node[right] {$L(a)$} (naaba);
    \end{tikzpicture}
    \caption{An example word DAG. We annotate the nodes with the word that they
    represent, even though in the memory representation the nodes are anonymous
    and the words are not represented. The labels $\pushL$ and $\pushR$ are
    abbreviated for legibility.}
    \label{fig:worddag}
\end{figure}
\end{toappendix}

Example pre-word DAGs and word DAGs are shown on
Figures~\ref{fig:worddagnot} and~\ref{fig:worddag} in the appendix. In our machine model, each node is represented by a record; crucially, like for stratum graphs, the word that the node represents is not explicitly written.

Crucially, word DAGs do not us allow not to create two
different nodes that represent the same word -- these 
would be problematic since we have to enumerate
without repetition.  

\begin{toappendix}
  To show Proposition~\ref{prp:worddag}, we first prove
  Fact~\ref{fct:noduplicate}:
\end{toappendix}

\begin{factrep}
\label{fct:noduplicate}
There are no two different nodes in a word DAG that represent the same word.
\end{factrep}
\begin{proof}
  Let us write $G = (V,\eta,\rout)$.
Assume by way of contradiction that some word~$w$ is represented by two
distinct nodes~$n_1\neq n_2$ of the word DAG, and take~$|w|$ to be minimal.
  Then~$w$ cannot be the empty word because~$\rout$ is the only node reachable
  by a empty path from~$\rout$, so that $n_1 \neq \rout$ and $n_2 \neq \rout$.
Thus, we can write~$w$ as~$w=aw'$ for~$a\in \Sigma$.  Now, consider
the~$\pushL(a)$-predecessors~$n'_1$ and~$n'_2$ of~$n_1$ and~$n_2$, respectively
(which must exist because~$n_1$ and~$n_2$ are not~$\rout$).
  Then, taking any path from~$\rout$ to~$n'_1$ (resp., to~$n'_2$) and continuing it to a
  path to~$n_1$ (resp., to~$n_2$), we know that $n_1$ (resp., $n_2$) must
  represent~$w'$. As $|w'| < |w|$, by minimality of~$w$, we must have $n_1' =
  n_2'$. But then this node has two distinct successors $n_1$ and $n_2$ for the same label
  $\pushL(a)$, a contradiction.
\end{proof}

We can then show the following theorem, intuitively saying
that we can discover all the words of the language by only visiting words that
are not too far from it:

\begin{toappendix}
We then prove Proposition~\ref{prp:worddag2}, whose statement we recall here:
\end{toappendix}

\begin{propositionrep}
\label{prp:worddag2}
  We can build a word DAG $G$ representing the words
of~$L$ in amortized linear time: specifically, for some value $C$ that is exponential in~$|A|$, for
all~$i$, after $C\times \sum_{j=1}^i N_j$ computation steps,  for each word~$w$
of~$\Sigma^*$ whose push-pop distance to a word
of~$\bigcup_{j=1}^i\stratum_\ell(L,j)$ is no greater than~$d$, then $G$ contains a
node that represents~$w$. Moreover, there is also a value~$D$ exponential in~$|A|$ such that any
node that is eventually created in the word DAG represents a word that is at
push-pop distance at most~$D$ from a word of~$L$.
\end{propositionrep}
\begin{proofsketch}
  We progressively add nodes to a word DAG while 
efficiently preserving its properties, and thus avoid creating duplicate nodes.
By labeling each node with the element of~$Q\cup \{\bot\}$ achieved by the word
represented by that node, and also by the distance to the closest known word
of~$L$, we can restrict the exploration to nodes corresponding to words that
are close to the words of~$L$, which ensures the amortized linear time bound.
\end{proofsketch}

This is enough to prove Proposition~\ref{prp:worddag}: we run the algorithm of Proposition~\ref{prp:worddag2} and, whenever it has built a stratum, construct the stratum graph~$\Gamma_i$ and nodes~$v_{s_i},v_{e_i}$ by exploring the relevant nodes of the word DAG. 
Full proofs are deferred to the appendix.

\begin{toappendix}
At a high level, we will prove Proposition~\ref{prp:worddag2} by first explaining
how to enlarge word DAGs, which would allow us to use them to represent all
words of~$\Sigma^*$; and then explaining how we can specifically build them to
efficiently reach the words that are close to~$L$.

But before doing so, we first make two simple observations about word DAGs that
will be useful for the rest of the proof. Here is the first one:

\begin{claim}
  \label{clm:nobad1}
  In a word DAG, we cannot have a node $n$ with edges $\pushL(a)$ and $\pushR(b)$ to the same node~$n'$ with
  $a \neq b$.
\end{claim}

\begin{proof}
  Letting $v$ be the word captured by~$n$,
  the node $n'$ would witness that we have $av = vb$. This equation is clearly
  unsatisfiable as
  $a \neq b$ so the number of occurrences of~$a$ and~$b$ in the left-hand-side
  and right-hand-side are clearly different.
\end{proof}

Our second observation is that we can have a node $n$ with edges $\pushL(a)$ and $\pushR(a)$ to the same node~$n'$, but
this happens if and only if $n$ (and therefore $n'$) captures a
word of the form $a^i$ for some $i\geq 0$:
\begin{claim}
  \label{clm:nobad2}
  If a node $n$ has edges $\pushL(a)$ and $\pushR(a)$ to the same node~$n'$, then $n$ (and
  therefore $n'$) capture a power of~$a$.
\end{claim}

\begin{proof}
  Letting $v$ the word captured by~$n$, 
  the node $n'$ witnesses that we have $av = va$. We show by induction on length
  that $v = a^i$. The base case is trivial for a length of~$0$, and for a length
  of~$1$ indeed the only solution is~$a$. For the induction, if we have a
  solution of length $n+1$, then it starts and ends with $a$, i.e., it is of the
  form $a v_2 a$ with $v_2$ of length~$n-1$. Injecting this in the equation, we get $a^2 v_2 a = a v_2
  a^2$, and simplifying $a v_2 = v_2 a$, i.e., $v_2$ satisfies the same
  equation. By induction hypothesis we have $v_2 = a^{n-1}$, and then $v =
  a^{n+1}$, establishing the induction case and concluding the proof.
\end{proof}

We are ready to explain how we enlarge word DAGs.

\subparagraph*{Enlarging a word DAG.}
We will start from the \emph{initial word DAG}, which we define to be the word
DAG 
whose nodes are~$\{\rout\}\cup \{v_a\mid a\in \Sigma\}$ and where each~$v_a$
is both the~$\pushR(a)$ and the~$\pushL(a)$-successor of~$\rout$ for all~$a\in
\Sigma$.
For instance, the initial word DAG for alphabet~$\Sigma = \{a,b,c\}$ is shown
in Figure~\ref{fig:initial-worddag}. It is clear that this is indeed a word DAG.

\begin{figure}
    \centering
    \begin{tikzpicture}
      \node (root) at (0, 0) {$\rout$};
      \node (na) at (-5, 2) {$\bullet$};
      \node (nb) at (0, 2) {$\bullet$};
      \node (nc) at (5, 2) {$\bullet$};
      \draw[thick,->] (root) to [bend right=10] node[right,yshift=.2cm] {$R(a)$} (na);
      \draw[thick,->] (root) to [bend left=10] node[left,yshift=-.2cm] {$L(a)$} (na);

      \draw[thick,->] (root) to [bend right=17] node[right] {$R(b)$} (nb);
      \draw[thick,->] (root) to [bend left=17] node[left] {$L(b)$} (nb);

      \draw[thick,->] (root) to [bend right=10] node[right,yshift=-.2cm] {$R(c)$} (nc);
      \draw[thick,->] (root) to [bend left=10] node[left,yshift=.2cm] {$L(c)$} (nc);
    \end{tikzpicture}
    \caption{The initial word DAG for alphabet~$\{a,b,c\}$, with $\pushL$ and
    $\pushR$ respectively shortened to~$L$ and~$R$ for brevity.}
    \label{fig:initial-worddag}
\end{figure}
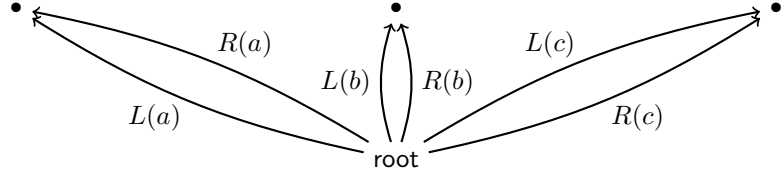

\begin{definition}
\label{def:complete}
Given a word DAG $G$ and a node~$n$ of~$G$, we say that~$n$ is \emph{complete}
if it has $2|\Sigma|$ 
outgoing edges; and \emph{incomplete} otherwise.
\end{definition}

For instance, the root of the initial word DAG is complete, but its children
are not.  We next explain how to complete an incomplete node.

\begin{definition}
\label{def:completion}
For a node~$n$ of~$G$ which is incomplete, we define the
completion $G\uparrow n$ of~$G$ on~$n$ as follows. 

First, for each label $\pushL(a)$ with~$a \in
\Sigma$ for which $n$ has no outgoing edge, we do what follows,
which is illustrated in Figure~\ref{fig:completion} for different cases.
We consider the successive strict
ancestors $n_1, \ldots, n_m$ of~$n$ for labels of the form $\pushR(b)$ for $b \in
\Sigma$ which do not have an $\pushL(a)$-successor (there may be none, in which case~$m=0$). We also consider $n_{m+1}$ the
closest ancestor of~$n$ via labels of the form~$\pushR(b)$ that has
  a~$\pushL(a)$ child, and call this child~$n_{m+1}'$.
  Notice that~$n_{m+1}$
is well-defined (because the root has a successor for each label).
Notice further that if~$m \geq 1$ then~$n_m$ is a~$\pushR(b)$-child of~$n_{m+1}$, and
if~$m=0$ then~$n$ itself is the~$\pushR(b)$-child of~$n_{m+1}$.
We let $\pushR(b_1),
\ldots, \pushR(b_m), \pushR(b_{m+1})$ be the labels of the edges to $n, n_1, \ldots, n_m$ along this path.

We then create $\pushL(a)$-successors $n_1', \ldots, n_m'$ of $n_1, \ldots, n_m$
respectively, and set the other predecessor of each of them to be, respectively,
$n_2$ with label $\pushR(b_2)$, $\ldots$,
  $n'_{m+1}$ with label $\pushR(b_{m+1})$,
We call these newly created nodes \emph{pillar nodes}.
Then we create a new node $n'$, and set it as the
$\pushL(a)$-successor of~$n$ 
and the $\pushR(b_1)$-successor of~$n_1'$.
Note that, in the case where $n_1' = n$ (which can only happen when~$m=0$ and
$b_1 = a$ by Claim~\ref{clm:nobad1}), we handle two labels at once, i.e., we
set~$n'$ as the $\pushL(a)$- and $\pushR(a)$-child of~$n$ as illustrated in
Figure~\ref{fig:compl-a3}.

Second, for each label $\pushR(a)$ with~$a \in \Sigma$ for which~$n$ has no
outgoing edge, we do the corresponding operation on a path of labels of the
form $\pushL(b)$, exchanging the role of the two kinds of labels.
\end{definition}

\begin{figure}
\begin{minipage}{\linewidth}
\begin{minipage}[t]{.4\linewidth}
\begin{subfigure}[t]{\linewidth}
    \begin{tikzpicture}
      \node (n3) at (0, 0) {$n_3$};
      \node (n2) at (2, 2) {$n_2$};
      \node (n1) at (4, 4) {$n_1$};
      \node (n) at (6, 6) {$n$};
      \node (np) at (4, 8) {$n'$};
      \node (np3) at (-2, 2) {$n_3'$};
      \node (np2) at (0, 4) {$n_2'$};
      \node (np1) at (2, 6) {$n_1'$};
      \draw[thick,->] (n3) -> (n2) node[right,pos=.4] {$R(b_3)$};
      \draw[thick,->] (n2) -> (n1) node[right,pos=.4] {$R(b_2)$};
      \draw[thick,->] (n1) -> (n) node[right,pos=.4] {$R(b_1)$};

      \draw[thick,->] (n3) -> (np3) node[left,pos=.4] {$L(a)$};
      \draw[thick,dashed,->] (n2) -> (np2) node[left,pos=.35] {$L(a)$};
      \draw[thick,dashed,->] (n1) -> (np1) node[left,pos=.35] {$L(a)$};
      \draw[thick,dashed,->] (n) -> (np) node[left,pos=.35] {$L(a)$};

      \draw[thick,dashed,->] (np3) -> (np2) node[right,pos=.35] {$R(b_3)$};
      \draw[thick,dashed,->] (np2) -> (np1) node[right,pos=.35] {$R(b_2)$};
      \draw[thick,dashed,->] (np1) -> (np) node[right,pos=.35] {$R(b_1)$};
    \end{tikzpicture}
    \caption{Case~$m=2$ and~$a\neq b_1$. The pillar nodes are~$n'_2$ and~$n'_1$.}
    \label{fig:compl-a1}
  \end{subfigure}
\end{minipage}
\begin{minipage}[t]{.25\linewidth}
\begin{subfigure}[t]{\linewidth}
    \hspace{1cm}\begin{tikzpicture}
      \node (n1) at (4, 4) {$n_1$};
      \node (n) at (6, 6) {$n$};
      \node (np) at (4, 8) {$n'$};
      \node (np1) at (2, 6) {$n_1'$};
      \draw[thick,->] (n1) -> (n) node[right,pos=.35] {$R(b_1)$};
      \draw[thick,->] (n1) -> (np1) node[left,pos=.35] {$L(a)$};
      \draw[thick,dashed,->] (n) -> (np) node[left,pos=.35] {$L(a)$};
      \draw[thick,dashed,->] (np1) -> (np) node[right,pos=.35] {$R(b_1)$};
    \end{tikzpicture}
    \caption{Case~$m=0$ and~$a\neq b_1$. There is no pillar node.}
    \label{fig:compl-a2}
\end{subfigure}
\end{minipage}
\begin{minipage}[t]{.25\linewidth}
\begin{subfigure}[t]{\linewidth}
    \begin{tikzpicture}
      \node (n1) at (0, 4) {$n_1$};
      \node (n) at (0, 6) {$n$};
      \node (np) at (0, 8) {$n'$};
      \draw[thick,->] (n1) to [bend right=45] node[right] {$R(a)$} (n);
      \draw[thick,->] (n1) to [bend left=45] node[left] {$L(a)$} (n);
      \draw[thick,dashed,->] (n) to [bend left=45] node[left] {$L(a)$} (np);
      \draw[thick,dashed,->] (n) to [bend right=45] node[right] {$R(a)$} (np);
    \end{tikzpicture}
    \caption{Case~$m=0$ and~$a = b_1$. There is no pillar node.}
    \label{fig:compl-a3}
\end{subfigure}
\end{minipage}
\end{minipage}
    \caption{Figures illustrating the completion procedure from
Definition~\ref{def:completion} for a label of the form~$\pushL(a)\in \Sigma$,
in the different possible cases.  Only the nodes of the word DAG that are
mentioned in the procedure are drawn. Dotted edges, and nodes that are the
successors of dotted edges, are created by the completion procedure illustrated. The
case~$m \geq 1$ and~$a=b_1$ cannot occur by definition of~$n_m$.}
    \label{fig:completion}
\end{figure}
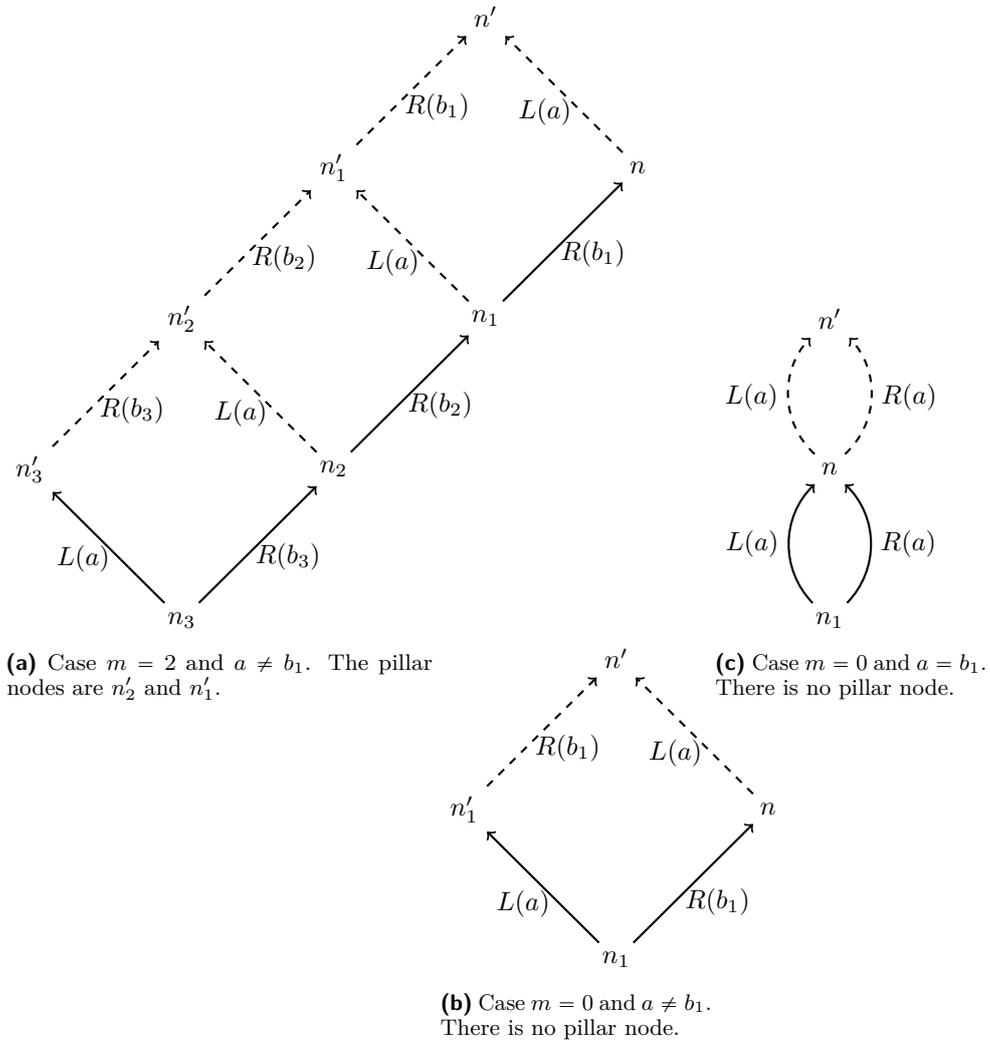

Thus, when completing an incomplete node~$n$, two types of nodes are added to
the word DAG: children of~$n$, and pillar nodes (possibly for several labels).
Further, notice that during this process, it is possible that another node
has been completed; for instance, this could happen to~$n'_1$ from Figure~\ref{fig:compl-a3}.

We prove next that this completion procedure does not break the properties of
word DAGs, and moreover that it can be performed in linear time with respect to
the number of newly created nodes (including pillar nodes). Note that the number
of pillar nodes created may be arbitrarily large, and it is the reason why the
complexity in Proposition~\ref{prp:worddag2} is stated as amortized linear time. We
do not see a way to avoid this, as the existence of all predecessors of a node
seems crucial to continue the construction and ensure that we do not create
multiple nodes for the same word (i.e., that we respect the properties of word
DAGs).

Let us show these results:

\begin{lemma}
\label{lem:complete-OK}
  For any word DAG $G$ and incomplete node~$n$ of~$G$, then $G
  \uparrow n$ is a word DAG. Further, $G \uparrow n$ has at least one additional
  node, and the running time to build it is proportional to the number of new
  nodes that are created.
\end{lemma}

\begin{proof}
We claim that the result is still a pre-word DAG. First, by definition all
newly created nodes have exactly two parents, one with a label of the form
$\pushL(a)$ and one with a label of the form~$\pushR(b)$ for some $a, b \in
\Sigma$.  Further, everything is still reachable from the root, and the root is
obviously still complete. Hence~$G \uparrow n$ is indeed a pre-word DAG.

We show next that it is also a word DAG, by
checking that this holds after each step
where we handle one missing outgoing label on~$n$, say $\pushL(a)$. 
Since all new paths in the word DAG after the operation end on one of the new nodes,
it suffices to check that every new node represents only one word. But it is clear by construction that any path from the root to a new node 
must go through one of the nodes $n_{m+1}'$, $n_m, \ldots, n_1, n$, which therefore all capture a unique word.
Specifically, letting $w$ be the word captured by~$n_{m+1}$, we know that
$n_{m+1}'$ captures $aw$, that $n_m$ captures $w b_{m+1}$, ..., that $n_1$
captures $w b_{m+1} \cdots b_2$, and $n$ captures $w b_{m+1} \cdots b_1$, and
the words captured by the new nodes are then defined in the expected way.

The claim about the running time is immediate, noting that the node was asserted
to be incomplete so there is at least one successor to create.
\end{proof}

Hence, one could in principle build the successive strata of~$L$ by starting
from the initial word DAG, keeping a FIFO queue (e.g., as a linked list) of the
nodes that are still incomplete, and progressively enlarging the word DAG with
this operation, thus eventually reaching all words of~$\Sigma^*$. The problem with this naive exploration strategy is that, for
most languages, it will not yield the amortized linear time bound that
we wanted. This is simply because the strata of our language~$L$ can be
much smaller than those of~$\Sigma^*$.

For this reason, we explain next how
to guide the construction of the word DAG to only reach words that are at
bounded push-pop distance from words of the
language.

From now on, to alleviate notation, we will often identify the nodes
of the word DAG with the words that they represent --- but we recall that the
words are not really written in the nodes of the word DAG, in particular
algorithms on word DAGs do not consider them.

\subparagraph*{Word DAGs for our regular language~$\bm{L}$.}
Intuitively, we would like to be able to efficiently decide if a node~$n$ of
the word DAG that we are building represents a word that is in the
language~$L = \L(A)$ or not. Indeed, informally, if we could do this, as we know
that every stratum is~$d$-connected by Proposition~\ref{prp:stratumconnected}
and has words that are close to words of adjacent strata, it would be enough to
only enlarge the word DAG by exploring the $d$-neighborhood of
words in~$L$.

To do this, for each node~$n$ of~$G$, letting~$w$ be the word represented
by~$n$, we will annotate~$n$ by the state~$q$ that is reached in the partial
run of $w$ from the initial state of~$A$ (or by~$\bot$ if this state is
undefined).  We can annotate the newly created nodes efficiently when we
enlarge a word DAG: for instance, if we have a $\pushR(a)$-successor~$n'$
of~$n$ in the word DAG, then we could simply annotate~$q'$ by~$\delta(q,a)$.
The same situation does not work if~$n'$ is a $\pushL(a)$-successor, but
we will always be able to define the annotation of new nodes by following the
$\pushR$-transitions from states with a known annotation.

We will also annotate each node of the word DAG by the distance
to the closest word of the language that is in the word DAG, if this distance
is~$< d$, or by~$\infty$ otherwise.
Moreover, for a technical reason that will become clear later, we will also
annotate every node of the word DAG with the modulo~$\ell$ of its \emph{depth},
i.e., the length of the unique word that the node represents.
Note that these values are all bounded by an exponential function of the input,
so they can be represented in our machine model.
We formalize this next.

\begin{definition}
An $A$-word DAG is a word DAG where each node~$n$ is additionally annotated with:
\begin{itemize}
   \item An element of~$Q\cup \{\bot\}$, called the \emph{state annotation};
   \item An integer in $\{0, \ldots, d-1\} \cup \{\infty\}$, called the \emph{distance annotation};
   \item An integer in~$\{0,\ldots,\ell-1\}$, called the \emph{modulo annotation}.
\end{itemize}
We then additionally require that:
\begin{itemize}
  \item  The state annotation represents the state that is reached
by reading~$w$ in~$A$, where~$w$ is the word represented by~$n$ (and~$\bot$ if
we do not reach any state). 
  \item For every node~$n$, letting $i$ be the shortest distance of an
undirected path between $n$ and a node representing a word of the language, the
integer labeling~$n$ is~$i$ if $i < d$ and $\infty$ otherwise.
  \item For every node~$n$, letting~$w$ be the word that it represents, then
the modulo annotation of~$n$ is congruent with~$|w|$ modulo~$\ell$.
\end{itemize}
We say that a node is \emph{successful} if $m$ corresponds to a word accepted
by the automaton, i.e., it is labeled by a final state. Equivalently, its distance
annotation is then~$0$.
\end{definition}

The \emph{initial~$A$-word DAG} is defined in the expected way by annotating the
initial word DAG with the correct elements. We show that these new annotations
can be maintained when we complete a node, in linear time in the number of newly
created nodes:

\begin{lemma}
\label{lem:wordexpand}
  If $G$ is an $A$-word DAG, for an incomplete node~$n$, we can extend $G$ to
  an $A$-word DAG also denoted~$G\uparrow n$ consisting of the word DAG
  $G\uparrow n$ and correctly updating the annotations. The running time is
  again proportional in the number of created nodes (recall that $A$ and~$d$ and~$\ell$ are
  fixed).
\end{lemma}

Note that the state annotation and modulo annotation
will be assigned once when a node is created and never modified afterwards,
whereas the distance annotation can be modified after the node is created, if we
later discover a new word of~$L$ sufficiently close to the node. We prove
the lemma:

\begin{proof}[Proof of Lemma~\ref{lem:wordexpand}]
  For the state annotation, we consider two cases, depending
  on whether the missing outgoing edge of node~$n$ that we are completing was a~$\pushR(a)$- or
  a~$\pushL(a)$- successor. If it was a~$\pushR(a)$ successor, then
  all newly created nodes have a $\pushR(a)$-predecessor that was already
  in the $A$-word DAG before the operation so we can easily determine
  their annotation by looking at the transition function of~$A$.
  If it was a~$\pushL(a)$-successor
(the case of Figure~\ref{fig:completion}), then we can determine the state
  annotation of each new node by starting from the state annotation of
  node~$n_{m+1}'$ from the construction (which was in the $A$-word DAG before
  the operation), again by simply reading the transition function of~$A$. This
  can clearly be done in linear time in the number of created nodes.
It is clear how to compute the modulo annotation of each newly created node, counting from the modulo
annotation of the expanded node, modulo~$\ell$.

Next, we update the distance annotations. These distance annotations only need
to be updated on the nodes of the word DAG that are within a distance $< d$
of the newly created nodes, so we recompute them there: for every new node, we
explore its neighborhood in the word DAG up to a distance of~$d-1$ and update the
distances of nodes in that neighborhood.

  For the running time claim, we use the fact that the degree of the word DAG is at
  most $2|\Sigma|+2$, so the neighborhood at distance $d-1$
of any node in~$G$ is of size exponential in~$|A|$,
namely, it has size $\leq \sum_{j=0}^{d-1} (2|\Sigma|+2)^j \leq
(2|\Sigma|+2)^{d}$. If we
create~$N$ nodes, updating the distances is thus done in linear time in~$N$.
\end{proof}

We are now ready to describe the exploration algorithm of
Proposition~\ref{prp:worddag2}. 

\subparagraph*{Initialization.} First expand the initial $A$-word DAG by
increasing length until the incomplete nodes are exactly the nodes
representing words of~$\Sigma^*$ of length~$\ell$, so that all words
of~$\Sigma^*$ of length~$\leq \ell$ are in the $A$-word DAG.  This can be
done in time exponential in~$|A|$ as there are~$\leq |\Sigma|^{\ell+1}$ such words.  Then
prepare a set of~$\ell$ empty FIFO queues~$B_0,\ldots,B_{\ell-1}$, each being
implemented as a doubly linked list in our machine model.  Also initialize an
empty buffer~$B$ as a linked list, and fill this buffer by adding a pointer to all the
nodes of the word DAG that are incomplete (their depth is~$\ell$) and whose
distance annotation is~$< \infty$. Also add for each such node a pointer to its
(unique) occurrence in~$B$.  

Intuitively, all queues~$B_j$ and the buffer~$B$ will hold (nodes corresponding
to) words that are incomplete and whose distance annotation is~$< d$.
Moreover, each~$B_j$ will hold nodes whose length is~$j$ modulo~$\ell$, and the
buffer~$B$ will also hold words whose length is~$0$ modulo~$\ell$ (intuitively
corresponding to the nodes of the next strata). 
This is indeed true after the initialization step,
and we will prove 
that this will always hold, since the algorithm will continue as follows:

\subparagraph*{Indefinitely repeat the following operation:}
\begin{itemize}

\item If one of the queues~$B_j$ is not empty, consider the smallest~$j$
such that~$B_j$ is not empty, pop a node~$n$ from~$B_j$ and expand it, i.e.,
build the~$A$-word DAG~$G\uparrow n$.  
When creating new nodes with this
operation, if their (updated) distance annotation is~$< \infty$ (i.e.,~$<
d$), then:
\begin{itemize}
 \item if the newly created node~$n'$ is a child of the node~$n$, then, letting~$t$
be the modulo annotation of~$n'$, we
put~$n'$ into the queue~$B_{t}$ if~$j<\ell -1$ and into the buffer~$B$
otherwise (i.e., if~$j=\ell-1$);
\item if the newly created node~$n'$ is a pillar node then we put $n'$ in
queue~$B_t$ where~$t$ is its modulo annotation.
\end{itemize}
Notice that all newly created nodes are incomplete by construction.  Similarly,
when updating the distance annotation of nodes that already existed, if they
are incomplete and their updated distance goes from~$\infty$ to~$<d$ then
we put them in queue~$B_i$ where~$i$ is their modulo annotation.  Whenever we
add a node to one of the lists or to the buffer, we also add in the node a pointer
to its occurrence in this list/buffer.  Moreover, when we make a node complete,
    either by expanding it or by attaching new edges to it when completing another
    node (which can be detected in the construction), we remove this node from the
    queue it was in (if any), using the pointer to find this occurrence.
    All of this can
be performed in time exponential in~$|A|$, for each new node added during this expansion step.
\item If all queues~$B_j$ are empty, then transfer all elements of the
buffer~$B$ to the queue~$B_0$ (in a constant number of operations as we work
with linked lists). 
\end{itemize}

This ends the description of the algorithm. Next we analyse it.

We define the \emph{$1$st phase} of the algorithm as the initialization phase,
and for~$i\geq 2$, we define the~\emph{$i$-th phase} to be 
the~$i-1$-th execution of the above while loop.
During the $i$-th phase, we intuitively enlarge the word DAG
with words of the $i$-th stratum,
plus possibly some words of lower strata because of pillar nodes.

We now show that this algorithm meets the requirements of
Proposition~\ref{prp:worddag2}.  To this end, we first show that the algorithm only
discovers nodes that are sufficiently close to words of the language, i.e., it
does not ``waste time'' materializing nodes that are too far away and will not
yield words of the language. 
In other words, we first show the claim corresponding to the last sentence
of Proposition~\ref{prp:worddag2}.
This is clear for the children of nodes that we
expand: when we expand a node $n$, its distance is $<d$, so its new
children~$n'$ will be at distance $\leq d$.  However, the expansion may
create an unbounded number of 
pillar nodes
and we have a priori no distance bound on
these. Fortunately, these nodes in fact represent factors of the word
represented by~$n'$. Thus, we show that the distance bound on~$n'$ extends to
give a weaker distance bound on these nodes:

\begin{lemma}
  \label{lem:factors}
  For any word $w$ of length $\geq 2 d$
  at push-pop distance $\leq d$ to a word of~$L$, for any factor $w'$
  of~$w$, then $w'$ is at push-pop distance  $\leq 6d$
  to a word of~$L$.
\end{lemma}

Note that the distance bound in this result is weaker than what is used in the
algorithm to decide which nodes to expand, i.e., the distance is $\leq 6d$ and
not $\leq d$. So the pillar nodes may have distance annotation $\infty$. Still,
the weaker bound will imply (in the next proposition) that all words created in
the word DAG obey some distance bound.

We can now prove the lemma:

\begin{proof}[Proof of Lemma~\ref{lem:factors}]
  As $|w| \geq 2 d$, we can write $w = s u t$ with $|s|, |t| \geq d$.
Moreover, any word $v$ at push-pop distance~$\leq d$ from~$w$ can be written 
as $v = \rho u \tau$ with $|\rho|, |\tau| \leq 2d$.  Take~$v$ to be such a word in~$L$,
which exists by hypothesis.  Now, the factor $w'$ of~$w$ is either a factor
of~$u$, or may include parts of~$s$ and/or of~$t$. But in all cases we can
write it $w' = s' u' t'$ with $|s'|, |t'| \leq d$ and $u'$ a factor of~$u$.
Now, by considering the accepting run of~$v$ and the occurrence of~$u'$ inside,
  considering the states $q_l$ and $q_r$ before and after this occurrence of~$u'$
  in the accepting run, we know we
  can take words $\rho'$ leading from the initial state to~$q_l$ and~$\tau'$
  from~$q_r$ to a final state, having length $\leq |Q| \leq d$, so that $v' =
  \rho' u'
  \tau'$ is a word of~$L$.
  But then observe that~$w' = s'u't'$ is at push-pop distance at
most~$6d$ from~$v' = \rho'  u' \tau'$: we can pop left~$s'$ (at most~$d$),
push left~$\rho'$ (at most~$2d$), pop right~$t'$ (at most~$d$)
and then push right~$\tau'$ (at most~$2d$). This concludes since~$v'$ is
in~$L$.
\end{proof}

This allows us to show that indeed, all words that will eventually be
discovered by the algorithm are within some bounded distance of a word of the
language. Formally:
\begin{proposition}
\label{prp:nottoofar}
There is a value~$D$ exponential in~$|A|$ such that any (node corresponding to a) word~$w$
that is added to the word DAG at some point is at push-pop distance at most~$D$
from a word of the language.
\end{proposition}
\begin{proof}
  To account for initialization phase,
let~$D'$ be the maximal push-pop distance of a word (of~$\Sigma^*$) of length~$< 2d$ to a
word of the language. Now, let~$D \colonequals \max(D',6d)$, and let us show
  that this value of~$D$ is suitable. 

  Let~$w$ be a
word that is discovered by the algorithm. If~$|w|< 2d$ then we are done
by our choice of~$D'$, so assume~$|w|\leq 2d$.
This implies that~$w$ has been discovered during some~$i$-th phase for~$i\geq 2$
  (since the initialization phase only discovers words of size~$\leq \ell$ and~$d=2\ell$).
  Now, at the moment where~$w$ was
discovered, there are two cases. The first case is when~$w$ was created as a child of a node~$w'$ that we have
expanded, hence whose annotation was~$<d$, so~$w$ itself is at distance at
most~$d$, hence at most~$D$, from a word of the language. The second case is
  when~$w$ is was created as an
ancestor of a child of a node~$w''$ that has been expanded (i.e.,~$w$ is a
pillar node), hence a factor of~$w''$: then Lemma~\ref{lem:factors} applies
because~$w''$ is at push-pop distance at most~$d$ of a word of the language by
the same reasoning and moreover we have~$|w''|\geq |w|$ since~$w''$ is a factor
of~$w$, and since~$|w|\geq 2d$ by hypothesis we have indeed~$|w''| \geq 2d$, so
our choice of~$D$ works.
\end{proof}

Notice that Proposition~\ref{prp:nottoofar} would hold no matter in which order
we expand incomplete nodes of the word DAG.

Next, we show that, thanks to our exploration strategy, the algorithm discovers
all words of the language, stratum by stratum. The proof is more involved. We
start by proving the property that we informally mentioned after describing
the initialization step of the algorithm.

\begin{claim}
\label{clm:bloup}
For all~$i\geq 1$, all the words in the word DAG that are discovered before or during
the~$i$-th phase have length~$\leq i\ell$.  Moreover, during the $i$-th phase,
each queue~$B_j$ only ever contains words whose length is~$\leq i\ell -1$ and
is~$j$ modulo~$\ell$, and the buffer~$B$ only ever contains words whose length
is~$0$ modulo~$\ell$.
\end{claim}
\begin{proof}
By induction on~$i$. 
\begin{itemize}
\item Case~$i=1$. This is trivially true for the~$1$st phase (the initialization phase).
\item Case~$i>1$. By induction hypothesis, at the beginning of the~$i$-th phase
the word DAG only contains nodes of size~$\leq (i-1)\ell$, hence~$\leq i\ell$.
Moreover, considering the last step of the~$(i-1)$-th phase, we see that~$B_0$
only contains words of size~$(i-1)\ell$, and this is allowed.

Now, let~$n_1,n_2,\ldots$ be the words that are discovered during the $i$-th phase of the algorithm,
considering that, e.g., when we expand a
node~$n$ for some missing~$\pushR(a)$-label, we first “discover”
its~$\pushR(a)$-child, and then discover the pillar nodes (if any) in descending order.
Then we show the following claim by induction on~$j$: ($\star$)
[the node~$n_j$ has size~$\leq i\ell$ and, if at some point we
add it to some queue~$B_t$ then its size is~$\leq i\ell -1$ and is~$t$
modulo~$\ell$ and if we add it to~$B$ its size is~$0$ modulo~$\ell$]. This is
clear of the first word that we discover:
indeed, $n_1$ is a child of the very first node that we popped from~$B_0$
to expand, so the size of~$n_1$ is exactly~$(i-1)\ell+1$ and, since its modulo annotation is then~$1$,
the only queue to which we could ever add it is~$B_1$ and~$(i-1)\ell +1 \leq i\ell-1$ indeed as~$\ell \geq 2$.
Let now~$n_{j+1}$ be the~$(j+1)$-th word
discovered during the first phase and assume~$(\star)$ to hold for all
previously discovered nodes of that phase. We first show that~$|n_{j+1}|\leq
i\ell$. Consider indeed the moment that this node was discovered: it was either
the child of some node~$n_{\rho}$ for~$\rho\leq j$ that we expanded, or it was
a pillar node of some node~$n_{\rho}$ that we expanded. Since~$(\star)$ is true
for~$n_\rho$ and since we never expand nodes of the buffer~$B$ during a phase, it follows
that indeed~$|n_{j+1}|\leq i\ell$. 
We next show the claim on the queues/buffer.
Observe then that the only problematic case is if we
add~$n_{j+1}$ to the queue~$B_0$ and~$|n_{j+1}| = i\ell$; indeed,
it is clear that when we add a node to some queue its modulo is correct
with respect to that queue,
so the only
constraint that could be violated is that~$|n_{j+1}|\leq i\ell -1$. So let us
assume that~$|n_{j+1}|=i\ell$ by way of contradiction.  But then, considering
again how we have discovered~$n_{j+1}$, we see that the only possibility is
that it is a child of some node~$n_\rho$ that we have expanded ($n_{j+1}$
cannot be a pillar node because all nodes~$n_{j'}$ with~$j'\leq j$ have
size~$\leq i\ell$ and we defined that in the discovering order we discover
children of expanded nodes before their pillars), but then by induction
hypothesis~$n_\rho$ must have been popped from the~$B_{\ell-1}$ queue and then
the algorithm must have added~$n_{j+1}$ into~$B$ and not into~$B_0$.
\end{itemize}
This concludes the proof.
\end{proof}

Next, we observe that every word of the~$(i+1)$-th
stratum can be obtained from some word of the~$i$-th stratum by a specific sequence
of at most~$d$ push-pop edits. This claim is reminiscent of the proof of
Lemma~\ref{lem:scord}, where we showed the existence of so-called ladders. Let
us formally state the result that we need:

\begin{claim}
\label{clm:shorten}
  For any $i\geq 1$, for any word~$w$ of~$\stratum_\ell(L,i+1)$,
  there is a word $w'\in \stratum_\ell(L,i)$ with push-pop distance
  at most $d$ to~$w$ and that can be built from~$w$ as follows:
  first remove a 
  prefix 
  of length at most $\ell+|Q|$ 
  to get a word $w''$ (not necessarily in~$L$) of length exactly $i\ell-|Q| -
1$, and then add back a prefix corresponding to some path of length $\leq |Q|$
from the initial state to get $w$ as desired.
\end{claim}
\begin{proof}
  The prefix removal and substitution is simply by replacing a prefix of an
  accepting run with some simple path from the initial state that leads to the
  same state: note that the exact same argument was already used in the proof of
  Lemma~\ref{lem:density}. 
The fact that~$w'$ can be chosen to have length exactly~$i\ell -|Q|-1$ is simply
because~$\ell \geq |Q|$.
For the distance bound, 
notice that~$\ell + 2|Q| \leq d$.
\end{proof}

These observations allow us to show that the algorithm discovers
all words of the language, stratum by stratum:

\begin{proposition}
\label{prp:discovers-all}
For~$i\geq 1$, at the end of the~$i$-th phase, the algorithm
has discovered all words of~$\bigcup_{j=1}^{i} \stratum_\ell(L,j)$.
\end{proposition}
\begin{proof}
We show this by induction on $i$. 
\begin{itemize}
  \item Case~$i=1$. This is trivially true for the initialization step. 
  \item Case~$i>1$. 
By induction hypothesis we know that the algorithm has discovered all words of
$\bigcup_{j=1}^{i-1} \stratum_\ell(L,j)$. 
    So, let~$w\in \stratum_\ell(L,i)$ and we
show that the algorithm will discover $w$ during the~$i$-th 
phase, which would conclude the proof.

By Claim~\ref{clm:shorten} there is a word~$w'\in \stratum_\ell(L,i-1)$ such
that~$\delta_\pp(w,w')\leq d$ and that can be transformed into~$w$ by first
popping-left a prefix of size at most~$|Q|$, obtaining a word~$w''$ whose size
is exactly~$(i-1)\ell-|Q|-1$, and then pushing-left a prefix of size at
most~$\ell +|Q|$ to obtain~$w$. Thus, let us write~$w = w'' a_1 \ldots a_t$ with~$a_j\in
\Sigma$ and $t\leq \ell + |Q|$. Now, by induction hypothesis, the algorithm
has discovered~$w'$ during the~$(i-1)$-th phase, and
    by the properties of a word DAG the word~$w''$ is stored at a node which is
    an ancestor of~$w'$ and has also been discovered.
Crucially, notice that thanks to its size, 
$w''$ is actually in the~$(i-1)$-th stratum of~$L$.
    Now, this means that, at some point of the $i-1$-th phase, the
    algorithm has discovered~$w''$. 
Hence, at some point during the~$(i-1)$-th phase, the algorithm has already
    discovered both~$w'$ and~$w''$, and in fact all nodes in some simple path from~$w'$ to~$w''$. But then at that point, the distance annotation
    of~$w''$ was~$<d$, since it is then at distance~$\leq |Q|< d$
in the word DAG from~$w'$.
Thus, during the $i-1$-th phase,
    all children of~$w''$, and all its descendants up to a depth of~$|Q|+1$, must
    have been created, because $2|Q|+1 < d$,
and then the prefix $w'' a_1 \ldots a_{|Q|+1}$ of~$w$ of length
    $(i-1)\ell$ was added to the buffer~$B$, and then in queue~$B_0$ at the end of that phase.
    We now show that all longer prefixes of~$w$, i.e., all prefixes of~$w$ which
    are of the form $w''
    a_1 \ldots a_{|Q|+1} a_{|Q|+1+1} \cdots a_t$, including $w$ itself, were
    discovered.
    We know that the prefix $w'' a_1 \ldots a_{|Q|+1}$ was discovered and put
    in the buffer~$B$ for the $(i-1)$-th phase. Now, during the $i$-th phase, we have
    also completed this node $w'' a_1 \ldots a_{|Q|+1}$, and its descendants up
    to depth $\ell$, because again,~$\ell + 2|Q|\leq d$.
    Thus, indeed $w$ was discovered, which concludes the proof.\qedhere
\end{itemize}
\end{proof}

We point out that, when we are at the $i$-th
phase, we can create arbitrarily long
paths of pillar nodes during a single expansion operation, including nodes
representing words having length $<(i-1)\ell$. However,
the above Proposition implies that no such node can 
represent a word of the language,
because all the words of~$L$ of the $(i-1)$-th stratum
have already been discovered. 

The last ingredient to show Proposition~\ref{prp:worddag2} is then to prove that
by the end of the~$(i+2)$-th phase, the algorithm has discovered all words
whose push-pop distance to a word of the first $i$-th strata
is no greater than~$d$. Formally:

\begin{proposition}
\label{prp:discovers-neigh}
For all~$i\geq 1$, by the end of the~$(i+2)$-th phase, the algorithm has discovered
all words of~$\Sigma^*$ whose push-pop distance to a word of the~$i$ first strata is no greater than~$d$.
\end{proposition}

\begin{proof}
  We know by Proposition~\ref{prp:discovers-all} that by the end of the~$i$-th
phase (hence by the end of the~$(i+2)$-th phase) the algorithm  has discovered
all words of~$\bigcup_{j=1}^{i} \stratum_\ell(L,j)$. Let~$w$ be a word that is
at distance~$d$ from some word~$w'$ of $\bigcup_{j=1}^{i} \stratum_\ell(L,j)$,
and let us show that the algorithm discovers it by the end of the~$(i+2)$-th
phase.  We show it by induction on~$t \colonequals \delta_\pp(w,w')$. If~$t=0$
then~$w=w'$ and we are done.
For the inductive case let~$t>1$ and assume this is true for all~$j\leq t$.
As~$\delta_\pp(w,w')= t$, there is a sequence of~$t$ push/pop operations
that transform~$w$ into~$w'$.
Let~$w''$ be the word just before~$w$ 
that this sequence defines. By induction hypothesis we will have discovered~$w''$
by the end of the~$(i+2)$-th phase.
But then observe that~$|w''|\leq (i+2)\ell -2$
because~$d \leq 2\ell$. But then it is clear that we will also discover~$w$
by the end of the~$(i+2)$-th phase, because~$w''$ will be made complete before
  the end of the $(i+2)$-th phase. Indeed, the distance annotation of~$w''$
  is~$<d$ as witnessed by~$w'$. Further, we have $|w''| \leq |w'|+d$, and as $d
  = 2 \ell$ and $|w'| < i \ell$, we have $w'' < (i+2)\ell$, so indeed we will
  complete it before the $(i+2)$-th phase concludes.
\end{proof}

We are now equipped to prove Proposition~\ref{prp:worddag2}, which claims that the
word DAG construction is in amortized linear time and that all created nodes are
within some bounded distance to words of the language (keeping in mind that the latter was already
established in Proposition~\ref{prp:nottoofar}):

\begin{proof}[Proof of Proposition~\ref{prp:worddag2}.]
By Proposition~\ref{prp:discovers-neigh} and Proposition~\ref{prp:nottoofar},
it is enough to show that there is a value~$C$ exponential in~$|A|$ such that for all~$i\geq 1$,
after~$C\times \sum_{j=1}^i N_i$, the algorithm has finished the~$(i+2)$-th phase.
  Consider then the point~$P$ in the execution of the algorithm where it
has finished the~$(i+2)$-th phase, and let us find out how much time this
has taken. By Claim~\ref{clm:bloup} at~$P$ the algorithm 
has not discovered any word of
length~$>(i+2)\ell$. Moreover, we know by Proposition~\ref{prp:nottoofar}
that all the words represented by the nodes added to the word DAG are within some bounded push-pop distance~$D$ from
words of the language. Note that it could be the case that the witnessing words
  of the language are not yet discovered, in particular that there are in higher
  strata. However, we can let~$K = \lceil D/\ell \rceil$, and then we know that 
  any word discovered is at distance at most~$D$
  from a word of the first~$i+2+K$ strata.

Letting as usual~$N_i$ be the size of the~$i$-th
$\ell$-stratum of~$L$, we then have that the algorithm has discovered at 
most~$\sum_{j=1}^{i+2+K} (2|\Sigma|+2)^D N_j$ nodes.  Moreover, for~$C_A$ the
bound from Lemma~\ref{lem:density}, we have for all~$i$ the inequality $N_{i+1} \leq
C_A N_i$, so that we can bound the number of discovered nodes at follows: 

\begin{align*}
\sum_{j=1}^{i+2+K} (2|\Sigma|+2)^D N_j &=  (2|\Sigma|+2)^D \bigg(\sum_{j=1}^{i} N_j + \sum_{e=1}^{K+2} N_{i+e} \bigg)\\
&=  (2|\Sigma|+2)^D \bigg(\sum_{j=1}^{i} N_j + \sum_{e=1}^{K+2} C_A^e N_{i} \bigg)\\
&\leq (2|\Sigma|+2)^D \bigg(\sum_{j=1}^{i} N_j + C_A^{K+3} N_{i} \bigg)\\
&\leq  (1+C_A^{K+3})(2|\Sigma|+2)^D \sum_{j=1}^{i} N_j.
\end{align*}
Hence, this value is $C' \sum_{j=1}^{i} N_j$ for some value~$C'$ exponential in~$|A|$. Thus, the total number of
  nodes added to the word DAG is indeed proportional to $\sum_{j=1}^i N_j$. Now,
  as we only add nodes by performing completions, and only consider
  nodes on which there is indeed a completion to perform (in
  particular removing from the lists $B_j$ the nodes that have become
  complete), Lemma~\ref{lem:wordexpand} ensures that the running time satisfies
  a similar
bound, which 
  is what we needed to show.
\end{proof}

Next, we show Proposition~\ref{prp:worddag}, which claims that we can produce the
stratum graph sequence in amortized linear time:

\begin{proof}[Proof of Proposition~\ref{prp:worddag}]
For this, we will extend the
algorithm from Proposition~\ref{prp:worddag2} to compute the stratum graphs.
Specifically: once we
have finished the~$(i+2)$-th phase (during which we discover all nodes of
the~$(i+2)$-th stratum), and before entering the~$(i+3)$-th phase, we
prepare in linear time in~$N_{i}$ the stratum graph~$\Gamma_{i}$ for stratum~$i$ and the
corresponding starting and exit nodes. Note that we can easily compute the first 
stratum graph~$\Gamma_1$ during the initialization, as well as starting and exit points for it;
we pick an exit point for the $1$st stratum which is a word that is at distance
$\leq d$ to a word of the second stratum, as can be checked by naively testing all
possible edit scripts of at most $d$ operations. All of this can be computed by
a naive algorithm, and will only increase the delay by an amount that is exponential in~$|A|$.
Also notice the following claim~$(\star)$ by the end of the~$(i+2)$-th phase, we have found all
words of~$\Sigma^*$ whose push-pop distance to a word of~$\bigcup_{j=1}^i\stratum_\ell(L,j)$
is~$\leq d$, and these words are all of size at least~$(i-3)\ell$ and at most~$(i+2)\ell$: this is by Proposition~\ref{prp:discovers-neigh} and because~$2\ell \geq d$.
We can moreover easily modify the algorithm of Proposition~\ref{prp:worddag2}
so that it keeps, when it is in the~$i$-th phase, a
pointer to some successful node~$w_i$ of the~$i-2$-th stratum (or a null pointer
when~$i<3$).

Let us explain how we compute the stratum graph~$\Gamma_{i}$ in linear time in~$N_i$ after the
$(i+2)$-th phase has concluded. We assume that we already know the ending
point~$v_{e_{i-1}}$ of the previous phase, and that it was picked to ensure
that there was a word of~$\stratum_\ell(L,i)$ at distance $\leq d$ (as we did
above for the 1st stratum).
We do the computation in four steps:

  \begin{itemize}
  \item First, we explore the $A$-word DAG using a DFS (e.g., with a stack
implemented as a linked list, and without taking into account the orientation
of the edges) starting from the node~$w_i$ from the~$i$-th strata, only
visiting nodes corresponding to words of length $\geq (i-3)\ell$ and $\leq
(i+2)\ell$ (which can be detected with the modulo annotations).  During this
exploration, whenever we see a successful node that is in the~$i$-th strata, we
store it in a linked list (also marking it in the word DAG).  This is in linear
time in the number of nodes created in this length interval, hence, by
Lemma~\ref{lem:density}, in linear time in $N_i$.  By~$(\star)$, and because
each stratum is~$d$-connected, we know that we will see every node
that corresponds to a word of the~$i$-th stratum.  So the linked list stores all
nodes of the~$i$-th stratum. These form the vertices of $\Gamma_{i}$.

  \item For each node $n$ in the list corresponding to a word~$w$, we do an
exploration of all the (not necessarily simple)
undirected paths of length~$\leq d$ that start from~$n$ in the word DAG,
     where we remember the edit script corresponding to the current path (i.e.,
each time we traverse an edge in the forward direction we append its operation to
the script; each time we traverse an edge $\pushL$ or $\pushR$ in the reverse
direction we append $\popL$ and $\popR$ respectively to the script).
    We consider all marked nodes seen in this exploration, and add edges from
    the vertex for~$w$ in $\Gamma_{i}$ to these other nodes with the label
of these edges being the edit script of the corresponding path. 
By~$(\star)$, this indeed builds all the edges of~$\Gamma_i$.
This process
    takes total time proportional in $N_{i}$, because $d$ only depends on the language and the
degree of the word DAG also only depends on the language.
  \item Last, we pick our starting and ending vertices for the current
    stratum~$i$.
      Knowing the ending point $v_{e_{i-1}}$ of the $(i-1)$-th stratum, we know that
    there was some word of the current stratum at distance $\leq d$ of~$v_{e_{i-1}}$, so we explore
    from $v_{e_{i-1}}$ in the word DAG and find a witnessing node, which we pick as
    starting node $v_{s_{i}}$. For the ending node of the $i$-th stratum, we consider all nodes of
    $\stratum_\ell(L,i)$ again and explore all possible edit scripts for them
    to check if there is a node of the next stratum at distance $\leq d$; we
      know that this must succeed by Claim~\ref{clm:shorten}. We pick any one of them as
    $v_{e_{i}}$
  \item Last, we unmark all the nodes that we had previously marked in the word
    DAG, again in linear time in the size of the current stratum.
\end{itemize}

This indeed computes in linear time the stratum graph~$\Gamma_i$.  The
amortized linear time bound for this whole algorithm can then be shown using
the same reasoning as in the proof of the amortized linear time bound for
Proposition~\ref{prp:worddag2}.
\end{proof}

Finally we show our bounded-delay enumerability result
(Theorem~\ref{thm:enum-main}):
\begin{proof}[Proof of Theorem~\ref{thm:enum-main}.]
As we have already explained in Section~\ref{sec:worddag}, we simply have to combine
the two components of our enumeration algorithm: let us call~$\mathcal{A}$ the
algorithm from Proposition~\ref{prp:worddag} with~$K_A$ the value in its
  amortized linear time complexity (called $C$ in the statement), and~$\mathcal{B}$ that of
Proposition~\ref{prp:second-comp}, with~$K_B$ the value in its
  (non-amortized) linear time complexity (i.e., for each~$i$ it runs in time $K_B
  |\Gamma_i|$). Last, let~$C_A$ be the value from
Lemma~\ref{lem:density}.  

We first show that there is an algorithm~$\mathcal{C}$ that produces the
sequences of edit scripts (and stores them in a FIFO implemented as a
  double-ended linked list, to be read later) in amortized linear time, i.e., for some~$K_C$, for all~$i\geq 1$,
after $K_C \sum_{j=1}^i N_i$ computation steps, the algorithm~$\mathcal{C}$ has
produced a sequence of edit scripts corresponding to all the words of the
first~$\leq i$ strata. To do this we simply start algorithm~$\mathcal{A}$, and whenever it
  has produced $(\Gamma_{i},v_{s_{i}},v_{e_{i}})$ we pause~$\mathcal{A}$ and run
  algorithm~$\mathcal{B}$ on it, resuming $\mathcal{A}$ afterwards.
  Notice
that the size~$|\Gamma_i|$ of~$\Gamma_i$ is $\leq C_S N_i$ for some
  value~$C_S$ exponential in~$|A|$. This is because the nodes
of~$\Gamma_i$ correspond to words of the~$i$-th strata and its degree is
bounded by definition.  Then it is clear that after 
a number of computation steps
  $K_A \sum_{j=1}^i N_i + \sum_{j=1}^i K_B
C_S N_i = (K_A + K_B C_S) \sum_{j=1}^i N_i$ we have indeed computed the
sequence of edit scripts for strata~$\leq i$, so we can take~$K_C \colonequals
K_A + K_B C_S$.

  Our bounded-delay algorithm $\mathcal{C}'$ to enumerate~$L(A)$ is
then as follows:

We initialize by launching~$\mathcal{C}$ until it has produced the sequence of
edit scripts for the first stratum. We store the sequence of edit scripts in a
  FIFO. Then, we continue the execution
of~$\mathcal{C}$ to produce the rest of the sequence of edit scripts at the end
  of the
  FIFO. However, every $E \colonequals (1+C_A)K_C$ steps of~$\mathcal{C}$,
we output one edit script from the prepared sequence, i.e., from the beginning
  of the FIFO. We know that
  the edit script sequence thus produced is correct (it is the one produced
  by~$\mathcal{C}$), and that algorithm~$\mathcal{C}'$ has bounded delay
  (namely, delay $E+E'$ for $E'$ the constant time needed to read one edit
  script from the FIFO and output it), but the only point to show is that we
  never try to pop from the FIFO at a point when it is empty, i.e., we never
  ``run out'' of prepared edit scripts.

  To see why this is true, we know that $\mathcal{C}$ adds edit script sequences
  to the FIFO for each stratum, i.e., once we have concluded the computation of
  the $i$-th stratum graph~$\Gamma_i$
  and the execution of algorithm~$\mathcal{B}$ over $\Gamma_i$,
  we add precisely $N_i$ edit scripts to the FIFO.
  So it suffices to show that, for any~$i$, the FIFO does not run out after we
  have enumerated the edit scripts for the strata $1, \ldots, i$.
  Formally, we must show that for any $i \geq 1$, after we have popped
  $\sum_{j=1}^i N_j$
  edit scripts from the FIFO, then algorithm~$\mathcal{C}$ must have already added
  to the FIFO the edit scripts for the $(i+1)$-th stratum. We know that the
  time needed for algorithm~$\mathcal{C}$ to finish processing the $(i+1)$-th
  stratum is at most $K_C \sum_1^{i+1} N_i$, which by Lemma~\ref{lem:density} is
  $\leq K_C \sum_{j=1}^i N_j + C_A N_i \leq E \sum_{j=1}^i N_j$. Now, by
  the definition of algorithm~$\mathcal{C}'$, if we have popped 
  $\sum_{j=1}^i N_j$
  edit scripts from the FIFO, then we have already run at least $E \sum_{j=1}^i
  N_j$ steps of algorithm~$\mathcal{C}$. Hence, we know that algorithm~$\mathcal{C}$
  has finished producing the edit scripts
  for the $(i+1)$-th stratum and the FIFO is not empty. This
  concludes the proof.
\end{proof}

\end{toappendix}

\section{Extensions}
\label{sec:extensions}
\subparagraph*{Complexity of determining the optimal distance.}
We have shown in Result~\ref{res:main} that, given a DFA $A$, we can compute in PTIME a minimal cardinality partition of~$\L(A)$ into languages that are each $d$-orderable, for~$d=48|A|^2$.
However, we may achieve a smaller distance~$d$ if we increase the cardinality, e.g., $a^* + bbb a^*$ is $(1,3)$-partition-orderable and not~$(1,d)$-partition-orderable for~$d <3$, but is $(2,1)$-partition-orderable. This tradeoff between $t$ and~$d$ seems difficult to characterize, and in fact it is NP-hard to determine if an input DFA is $(t,d)$-partition-orderable, already for fixed $t,d$ and for finite languages.
Indeed, there is a simple reduction pointed out in~\cite{cstheory}
from the Hamiltonian path problem on grid graphs~\cite{itai1982hamilton}:

\begin{toappendix}
\subsection{Proof of the complexity results}
\end{toappendix}

\begin{propositionrep}[\cite{cstheory}]
  \label{prp:nphard}
  For any fixed $t, d \geq 1$, it is NP-complete, given a DFA~$A$ with $\L(A)$
  finite, to decide if $\L(A)$ is $(t,d)$-partition-orderable (with the push-pop or Levenshtein distance).
\end{propositionrep}

\begin{proof}
  Our proof will show NP-hardness both for the push-pop distance and for the Levenshtein distance (both problems being a priori incomparable). We denote the distance by~$\delta$. The membership in NP is immediate as a witnessing $t$-tuple of edit script sequences has polynomial size.
  
  Recall that a grid graph is a finite node-induced subgraph of the infinite grid.
  Fix the integers $d \geq 1$ and $t \geq 1$. We work on the alphabet $\Sigma = \{a, b\}$.
  We reduce from the Hamiltonian path problem on grid graphs, which is NP-hard~\cite{itai1982hamilton}.
  Given a grid graph $G$, letting $n$ be its number of vertices,
  we assume without loss of generality that~$G$ is connected, as otherwise it trivially does not have a Hamiltonian path.
  Hence, up to renormalizing, we can see each vertex of~$G$ as a pair $(i,j)$ such that two vertices $(i,j)$ and $(i',j')$ are adjacent if and only if
  $|i-i'| + |j-j'| = 1$, with $0 \leq i, j \leq n$. (Indeed, if some node is labeled $(0, 0)$ and the graph is connected, then any vertex must have values $(i,j)$ with $i+j \leq n$.)
  
  We code $G$ as a set of words of size polynomial in~$G$ defined as follows: for each vertex $(i,j)$ we have the word 
  $n_{(i,j)} \colonequals a^{di} b^{d+1} a^{dj}$.
  Let $L'$ be the language of these words.
  Then, let $L$ be $L'$ with the set of $t-1$ words $b^{(j+1)(d+1)}$
  for~$1\leq j \leq t-1$. This coding is in polynomial time, and we can obtain from~$L'$ a DFA recognizing it in polynomial time.
  
  Let us show that the reduction is correct.
  For this, let us first observe that $L$ is $(t,d)$-enumerable iff $L'$ is $(1,d)$-enumerable. Indeed, if $L'$ is $(1,d)$-enumerable then we enumerate $L$ by adding one singleton sequence for each of the $t-1$ words of $L \setminus L'$. Conversely, if $L$ is $(t,d)$-enumerable then as the words of $L'\setminus L$ are at distance $>d$ from one another on from the words of~$L'$ (as each edit can only change the number of $b$'s by one), then each of the $t-1$ words of $L' \setminus L$ must be enumerated in its own singleton sequence, and $L'$ is $(1,d)$-enumerable.
  
  Now, we 
  define~$G'$ to be the graph whose nodes are the words of~$L'$ and where we connect two words if they are different and the distance between them is~$\leq d$. Clearly $G'$ has a Hamiltonian path iff $L'$ is $d$-orderable. We claim that $G'$ is isomorphic to~$G$, which concludes the proof because then $G$ has a Hamiltonian path iff $G'$ does. So let us take two distinct vertices $(i,j)$, $(i',j')$ of~$G$ and show that they are adjacent in~$G'$ (i.e., $|i-i'| + |j-j'| = 1$) iff $a^{id} b^{d+1} a^{jd}$ is at distance $\leq d$ (for the Levenshtein or push/pop distance) to $a^{i'd} b^{d+1} a^{j'd}$. For the forward direction, it is clear that increasing/decreasing $i$ or $j$ amounts to pushing/popping $a^d$ at the beginning or end of the word. For the backward direction, proving the contrapositive, if the vertices are not adjacent then either $|i-i'| \geq 2$ or $|j-j'| \geq 2$ or both $i\neq i'$ and $j \neq j'$. In all three cases, noting that all words reachable at Levenshtein edit distance $\leq d$ must include some $b$'s in the middle, if we edit one of the words with $\leq d$ operations then the endpoints of the longest contiguous block of $b$'s cannot have moved by more than $d/2$ relative to where they were in the original word, so the only operations that can give a word of the right form amount to modifying the number of $a$'s to the left or right of the block of $b$'s, and with $d$ editions we cannot change both numbers nor can we change a number by at least $2d$.
  
  We have shown that $G$ and $G'$ are isomorphic, which establishes the correctness of the reduction and concludes the proof.
\end{proof}

\begin{toappendix}
\subsection{Proof of the results on the push-pop-right distance}
We prove our result on the push-pop-right distance in this section:
\end{toappendix}

\subparagraph*{Push-pop-right distance.}
A natural restriction of the push-pop distance 
would be to only allow editions at the right endpoint of the word, called
the \emph{push-pop-right} distance.
A $d$-ordering for this distance witnesses that 
the words of the language can be produced successively while being stored in a stack, each word being produced after at most $d$ edits.

Unlike the push-pop distance, one can show that 
some regular languages are not even partition-orderable for this distance, e.g., 
$a^*b^*$ is not~$(t,d)$-partition-orderable  with any~$t,d\in \NN$.
The enumerable regular languages for this distance in fact correspond to the well-known notion of \emph{slender languages}.
Recall that a regular language~$L$ is \emph{slender}~\cite{mpri} if there is a bound~$C \in \NN$
such that, for each $n \geq 0$, we have $|L \cap \Sigma^n| \leq C$. 
It is known~\cite{mpri} that we can test in PTIME if an input DFA represents a slender language.
Rephrasing Result~\ref{res:slender} from the introduction,
we can show that a regular language is enumerable for the push-pop-right distance
if and only if it is slender; further, if it is, then we can tractably compute the
optimal number $t$ of sequences (by counting the number of different paths to 
loops in the automaton), and we can do the enumeration with bounded delay:

\begin{theoremrep}
  \label{thm:slender}
  Given a DFA~$A$, the language $\L(A)$ is $(t,d)$-partition-orderable for the push-pop-right distance for some $t, d \in \NN$ if and only if $\L(A)$ is slender.
  Further, if $\L(A)$ is slender, we can compute in PTIME the smallest $t$ such that $\L(A)$ is $(t,d)$-partition-orderable for some~$d \in \NN$ for the push-pop-right distance.

  In addition, there is an algorithm which, given a DFA~$A$ for which $\L(A)$ is
  slender and $t=1$,
  enumerates the language $\L(A)$ with push-pop-right distance bound $2k$ and
  linear delay in~$|A|$. Further, the sequence of edit scripts produced by the
  algorithm is ultimately periodic.
\end{theoremrep}

\begin{toappendix}
  We start with the proof of the characterization: $\L(A)$ is~$(t,d)$-partition-orderable for
the push-pop-right distance iff~$\L(A)$ is slender.

  We consider the infinite tree~$T$ whose nodes are $\Sigma^*$
  and where, for every~$w\in \Sigma^*$ and~$a\in \Sigma$, the node~$w$ has
  an~$a$-labeled edge to the node~$wa$ (i.e., $wa$ is a child of~$w$). A
  \emph{$L$-infinite branch} of this tree is an infinite branch of the tree
   such that there are infinitely many nodes $n$ on that branch that
  have a descendant in~$L$.
  Formally, there is an infinite
  sequence $w = a_1 a_2 \cdots$ (i.e., an infinite word, corresponding to a
  branch) such that, for infinitely many values $i_1, i_2, \ldots$, there are
  words $x_j$ 
  such that $a_1 \cdots a_{i_j} x_j$ is a word of~$L$.

  We show that a $(t,d)$-partition-orderable regular language for the push-pop-right
  distance must contain finitely many $L$-infinite branches:

  \begin{claim}
    \label{clm:orderfinbranch}
    If $L$ is $(t,d)$-partition-orderable language for the push-pop-right
  distance and is regular, then it must contain at most~$t$ $L$-infinite
    branches.
  \end{claim}

  \begin{proof}
    We show that a language with $\geq t+1$ many $L$-infinite branches is not
    $(t,d)$-partition-orderable. Indeed, assume by contradiction that it is, for
    some distance~$d$.
    Consider a depth $m$ at which the $t+1$ $L$-infinite branches have
    diverged, i.e., we have distinct nodes $n_1, \ldots, n_{t+1}$ at depth~$m$
    that all have infinitely many descendants in~$L$.
    Consider a moment at which all words of the language of length
    $\leq m+d$ have been enumerated. Then by the pigeonhole principle there must
    be some language in the partition that still has infinitely many words to
    enumerate from two different branches, say descendants of~$n_i$ and $n_j$
    with $i \neq j$. Now, all descendants of $n_i$ at depth $> d$ are at
    distance $>d$ from all descendants of~$n_j$ at depth~$>d$, which contradicts
    the assumption that the language in the partition can move from one to the
    other.
  \end{proof}

  And we show that a regular language having finitely many $L$-infinite branches must be
  slender.

  \begin{claim}
    \label{clm:finbranchslender}
    If a regular language~$L$ has finitely many $L$-infinite branches, then it
    is slender.
  \end{claim}

  This follows from an ancillary claim shown by pumping:

  \begin{claim}
    \label{clm:branch}
    For an infinite regular language $L$, there is a value $d \in \NN$ such that, for
    each word $w$, there is an $L$-infinite branch
    in~$T$ such that $w$ is at depth at most~$d$ from a node of the branch.
  \end{claim}

  \begin{proof}
    Let $d$ be the number of states of a DFA recognizing~$L$. 
    As $L$ is infinite, it clearly has at least some $L$-infinite branch
    obtained by considering $r s^* t$ for $s$ a simple loop on a loopable state.
    Hence, the claim is trivial if $w$ has length $<d$ because the root is a
    node of this $L$-infinite branch.

    If $|w| \geq d$, 
    we know that there is some loopable state of~$w$ that occurs in the
    suffix of length $d$, i.e., we can write $w = r t$ where the state $q$
    between $r$ and~$t$ is loopable. Now, let $s$ be a loop on~$q$ of length at
    most~$d$, and consider the word sequence~$w_i = r s^i t$ of~$L$ starting at
    $w_1 = w$. All words of~$w_i$ are in~$L$. Further, let $w_i'$ be the
    sequence where $w_1'$ is the least common ancestor (LCA) of~$w_1$ and~$w_2$,
    $w_2'$ is the LCA of~$w_2$ and~$w_3$, and so on: clearly $w_i' = r s^i$ for
    all~$i$ (but they are not words of~$L$), in particular each $w_i'$ is an
    ancestor of~$w_{i+1}'$, so the infinite sequence of the $w_i'$ defines an
    infinite branch in~$T$. And is an $L$-infinite branch, because each $w_i'$
    has a descendant.
    Thus indeed $w$ is at depth~$d$ from a node of this infinite branch.
  \end{proof}

  We can then prove Claim~\ref{clm:finbranchslender}:
  \begin{proof}[Proof of Claim~\ref{clm:finbranchslender}]
    If $L$ is finite then it is slender. Otherwise, letting $d$ be the number of
    states of a DFA recognizing~$L$, we know by Claim~\ref{clm:branch} that all words of the language
    must be at depth $\leq d$ from a node of some $L$-infinite branch, hence of
    one of the finite collection of $L$-infinite branches.
    This directly implies that $L$ is slender, because for
    any length $n \geq d$, considering $L \cap \Sigma^n$, the number of nodes of~$T$ at
    depth~$n$ that can be in~$L$ are the descendants of the nodes of the
    branches at depth between $n$ and $n-d$, i.e., some number that only depends on the language.
  \end{proof}

  Thanks to Claims~\ref{clm:orderfinbranch} and~\ref{clm:finbranchslender}, we
  know that, among the regular languages, only the slender languages can be $(t,d)$-partition-orderable language for the push-pop-right
  distance. Indeed, if a language is $(t,d)$-partition-orderable then it has
  finitely many $L$-infinite branches by the first claim, which implies that it
  is slender by the second claim. So in the sequel it suffices to focus on
  slender languages.

  We will refine a known characterization of slender languages. We know 
  from~\cite[Chapter XII, Theorem 4.23]{mpri} the following characterization:

  \begin{theorem}[\cite{mpri}]
    \label{thm:slendercarac}
    The following are equivalent on a regular language~$L$:
    \begin{itemize}
      \item $L$ is slender
      \item $L$ can be written as a union of regular expressions of the form $x y^* z$.
      \item The minimal DFA for~$L$ does not have a connected pair of simple
        cycles.
    \end{itemize}
  \end{theorem}
  This implies in particular that we can check in PTIME whether a language is
  slender given a DFA~$A$ for the language, by computing in PTIME the minimal
  DFA equivalent to~$A$, and then checking if there are two connected simple
  cycles. 

  For $t \in \NN$, a \emph{$t$-slender} language is one that can be written as a
  disjoint union of a finite language $L'$ and of $t$ languages $r_i s_i^* L_i$ with
  $r_i$ and $s_i$ words and $L_i$ a finite language, for $1 \leq i \leq t$, and
  all $r_i$ are pairwise incomparable (i.e., they are pairwise distinct and none
  is a strict prefix of another)
  and no word of~$r_i$ is a prefix of a word of~$L'$.
  Consider a trimmed DFA~$A$, its \emph{non-loopable prefixes}
  are the words $w$ such that reading $w$ in~$A$ brings us to a loopable state, and
  reading any strict prefix of~$w$ does not. We claim the following:

  \begin{proposition}
    \label{prp:dfa}
    The following are equivalent on a regular language $L$:
    \begin{itemize}
      \item $L$ is recognized by a DFA without a connected pair of simple
        cycles and with exactly $t$ non-loopable prefixes.
      \item $L$ can be written as a $t$-slender language.
    \end{itemize}
  \end{proposition}

  This proposition implies in particular that a slender language is necessarily
  $t$-slender for some~$t$, by considering its minimal DFA and counting the
  minimal infinitely continuable prefixes. Note that, given a DFA, we can
  compute in PTIME the equivalent minimal DFA and count in PTIME the
  non-loopable prefixes. We prove the proposition:

  \begin{proof}[Proof of Proposition~\ref{prp:dfa}]
    If $L$ is recognized by an automaton of the prescribed form, we can write
    $L$ as a disjoint union of the non-loopable words (a finite language) and the
    languages $r_i L_{q_i}$ where the $r_i$ are the non-loopable prefixes and
    the $L_{q_i}$ is the language accepted by starting at the loopable
    state~$q_i$ at which we get after reading the non-loopable prefix. (Some of
    the $q_i$ may be identical.)
    Note that by construction the $r_i$ are pairwise incomparable and none is a
    prefix of a non-loopable word.
    Now, each $L_{q_i}$ can be decomposed between
    the words accepted without going to~$q_i$ again, and those where we do. As
    $L$ has no connected pair of simple cycles, if we do not go to~$q_i$ again,
    then we cannot complete the simple cycle on~$q_i$ and we cannot go to
    another cycle, so the possible words leading to a finite state form a finite
    language $L_i$. If we do, then the word starts with the (unique) label of
    the (unique) simple cycle starting at~$q_i$, i.e., $s_i$, and then we have a
    word of~$L_{q_i}$. Thus, we can write $L$ as a disjoint union of the
    non-loopable words and of the $r_i (s_i)^* L_i$ with $L_i$ finite.

    Conversely, if $L$ is written as a $t$-slender language then we obtain the
    automaton in the obvious way: start with an acyclic DFA $A$ for the finite
    language, construct DFAs $A_i$ for each $s_i^* L_i$ which have exactly one simple
    cycle on which the initial state is located, then for each $r_i$ extend~$A$
    with a path labeled by~$r_i$ going from the initial state to the initial
    state of~$A_i$: some states of the path may already exist (because of the
    words of~$L'$ or of the other words of~$r_i$), but the condition on the $r_i$
    and on~$L'$ guarantee that we always create at least one new transition.
    The resulting automaton accepts by construction the words of~$L'$ and the
    words of the $r_i s_i^* L_i$, and for any accepting path in the automaton
    either it ends at a state of~$A$ and the accepted word is a word of~$L'$, or
    it goes into an $A_i$ and the accepted word is a word of some~$r_i s_i^*
    L_i$
  \end{proof}

  We now claim that the $t$-slender languages are precisely those that are
  $(t,d)$-partition-orderable for some~$d$:

  \begin{proposition}
    A $t$-slender language is
    $(t,d)$-partition-orderable for some~$d$. Conversely, if a regular language
    is $(t,d)$-partition-orderable then it is $t$-slender.
  \end{proposition}

  \begin{proof}
    For the first claim, it suffices to show that a $1$-slender language is
    $d$-orderable for the push-pop-right distance for some~$d$. This is easy:
    first enumerate the words of~$L'$ in some naive way, and then enumerate the
    words of $rL''$, then $rsL''$, etc. The distance within each sequence is
    bounded because $L'$ and $L''$ are finite, and the distance when going from
    the last word of $rs^iL''$ to the first word of $r s^{i+1} L''$ is bounded
    too.

    For the second claim, we know by Claim~\ref{clm:orderfinbranch} that a
    regular language $L$ that is $(t,d)$-partition-orderable for some~$d$ must have
    at most $t$ $L$-infinite branches. Now, we know by
    Claim~\ref{clm:finbranchslender} that $L$ must then be slender.
    Assuming by way of contradiction that $L$ is not $t$-slender, by
    Theorem~\ref{thm:slendercarac} together with Proposition~\ref{prp:dfa} we
    know $L$ must be $t'$-slender for some $t'$, implying that $L$ is
    $t'$-slender for some $t' > t$. Now, being $t'$-slender implies that $L$ has
    $t'$ infinite $L$-branches, namely, those starting at the $r_i$ which are
    pairwise incomparable. This contradicts our assumption that $L$ has at most
    $t$ $L$-infinite branches, and concludes the proof.
  \end{proof}

  We have thus shown the first part of Theorem~\ref{thm:slender}: if the
  language $\L(A)$ is $(t,d)$-partition-orderable for the push-pop-right
  distance for some $t, d \in \NN$ then it is $t$-slender by
  the proposition, and conversely if it is slender then it is $t$-slender for
  some $t$ by Theorem~\ref{thm:slendercarac} and Proposition~\ref{prp:dfa} and
  is then $(t,d)$-partition-orderable for some~$d \in \NN$ by the proposition.
  Further, if $L$ is slender, using the characterization of
  Proposition~\ref{prp:dfa} we can compute in PTIME the smallest $t$ such that $L$ is
  $t$-slender, and then we know that $L$ is $(t,d)$-partition-orderable for some
  $d$ but not $(t-1,d)$-partition-orderable, thanks to the proposition.

  We last show that, if $t = 1$, we can compute the description of an ultimately
  periodic sequence of edit scripts that enumerates $\L(A)$ respecting a linear
  bound on the push-pop-right distance and in linear delay:
  \begin{proposition}
    \label{prp:slenderub}
    There is an algorithm which,
    given a DFA~$A$ with $k$ states representing a $1$-slender language,
    computes a sequence of edit scripts that enumerates $\L(A)$ and is
    ultimately periodic, with push-pop-right distance bound~$2k$ and delay
    linear in~$|A|$.
  \end{proposition}

  Note that this proposition admits a converse: the ultimately periodic
  sequences of edit scripts can only achieve slender languages (see
  Proposition~\ref{prp:slenderlb}).

  We now prove the proposition:

  \begin{proof}[Proof of Proposition~\ref{prp:slenderub}]

    Recall that we are given the input DFA in the way described in
    Appendix~\ref{apx:machine}.  We first locate in the input DFA the unique
    simple cycle and path from the initial state to that simple cycle: this can
    be performed in linear time by running a depth-first search (DFS) on the
    automaton from the initial state, marking states when the DFS starts
    visiting them and is done visiting them. We stop at the moment where the
    exploration reaches a vertex which is currently being visited: then the
    recursion stack gives us (from bottom to top) the unique path from the
    initial state to a state of the loop, and then the loop itself.

    We now start the computation of the ultimately periodic sequence by
    producing scripts to enumerate all words of the finite language
    $L'$ of the non-loopable words. For this, we perform a DFS from the initial
    state where we do not mark
    vertices as visited (so as to enumerate all paths).
    We do so on a slightly modified version of the automaton where we remove the
    states of the loop (as we must produce only non-loopable words), and where
    we trim the automaton to remove all states which no longer have a path to an
    accepting state.
    Each time we reach a final state in the DFS, we produce an edit
    sequence achieving the word described by the recursion stack from the
    previously produced word; this is easy to do by stacking push and pop
    operations which correspond to what the DFS does. Remember that the
    automaton is trimmed, so we know that 
    we achieve a new non-loopable word every $k$ steps at most, for~$k$ the
    number of automaton states.

    Now, we continue the enumeration sequence with an edit script that goes from
    the last produced word to the word corresponding to the unique path from an
    initial state to the first state $q$ of the loop. We are now ready to start the
    computation of the periodic part of the ultimately periodic sequence.

    To do this, we first enumerate all words that can be accepted from~$q$
    without going back to~$q$. This can be done by a DFS in the same way that we
    did previously, with the same distance and complexity bounds. Next, from the
    last word enumerated in this fashion, we pop until we arrive to a state of
    the loop, then we complete the loop to reach the state~$q$. This completes
    the description of the periodic part.

    Note that, in all cases, the distance between two words, hence the delay in
    producing the corresponding elements of the sequence, is at most~$2k$. This
    is because, when going from one word to the next, the algorithm is always
    popping a simple path of states, and pushing a simple path, so each state
    can at most occur twice. The only exception is the previous paragraph, where
    we pop until we arrive to a state of the loop, then push a completion of the
    loop: the next produced word may then be pushing states that already occur
    in the completion of the loop, but these states occur at most two times in
    total and did not occur in the pop, so the overall bound still applies.
  \end{proof}
\end{toappendix}

Of course, our results for the push-pop-right distance extend to the
push-pop-left distance up to reversing the language, except for the complexity
results because the reversal of the input DFA is generally no longer
deterministic.

\section{Conclusion and future work}
\label{sec:conc}
We have introduced the problem of ordering 
languages as sequences while bounding the maximal distance between successive
words, and of enumerating these sequences with small edit scripts to achieve bounded delay.
Our main result is a PTIME characterization of the regular languages that can be
ordered in this sense for the push-pop distance (or equivalently the Levenshtein
distance), for any specific number of sequences; and a bounded-delay
enumeration algorithm for the orderable regular languages. Our characterization
uses the number of classes of interchangeable states of a DFA~$A$ for the
language, which, as our results imply, is an intrinsic parameter of~$\L(A)$,
shared by all (trimmed) DFAs recognizing the same language. We do not know if
this parameter can be of independent interest.

Our work opens several questions for future research. The questions of
orderability and enumerability can be studied for more general languages (e.g.,
context-free languages), other distances (in particular substitutions plus
push-right operations, corresponding to the Hamming distance on a right-infinite
tape), or other enumeration models (e.g., reusing factors of previous words). We
also do not know the computational complexity, e.g., of optimizing the distance
while allowing any finite number of threads, in particular for slender
languages. Another complexity question is to understand if the 
bounded delay of our enumeration algorithm could be made polynomial in the input
DFA rather than exponential, or what delay can be achieved if the input
automaton is nondeterministic.

\bibliography{main.bib}
\end{document}